\documentclass[11pt]{article}

\usepackage{amsmath}
\usepackage{amssymb}
\usepackage{latexsym}
\usepackage{graphicx}
\usepackage{epsfig}
\usepackage{stmaryrd}

\newtheorem{theorem}{Theorem}[section]
\newtheorem{proposition}[theorem]{Proposition}
\newtheorem{definition}[theorem]{Definition}

\newtheorem{lemma}[theorem]{Lemma}

\newtheorem{example}[theorem]{Example}

\newtheorem{corollary}[theorem]{Corollary}

\def\pushright#1{{\parfillskip=0pt\widowpenalty=10000
\displaywidowpenalty=10000\finalhyphendemerits=0\leavevmode\unskip
\nobreak\hfil\penalty50\hskip.2em\null\hfill{#1}\par}}

\def\qed{\xqed\global\SuppressEndOfProoftrue}
\newif\ifSuppressEndOfProof\SuppressEndOfProoffalse
\def\xqed{\pushright\markendofproof}
\def\markendofproof{\rule{1.3217ex}{2ex}}

\newcommand{\lsem}{\mbox{$\lbrack\!\lbrack$}}
\newcommand{\rsem}{\mbox{$\rbrack\!\rbrack$}}

\begin{document}

\title{Extensional Higher-Order Logic Programming\footnote{An early
version of this paper appears in: A. Charalambidis, K. Handjopoulos,
P. Rondogiannis, W. W. Wadge. Extensional Higher-Order Logic
Programming. Proceedings of the {\em 12th European Conference on
Logics in Artificial Intelligence} (JELIA). LNCS 6341, Springer,
pages 91-103, 2010.}}

\author{A. Charalambidis\footnote{The research of A. Charalambidis has been co-financed by the European Union (European Social Fund - ESF)
and Greek national funds through the Operational Program ``Education
and Lifelong Learning'' of the National Strategic Reference
Framework (NSRF) - Research Funding Program: Heracleitus II.
Investing in knowledge society through the European Social Fund. },
K. Handjopoulos, P. Rondogiannis\\
Department of Informatics \& Telecommunications\\
University of Athens\\
Panepistimiopolis, 157 84 Athens, Greece\\
\\
\and
William W. Wadge\\
Department of Computer Science\\
University of Victoria\\
Victoria, BC, Canada V8W 3P6}
%
%
\maketitle
\thispagestyle{empty}
\begin{abstract}
We propose a purely extensional semantics for higher-order logic
programming. In this semantics program predicates denote sets of
ordered tuples, and two predicates are equal iff they are equal as
sets. Moreover, every program has a unique minimum Herbrand model
which is the greatest lower bound of all Herbrand models of the
program and the least fixed-point of an immediate consequence
operator. We also propose an SLD-resolution proof procedure which is
proven sound and complete with respect to the minimum model
semantics. In other words, we provide a purely extensional
theoretical framework for higher-order logic programming which
generalizes the familiar theory of classical (first-order) logic
programming.
%
%
\end{abstract}

\pagestyle{plain}

\date{}

\section{Introduction}\label{intro-section}
The two most prominent declarative paradigms, namely logic and
functional programming, differ radically in an important aspect:
logic programming is traditionally first-order while functional
programming encourages and promotes the use of higher-order
functions and constructs. One problem is that even second-order
logic fails in terms of vital properties such as completeness and
compactness. It would seem, on the face of it, that there would be
no hope of finding a complete resolution proof procedure for
higher-order logic programming.

The initial attitude of logic programmers towards higher-order logic
programming was somewhat skeptical: it was often argued (see for
example~\cite{W82-441}) that there exist ways of encoding or
simulating higher-order programming inside Prolog itself. However
ease of use is a primary criterion for a programming language, and
the fact that higher-order features can be simulated or encoded does
not mean that it is practical to do so.

Eventually extensions with genuine higher-order capabilities were
introduced  - roughly speaking, extensions which allow predicates to
be applied but also passed as parameters. The two most prominent
such languages are $\lambda$Prolog~\cite{MN86,Nad87} and
HiLog~\cite{CKW89-37,CKW93-187}. These two systems share a common
idea, namely they are both {\em intensional}. Intuitively speaking,
an intensional language places almost no restraints on the way in
which a predicate can be passed and used. In an intensional language
it is possible that two co-extensional predicates are not considered
equal. In other words, a predicate in such a language is more than
just the set of arguments for which it is true.

However, for many applications intensionality appears to be
appropriate. Suppose, for example, we have a database of
professions, both of their membership and their status. We might
have rules such as:
\[
\begin{array}{l}
\mbox{\tt engineer(tom).}\\
\mbox{\tt engineer(sally).}\\
\mbox{\tt programmer(harry).}
\end{array}
\]
with {\tt engineer} and {\tt programmer} used as predicates. But in
intensional higher-order logic programming we could also have rules
in which these are arguments, eg:
\[
\begin{array}{l}
\mbox{\tt profession(engineer).}\\
\mbox{\tt profession(programmer).}
\end{array}
\]
Now suppose {\tt tom} and {\tt sally} are also avid users of
Twitter. We could have rules:
\[
\begin{array}{l}
\mbox{\tt tweeter(tom).}\\
\mbox{\tt tweeter(sally).}
\end{array}
\]
In the absence of other rules, it is clear that the tweeters are
exactly the engineers; but the query:
\[
\begin{array}{l}
\mbox{\tt ?-profession(tweeter).}
\end{array}
\]
fails. This failure contradicts the extensionality principle, which
holds that predicates that succeed for exactly the same instances
are equal. However, in this context the failure of extensionality
does not seem unnatural.

Nevertheless, the failure of extensionality cannot in general be
taken lightly. It means that we cannot use our mathematicians'
intuitions of relations, intuitions based on hundreds of years of
mathematical development. It raises doubts that rules like those
just given can have a simple declarative meaning. Moreover, there
are many applications that call for higher-order logic (predicates
used as arguments) but do not involve intensional notions. As a
simple example, consider the predicate {\tt allmembers(L,P)} which
asserts that all elements of the list {\tt L} have property {\tt P}.
Predicate {\tt allmembers} raises no foundational issues and the
corresponding rules seem obvious:
\[
\begin{array}{l}
\mbox{\tt allmembers([],P).} \\
\mbox{\tt allmembers([H|T],P):-P(H),allmembers(T,P).}
\end{array}
\]
However this is not legitimate Prolog, and to write these rules we
currently have no choice but to use an intensional higher-order
language, even though the logic behind {\tt allmembers} is entirely
extensional. For example, if the query:
\[
\begin{array}{l}
\mbox{\tt ?-allmembers([a,b,c],p).}
\end{array}
\]
succeeds and {\tt q} is co-extensional with {\tt p}, we can be sure
that the query
\[
\begin{array}{l}
\mbox{\tt ?-allmembers([a,b,c],q).}
\end{array}
\]
will also succeed.

Are there more modest higher-order extensions of logic programming
that do not entail intensionality? After all, higher-order
extensions of functional programming are almost all extensional.
This question was first raised by W. W. Wadge in~\cite{Wa91a} and
answered in the affirmative. Wadge discovered a simple syntactic
restriction which, though it limited the applicability of the
language, ensured that compliant programs have an extensional
declarative reading. The restriction forbids a user-defined
predicate to appear as an argument in the head of a clause. For
example, of the rules cited already,
\[
\begin{array}{l}
\mbox{\tt profession(engineer).}
\end{array}
\]
violates the restriction, but
\[
\begin{array}{l}
\mbox{\tt allmembers([],P).}
\end{array}
\]
complies with it. Roughly speaking, the restriction says that rules
about predicates can state general principles but cannot pick out a
particular predicate for special treatment. Wadge gave several
examples of useful extensional higher-order programs and outlined
the proof of a minimum-model result. He also showed that in this
model the denotations of program predicates are monotonic and
continuous. Continuity, in this context, is a kind of finitaryness.
For example, if {\tt foo(p)} succeeds and {\tt foo} is continuous,
it means there is a finite set of arguments $\{{\tt a}_1,\ldots,{\tt
a}_k\}$ for which $\mbox{\tt p(a}_1\mbox{\tt )},\ldots,\mbox{\tt
p(a}_k\mbox{\tt )}$ all succeed; moreover, if $\mbox{\tt
q(a}_1\mbox{\tt )},\ldots,\mbox{\tt q(a}_k\mbox{\tt )}$ also
succeed, then {\tt foo(q)} succeeds.

\noindent{\bf Contributions:} In this paper we extend the study
initiated in~\cite{Wa91a} and derive the first, to our knowledge,
complete theoretical framework for extensional higher-order logic
programming, both from a semantic as-well-as from a proof theoretic
point of view.

Our first contribution is the development of a novel extensional
semantics for higher-order  logic programming that is based on {\em
algebraic lattices} (see for example~\cite{Gra78}), a subclass of
the familiar complete lattices that have traditionally been used in
the theory of first-order logic programming. For every predicate
type of our language, algebraic lattices single out a subset of
``finite'' objects of that type. In other words, the proposed
semantics reflects in a direct way the finitary nature of continuity
that is implicit in~\cite{Wa91a}. The benefit of the new approach
compared to that of~\cite{Wa91a} is that all basic properties and
results of classical logic programming are now transferred in the
higher-order setting in a natural way. Moreover, the new semantics
leads to a relatively simple sound and complete proof procedure (see
below) even for a language that is genuinely more powerful than the
one considered in~\cite{Wa91a}. More details on the connections
between the two approaches will be given
in~\ref{comparison-with-Wadge}.

Our second contribution fixes a major shortcoming of Wadge's
language by allowing clause bodies and program goals to have
uninstantiated higher-order variables. To understand the importance
of this extension, consider the following axiom for {\tt bands}
(musical ensembles):
\[
\begin{array}{l}
\mbox{\tt
band(B):-singer(S),B(S),drummer(D),B(D),guitarist(G),B(G).}
\end{array}
\]
This says that a band is a group that has at least a singer, a
drummer, and a guitarist. Suppose that we also have a database of
musicians:
\[
\begin{array}{l}
\mbox{\tt singer(sally).}\\
\mbox{\tt singer(steve).}\\
\mbox{\tt drummer(dave).}\\
\mbox{\tt guitarist(george).}\\
\mbox{\tt guitarist(grace).}
\end{array}
\]
Our extensional higher-order language allows the query {\tt
?-band(B)}. At first sight a query  like this is impractical if not
impossible to implement. Since a band is a set, bands can be very
large and there can be many, possibly uncountably infinitely many of
them. In existing intensional systems such queries fail since the
program does not provide any information about any particular band.

However, in an extensional context the finitary behavior of the
predicates of our language, saves us. If the predicate {\tt band}
declares a relation to be a band, then (due to the finitaryness
described above) it must have examined only finitely many members of
the relation. Therefore we can enumerate the bands by enumerating
{\em finite} bands, and this collection is countable (in this
particular example it is actually finite). Actually, as we are going
to see, this enumeration can be performed in a careful way so as
that it avoids producing all finite relations one by one (see the
discussion in Section~\ref{intuitive-overview} that follows).

Our final contribution, and not the least, is a relatively simple
proof procedure for extensional higher-order logic programming,
which extends classical SLD-resolution. We demonstrate that the new
proof procedure is sound and complete with respect to the proposed
semantics. In particular, the derived completeness theorems
generalize the well-known such theorems for first-order logic
programming. This result may, at first sight, also seem paradoxical,
given the well-known failure of completeness for even second-order
logic. But the paradox is resolved by recalling that we are dealing
with a restricted subset of higher-order logic and that the
denotations of the types of our language are not arbitrary sets but
instead algebraic lattices (which have a much more refined
structure).

One very important benefit of the proof procedure is that it gives
us an operational semantics for our language. This means in turn
that we could probably extend it with cut, negation and other
operational features not easily specified in terms of model theory
alone.

The rest of the paper is organized as follows:
Section~\ref{intuitive-overview} presents in a more detailed manner
the basic ideas developed in this paper.  Section~\ref{syntax-of-H}
introduces the syntax of the higher-order logic programming language
${\cal H}$. Section~\ref{alglat-section} introduces the key
lattice-theoretic notions that will be needed in the development of
the semantics. Sections~\ref{semantics-of-H} and~\ref{min-Herbrand}
develop the semantics and the minimum Herbrand model semantics of
${\cal H}$; the main properties of the semantics are also
established. Section~\ref{proof-procedure-section} introduces an
SLD-resolution proof procedure for ${\cal H}$ and establishes its
soundness and completeness. Section~\ref{rw-section} presents a
brief description of related approaches to higher-order logic
programming. Section~\ref{future-work} briefly discusses
implementation issues and presents certain interesting topics for
future work. The lengthiest among the proofs have been moved to
corresponding appendices in order to enhance the readability of the
paper.

\section{The Proposed Approach: an Intuitive Overview}\label{intuitive-overview}
As discussed in the previous section, the purpose of this paper is
to develop a purely extensional theoretical framework for
higher-order logic programming which will generalize the familiar
theory of first-order logic programming. The first problem we
consider is to bypass one important restriction of~\cite{Wa91a},
namely the inability to handle programs in which clauses contain
uninstantiated predicate variables. The following example
illustrates these ideas:
\begin{example}\label{the-ultimate-example}
Consider the following higher-order program written in an extended
Prolog-like syntax:
\[
\begin{array}{l}
\mbox{\tt p(Q):-Q(0),Q(s(0)).} \\
\mbox{\tt nat(0).}\\
\mbox{\tt nat(s(X)):-nat(X).}
\end{array}
\]
The Herbrand universe of the program is the set of natural numbers
in successor notation. According to the semantics of~\cite{Wa91a},
the least Herbrand model of the program assigns to predicate {\tt p}
a continuous relation which is true of {\em all} unary relations
that contain at least {\tt 0} and {\tt s(0)}. Consider now the
query:
\[
\begin{array}{l}
\mbox{\tt ?-p(R).}
\end{array}
\]
which asks for all relations that satisfy {\tt p}. Such a query
seems completely unreasonable, since there exist uncountably many
relations that must be substituted and tested in the place of {\tt
R}.\qed
\end{example}

The above example illustrates why uninstantiated predicate variables
in clauses were disallowed in~\cite{Wa91a}. From a theoretical point
of view, one could extend the semantics to cover such cases, but the
problem is mainly a practical one: ``how can one implement such
programs and queries?''.

In more formal terms, the least Herbrand model of a higher-order
program under the semantics of~\cite{Wa91a} is in general an {\em
uncountable set}; in our example, this is evidenced by the fact that
there exists an uncountable number of unary relations over the
natural numbers that contain both {\tt 0} and {\tt s(0)}. This
observation comes in contrast with the semantics of first-order
logic programming in which the least Herbrand model of a program is
a countable set. How can one define a proof procedure that is sound
and complete with respect to this semantics? The key idea for
bypassing these problems was actually anticipated in the concluding
section of~\cite{Wa91a}:
\begin{quote}
{\em Our higher order predicates, however, are continuous: if a
relation satisfies a predicate, then some finite subset satisfies
it. This means that we have to examine only finite relations.}
\end{quote}
In the above example, despite the fact that there exists an infinite
number of relations that satisfy {\tt p}, all these relations are
supersets of the finite relation $\{\mbox{\tt 0},\mbox{\tt s(0)}\}$.
In some sense, this finite relation {\em represents} all the
relations that satisfy {\tt p}. But how can we make the notion of
``finiteness'' more explicit? In order to define a sound and
complete proof procedure for an interesting extensional higher-order
logic programming language, our semantics must in some sense reflect
the above ``finitary'' concepts more explicitly.

An idea that springs to mind is to define an alternative semantics
in which variables (like {\tt Q} in
Example~\ref{the-ultimate-example}) range over finite relations (and
not over arbitrary relations as in~\cite{Wa91a}). Of course, the
notion of ``finite'' should be appropriately defined for every
predicate type. But then an immediate difficulty appears to arise.
Given again the program in Example~\ref{the-ultimate-example}, and
the query
\[
\begin{array}{l}
\mbox{\tt ?-p(nat).}
\end{array}
\]
it is not immediately obvious what the meaning of the above is.
Since we have assumed that {\tt Q} ranges over finite relations, how
can {\tt p} be applied to an infinite one? To overcome this problem,
observe that in order for the predicate {\tt p} to succeed for its
argument {\tt Q}, it only has to examine a ``finite number of
facts'' about {\tt Q} (namely whether {\tt Q} is true of {\tt 0} and
{\tt s(0)}). This remark suggests that the meaning of {\tt p(nat)}
can be established following a non-standard interpretation of
application: we apply the meaning of ${\tt p}$ to all the ``finite
approximations'' of the meaning of {\tt nat}, ie., to all finite
subsets of the set $\{\mbox{\tt 0},\mbox{\tt s(0)},\mbox{\tt
s(s(0))},\ldots\}$. In our case {\tt p(nat)} will be true since
there exists a finite subset of the meaning of ${\tt nat}$ for which
the meaning of ${\tt p}$ is true (namely the set $\{\mbox{\tt
0},\mbox{\tt s(0)}\}$).

Notice that the new semantical approach outlined above, heavily
relies on the idea that the meaning of predicates (like {\tt nat})
can be expressed as the least upper bound of a set of simpler (in
this case, finite) relations. Actually, this is an old and
well-known assumption in the area of denotational semantics, as the
following excerpt from~\cite{Sto77}[page 98] indicates:
\begin{quote}
{\em So we may reasonably demand of all the value spaces in which we
hope to compute that they come equipped with a particular countable
subset of elements from which all the other elements may be built
up.}
\end{quote}
As we are going to demonstrate, the meaning of {\em every} predicate
defined in our language possesses the property just mentioned, and
this allows us to use the new non-standard semantics of application.
In fact, as we are going to see, for every predicate type of our
language, the set of possible meanings of this type forms an {\em
algebraic lattice}~\cite{Gra78}; then, the above property is nothing
more than the key property which characterizes algebraic lattices
(see Proposition~\ref{characterization-of-al}), namely that ``{\em
every element of an algebraic lattice is the least upper bound of
the compact elements of the lattice that are below it}''. More
importantly, for the algebraic lattices we consider, it is
relatively easy to identify these compact elements and to enumerate
them one by one. Based on the above semantics, we are able to derive
for higher-order logic programs many properties that are either
identical or generalize the familiar ones from first-order logic
programming (see Section~\ref{min-Herbrand}).

The new semantics allows us to introduce a relatively simple, sound
and complete proof procedure which applies to programs and queries
that may contain uninstantiated predicate variables. This is due to
the fact that the set of ``finite'' relations is now countable, and
as we are going to see, there exist interesting ways of producing
and enumerating them. The key idea can be demonstrated by continuing
Example~\ref{the-ultimate-example}. Given the query:
\[
\begin{array}{l}
\mbox{\tt ?-p(R).}
\end{array}
\]
one (inefficient and tedious) approach would be to enumerate all
possible finite relations of the appropriate type over the Herbrand
universe. Instead of this, we use an approach which is based on what
we call {\em basic templates}:  a basic template for {\tt R} is
(intuitively) a finite set whose elements are individual variables.
This saves us from having to enumerate all finite sets consisting of
ground terms from the Herbrand universe. For example\footnote{The
notation we use for representing basic templates will slightly
change in Section~\ref{proof-procedure-section}.}, assume that we
instantiate {\tt R} with the template $\{\mbox{\tt X},\mbox{\tt
Y}\}$. Then, the resolution proceeds as follows:
\[
\begin{array}{l}
 \mbox{\tt ?-p(R)}\\
 \mbox{\tt ?-p(\{X},\mbox{\tt Y\})}\\
 \mbox{\tt ?-\{X},\mbox{\tt Y\}(0)},\mbox{\tt \{X},\mbox{\tt Y\}(s(0))}\\
 \mbox{\tt ?-\{0},\mbox{\tt Y\}(s(0))}\\
 \Box
\end{array}
\]
and the proof procedure will return the answer ${\tt R} = \{{\tt
0},\mbox{\tt s(0)}\}$. The proof procedure will also return other
finite solutions, such as ${\tt R} = \{{\tt 0},\mbox{\tt s(0)},
\mbox{\tt Z}_1\}$, ${\tt R} = \{{\tt 0},\mbox{\tt s(0)}, \mbox{\tt
Z}_1, \mbox{\tt Z}_2\}$, and so on. However, a slightly optimized
implementation (see Section~\ref{future-work}) can be created that
returns only the answer ${\tt R} = \{{\tt 0},\mbox{\tt s(0)}\}$,
which represents all the finite relations produced by the proof
procedure. The intuition behind the above answer is that the given
query succeeds for all unary relations that contain at least {\tt 0}
and {\tt s(0)}. Similarly, for the {\tt band} example of
Section~\ref{intro-section}, the implementation will systematically
assemble all the minimal three-member bands from the talents
available.

\section{The Higher-Order Language ${\cal H}$: Syntax}\label{syntax-of-H}
In this section we introduce the higher-order logic programming
language ${\cal H}$, which extends classical first-order logic
programming to a higher-order setting. The language ${\cal H}$ is
based on a simple type system that supports two base types: $o$, the
boolean domain, and $\iota$, the domain of individuals (data
objects). The composite types are partitioned into three classes:
functional (assigned to function symbols), predicate (assigned to
predicate symbols) and argument (assigned to parameters of
predicates).
\begin{definition}
A type can either be \emph{functional}, \emph{argument}, or
\emph{predicate}:
\[
\begin{array}{lll}
\sigma & := & \iota \mid (\iota \rightarrow  \sigma)\\
\rho   & := & \iota \mid \pi\\
\pi    & := & o \mid (\rho \rightarrow \pi)
\end{array}
\]
We will use $\tau$ to denote an arbitrary type (either functional,
argument or predicate one).
\end{definition}

As usual, the binary operator $\rightarrow$ is right-associative. A
functional type that is different than $\iota$ will often be written
in the form $\iota^n \rightarrow \iota$, $n\geq 1$. Moreover, it can
be easily seen that every predicate type $\pi$ can be written in the
form $\rho_1 \rightarrow \cdots \rightarrow \rho_n \rightarrow o$,
$n\geq 0$ (for $n=0$ we assume that $\pi=o$).

We can now proceed to the definition of ${\cal H}$, starting from
its alphabet:
\begin{definition}
The \emph{alphabet} of the higher-order language ${\cal H}$ consists
of the following:
\begin{enumerate}
\item {\em Predicate variables} of every predicate type $\pi$
      (denoted by capital letters such as
      $\mathsf{P,Q,R,\ldots}$).

\item {\em Predicate constants} of every predicate type $\pi$
      (denoted by lowercase letters such as
      $\mathsf{p,q,r,\ldots}$).

\item {\em Individual variables} of type $\iota$
      (denoted by capital letters such as
      $\mathsf{X,Y,Z,\ldots}$).

\item {\em Individual constants} of type $\iota$ (denoted by lowercase
      letters such as $\mathsf{a,b,c,\ldots}$).

\item {\em Function symbols} of every functional type $\sigma \neq \iota$
      (denoted by lowercase letters such as $\mathsf{f,g,h,\ldots}$).

\item The following logical constant symbols: the propositional constants
      $\mathsf{0}$ and $\mathsf{1}$ of type $o$; the equality  constant $\approx$
      of type $\iota \rightarrow \iota \rightarrow o$; the generalized disjunction
      and conjunction constants $\bigvee_{\pi}$ and $\bigwedge_{\pi}$ of type
      $\pi \rightarrow \pi \rightarrow o$, for every predicate type $\pi$;
      the generalized inverse implication constants $\leftarrow_{\pi}$, of type
      $\pi \rightarrow \pi \rightarrow o$, for every predicate type
      $\pi$; the existential quantifier $\exists_{\rho}$, of type
      $(\rho \rightarrow o)\rightarrow o$, for every argument type
      $\rho$.

\item The abstractor $\lambda$ and the parentheses ``$\mathsf{(}$'' and ``$\mathsf{)}$''.
\end{enumerate}
The set consisting of the predicate variables and the individual
variables of ${\cal H}$, will be called the set of {\em argument
variables} of ${\cal H}$. Argument variables will be usually denoted
by $\mathsf{V}$ and its subscripted versions.
\end{definition}
The existential quantifier in higher-order logic is usually
introduced in a different way than in first-order logic. So, for
example, in order to express the quantification of the argument
variable $\mathsf{V}$ of type $\rho$ over the expression
$\mathsf{E}$ one writes $(\exists_{\rho} \,(\lambda
\mathsf{V}.\mathsf{E}))$. For simplicity, we will use in this paper
the more familiar notation $(\exists_{\rho}\mathsf{V}\,\mathsf{E})$.

We proceed by defining the set of {\em positive expressions} of
${\cal H}$:
\begin{definition}
The set of positive expressions of the higher-order language ${\cal
H}$ is recursively defined as follows:
\begin{enumerate}
\item Every predicate variable (respectively, predicate constant)
      of type $\pi$ is a positive expression of type $\pi$; every
      individual variable (respectively, individual constant) of
      type $\iota$ is a positive expression of type $\iota$; the
      propositional constants $\mathsf{0}$ and $\mathsf{1}$ are
      positive expressions of type $o$.

\item If $\mathsf{f}$ is an $n$-ary function symbol and
      $\mathsf{E}_1,\ldots,\mathsf{E}_n$ are positive expressions of type
      $\iota$, then $(\mathsf{f}\,\,\mathsf{E}_1 \cdots \mathsf{E}_n)$ is
      a positive expression of type $\iota$.

\item If $\mathsf{E}_1$ is a positive expression of type $\rho \rightarrow
      \pi$ and $\mathsf{E}_2$ is a positive expression of type $\rho$, then
      $\mathsf{(}\mathsf{E}_1\mathsf{E}_2\mathsf{)}$ is a positive
      expression of type $\pi$.

\item If $\mathsf{V}$ is an argument variable of type $\rho$ and $\mathsf{E}$
      is a positive expression of type $\pi$, then
      $\mathsf{(\lambda V.E)}$ is a positive expression of type $\rho \rightarrow \pi$.

\item If $\mathsf{E}_1,\mathsf{E}_2$ are positive expressions of type $\pi$,
      then  $(\mathsf{E}_1 \bigwedge_{\pi} \mathsf{E}_2)$
      and $(\mathsf{E}_1 \bigvee_{\pi} \mathsf{E}_2)$ are
      positive expressions of type $\pi$.

\item If $\mathsf{E}_1,\mathsf{E}_2$ are positive expressions of type
      $\iota$, then $(\mathsf{E}_1 \approx\mathsf{E}_2)$ is a
      positive expression of type $o$.

\item If $\mathsf{E}$ is an expression of type $o$ and $\mathsf{V}$ is
      an argument variable of any type $\rho$, then $(\exists_{\rho}
      \mathsf{V}\,\mathsf{E})$ is a positive expression of
      type $o$.
\end{enumerate}
The notions of \emph{free} and \emph{bound} variables of an
expression are defined as usual. An expression is called
\emph{closed} if it does not contain any free variables.
\end{definition}
The set of {\em clausal expressions} of ${\cal H}$ can now be
specified:
\begin{definition}\label{clausal-expressions-definition}
The set of {\em clausal expressions} of the higher-order language
${\cal H}$ is defined as follows:
\begin{enumerate}
\item If $\mathsf{p}$ is a predicate constant of type $\pi$ and
      $\mathsf{E}$ is a closed positive expression of type $\pi$ then
      $\mathsf{p} \leftarrow_{\pi} \mathsf{E}$ is a clausal expression of
      type $o$ of ${\cal H}$, also called a {\em program clause}.

\item If $\mathsf{E}$ is a positive expression of type $o$, then
      $\mathsf{0}\leftarrow_{o} \mathsf{E}$ (usually denoted by
      $\leftarrow_{o} \mathsf{E}$ or just $\leftarrow \mathsf{E}$) is
      a clausal expression of type $o$ of ${\cal H}$, also called a
      {\em goal clause}.
\end{enumerate}
\end{definition}
Notice that (following the tradition of first-order logic
programming) we will often talk about the ``empty clause'' which is
denoted by $\Box$ and is equivalent to the propositional constant
$\mathsf{1}$.

The union of the sets of positive and clausal expressions of ${\cal
H}$ will be called the set of {\em expressions} of ${\cal H}$. To
denote that an expression $\mathsf{E}$ has type $\tau$, we will
often write $\mathsf{E}:\tau$; additionally, we write
$type(\mathsf{E})$ to denote the type of expression $\mathsf{E}$.
Expressions of type $\iota$ will be called \emph{terms} and of type
$o$ will be called \emph{formulas}. We will write $\leftarrow$,
$\wedge$ and $\vee$ instead of $\leftarrow_{o}$, $\bigwedge_o$ and
$\bigvee_o$. Moreover, instead of $\exists_{\rho}$ we will often
write $\exists$. When writing an expression, in order to avoid the
excessive use of parentheses, certain usual conventions will be
adopted (such as for example the usual priorities between logical
constants, the convention that application is left-associative and
that lambda abstraction extends as far to the right as possible, and
so on). Given an expression $\mathsf{E}$, we denote by
$FV(\mathsf{E})$ the set of all free variables of $\mathsf{E}$. By
overloading notation, we will also write $FV(S)$, where $S$ is a set
of expressions.

Notice that in Definition~\ref{clausal-expressions-definition}
above, a goal clause may contain two types of occurrences of
variables that serve a similar purpose, namely free argument
variables and argument variables that are existentially quantified.
From a semantic point of view, these two types of variables are
essentially the same. However, in a later section we will
distinguish them from an operational point of view: the free
argument variables that appear in a goal are the ones for which an
answer is sought for by the proof procedure; the argument variables
that are existentially quantified are essentially free variables for
which an answer is not sought for (something like the underscored
variables in Prolog systems). This distinction is not an important
one, and we could have proceeded in a different way (eg. by
disallowing existentially quantified variables from goals).
\begin{definition}
A program of ${\cal H}$ is a set of program clauses of ${\cal H}$.
\end{definition}
\begin{example}\label{closure-example}
The following is a higher-order program that computes the closure of
its input binary relation {\tt R}. The type of {\tt closure} is $\pi
= (\iota \rightarrow \iota \rightarrow o) \rightarrow \iota
\rightarrow \iota \rightarrow o$.
\[
\begin{array}{l}
{\tt closure} \leftarrow_{\pi} \mbox{\tt $\lambda$R.$\lambda$X.$\lambda$Y.(R X Y)}\\
{\tt closure} \leftarrow_{\pi} \mbox{\tt
$\lambda$R.$\lambda$X.$\lambda$Y.$\exists$Z((R X Z)$\wedge$(closure
R Z Y))}
\end{array}
\]
or even more compactly:
\[
\begin{array}{l}
{\tt closure} \leftarrow_{\pi} \mbox{\tt
($\lambda$R.$\lambda$X.$\lambda$Y.(R X Y)) $\bigvee_\pi$
($\lambda$R.$\lambda$X.$\lambda$Y.$\exists$Z((R X Z)$\wedge$(closure
R Z Y)))}
\end{array}
\]
A possible query could be: $\leftarrow \mbox{\tt (closure R a b)}$
(which intuitively requests for those binary relations such that the
pair $({\tt a},{\tt b})$ belongs to their transitive closure).
In a  Prolog-like extended syntax, the above program would have been
written as:
\[
\begin{array}{l}
\mbox{\tt closure(R,\,X,\,Y)\,:-\,R(X,\,Y).}\\
\mbox{\tt closure(R,\,X,\,Y)\,:-\,R(X,\,Z),\,closure(R,\,Z,\,Y).}
\end{array}
\]
and the corresponding query as $\leftarrow \mbox{\tt
closure(R,\,a,\,b)}$.\qed
\end{example}
\begin{example}\label{ordered-example}
We define a predicate {\tt ordered} which checks whether its second
argument (a list) is ordered according to its first argument (a
binary relation). The type of {\tt ordered} is $\pi = (\iota
\rightarrow \iota \rightarrow o) \rightarrow \iota \rightarrow o$
(notice that the type of a list is also $\iota$ since a list is
nothing more than a term). In Prolog-like syntax, the program is the
following:
\[
\begin{array}{l}
\mbox{\tt ordered(R,\,[\,]).}\\
\mbox{\tt ordered(R,\,[X]).}\\
\mbox{\tt ordered(R,\,[X,Y|T])\,:-\,R(X,\,Y),\,ordered(R,\,[Y|T]).}
\end{array}
\]
In the syntax of ${\cal H}$ (slightly extended with the standard
notation for lists), the above program can be written as follows:
\[
\begin{array}{l}
{\tt ordered} \leftarrow_{\pi} \mbox{\tt $\lambda$R.$\lambda$L.(L$\approx$[\,])}\\
{\tt ordered} \leftarrow_{\pi} \mbox{\tt $\lambda$R.$\lambda$L.($\exists$X(L$\approx$[X]))}\\
{\tt ordered} \leftarrow_{\pi} \mbox{\tt
$\lambda$R.$\lambda$L.($\exists$X$\exists$Y$\exists$T((L$\approx$[X,Y|T])$\wedge$(R\,
X\, Y)$\wedge$(ordered\, R\, [Y|T])))}
\end{array}
\]
Assume that we have also defined a binary relation {\tt less} which
succeeds if its first argument (eg. a natural number) is less than
the second one. Then, the query $\leftarrow \mbox{\tt
ordered\,\,less\,\,[1,4,7,10]}$ is expected to succeed. On the other
hand, the query $\leftarrow \mbox{\tt ordered\,\,R\,\,[a,b,c,d]}$
requests for all binary relations under which the list {\tt
[a,b,c,d]} is ordered. As it will become clear in the subsequent
sections of the paper, this is a meaningful question which can
obtain a reasonable answer.\qed
\end{example}

\section{Algebraic Lattices}\label{alglat-section}
In order to develop the semantics of ${\cal H}$, we first need to
introduce certain lattice-theoretic concepts. As it is well-known,
the standard semantics of classical (first-order) logic programming,
is based on {\em complete lattices} (see for example~\cite{lloyd}).
As we are going to see, the development of the semantics of ${\cal
H}$ is based on a special class of complete lattices, namely {\em
algebraic lattices} (see for example~\cite{Gra78}). An algebraic
lattice is a complete lattice in which every element can be created
by using certain {\em compact} (intuitively, ``simple'') elements of
the lattice. In our setting, these compact elements will be the ones
that the proof procedure will generate in order to answer queries
that involve uninstantiated predicate variables. We should mention
at this point that algebraic partially ordered sets are widely used
in domain theory (see for example~\cite{AJ94}).

We start by introducing some mathematical preliminaries concerning
lattice theory. Since the bibliography on partially ordered sets is
huge, certain results appear in one form or another in various
contexts, and they are often hard to locate in the exact form
needed. Propositions~\ref{diagonal},
\ref{monotonic-functions-make-a-complete-lattice} and
\ref{characterization-of-al} fall into this category; for reasons of
completeness, we have included short proofs for them. On the other
hand, Lemma~\ref{algebraicity-composition} is, to our knowledge,
new. We start with some basic definitions:
\begin{definition}
A set $P$ with a binary relation $\sqsubseteq_P$ is called a
\emph{partially ordered set} or \emph{poset} if $\sqsubseteq_P$ is
reflexive, transitive and antisymmetric.
\end{definition}
Usually, the subscript $P$ in $\sqsubseteq_P$ will be omitted when
it is obvious from context.
\begin{definition}
Let $P$ be a poset. An element $x \in P$ is called an \emph{upper
bound} for a subset $A \subseteq P$, if for every $y \in A$, $y
\sqsubseteq x$. If the set of upper bounds of $A$ has a least
element, then this element is called the \emph{least upper bound}
(or \emph{lub}) of $A$ and is denoted by $\bigsqcup A$.
Symmetrically, one can define the notions of \emph{lower bound} and
\emph{greatest lower bound} (or \emph{glb}) of $A$ (this last notion
denoted by $\bigsqcap A$).
\end{definition}
The following proposition (see for example~\cite{AJ94}[Proposition
2.1.4]) will prove useful later in the paper:
\begin{proposition}\label{set-of-sets}
Let $P$ be a poset and let $A, B, (A_i)_{i \in I}$ be subsets of
$P$. Then, the following statements hold (provided the glbs and lubs
occurring in the formulas exist):
\begin{enumerate}
\item $A \subseteq B$ implies $\bigsqcup A \sqsubseteq \bigsqcup B$.
\item If $A= \bigcup_{i \in I} A_i$, then $\bigsqcup A = \bigsqcup_{i \in
I} (\bigsqcup A_i)$.
%
\end{enumerate}
\end{proposition}
\begin{definition}
Let $P$ be a poset. A subset $A$ of $P$ is \emph{directed}, if it is
nonempty and each pair of elements of $A$ has an upper bound in $A$.
\end{definition}
\begin{definition}
Let $P$ and $Q$ be posets. A function $f: P \rightarrow Q$ is called
\emph{monotonic} if for all $x,y\in P$ with $x\sqsubseteq_P y$, we
have $f(x) \sqsubseteq_Q f(y)$. The set of all monotonic functions
from $P$ to $Q$ is denoted by $[P \stackrel{m}{\rightarrow} Q]$.
\end{definition}

Notice that monotonicity can be generalized in the obvious way for
functions $f: P^n \rightarrow Q$, $n>0$, since $P^n$ is also a poset
(where the partial order in this case is defined in a point-wise
way).

We are particularly interested in one type of posets, namely
\emph{complete lattices}:
\begin{definition}
A poset $L$ in which every subset has a least upper bound and a
greatest lower bound, is called a \emph{complete lattice}.
\end{definition}
In fact, there is a symmetry here: the existence of all least upper
bounds suffices to prove that a poset is indeed a complete lattice,
a fact that we will freely use throughout the paper.
\begin{proposition}\label{diagonal}
Let $P$ be a poset, $L$ a complete lattice and let $f:P\times P
\rightarrow L$ be a monotonic function. Then, $\bigsqcup_{x \in P,y
\in P} f(x,y) = \bigsqcup_{x \in P} f(x,x)$.
\end{proposition}
\begin{proof}
An easy proof using basic properties of posets (see for example the
corresponding proof for {\em domains}~\cite{Ten91}[Lemma 5.3, page
92]).\qed
\end{proof}
\begin{definition}\label{continuous-function}
Let $L$ be a complete lattice. A function $f:L \rightarrow L$ is
called \emph{continuous} if it is monotonic and for every directed
subset $A$ of $L$, we have $f(\bigsqcup A) = \bigsqcup f(A)$.
\end{definition}
We will write $\perp_L$ for the greatest lower bound of a complete
lattice $L$ (called the {\em bottom element of $L$}). A very useful
tool in lattice theory, is Kleene's fixpoint theorem:
\begin{theorem}
Let $L$ be a complete lattice. Then, every continuous function $f: L
\rightarrow L$ has a least fixpoint $\mathit{lfp}(f)$ given by
$\mathit{lfp}(f)=\bigsqcup_{n <\omega} f^n(\perp_L)$.
\end{theorem}
Let $A$ be an arbitrary set and $L$ be a complete lattice. Then, a
partial order can be defined on $A\rightarrow L$: for all $f,g \in A
\rightarrow L$, we write $f \sqsubseteq_{A \rightarrow L} g$ if for
all $a\in A$, $f(a) \sqsubseteq_L g(a)$. We will often use the
following proposition:
\begin{proposition}\label{monotonic-functions-make-a-complete-lattice}
Let $A$ be a poset, $L$ a complete lattice and let $F\subseteq [A
\stackrel{m}{\rightarrow} L]$. Then, for all $a \in A$, $(\bigsqcup
F)(a) = \bigsqcup_{f \in F} f(a)$ and $(\bigsqcap F)(a) =
\bigsqcap_{f \in F} f(a)$. Therefore, $[A \stackrel{m}{\rightarrow}
L]$ is a complete lattice.
\end{proposition}
\begin{proof}
We give the proof for $\bigsqcup$ (the proof for $\bigsqcap$ is
symmetrical).  Let $h \in A \rightarrow L$ such that $h(a) =
\bigsqcup_{f \in F} f(a)$, where the least upper bound is
well-defined because $L$ is a complete lattice. Notice that $h$ is
obviously an upper bound of $F$. Now let $g$ be an arbitrary upper
bound of $F$. For each $a \in A$, it holds that $g(a)$ is an upper
bound of $\{f(a) \mid f \in F\}$, thus $\bigsqcup_{f \in F} f(a)
\sqsubseteq g(a)$ which means that $h \sqsubseteq g$. Therefore,
$h=\bigsqcup F$.

It remains to show that $h$ is monotonic. Consider $x, y \in A$ such
that $x \sqsubseteq y$. For all $f \in F$ we have $f(x) \sqsubseteq
f(y)$ due to the monotonicity of $f$. Since $\bigsqcup_{f \in F}
f(y)$ is an upper bound of $\{f(y) \mid f \in F\}$, it is also an
upper bound of $\{f(x) \mid f \in F\}$. Therefore, $\bigsqcup_{f \in
F} f(x) \sqsubseteq \bigsqcup_{f \in F} f(y)$ and consequently $h$
is monotonic.\qed
\end{proof}
We will be interested in a certain type of complete lattices in
which every element can be ``created'' by using a set of {\em
compact} (intuitively, ``simple'') elements of the lattice:
\begin{definition}
Let $L$ be a complete lattice and let $c \in L$. Then $c$ is called
\emph{compact} if for every $A\subseteq L$ such that $c \sqsubseteq
\bigsqcup A$, there exists finite $A'\subseteq A$ such that $c
\sqsubseteq \bigsqcup A'$.
The set of all compact elements of $L$ is denoted by ${\cal K}(L)$.
\end{definition}
We can now define the notion of algebraic lattice (see for
example~\cite{Gra78}), which will prove to be the key
lattice-theoretic concept applicable to our context.
\begin{definition}\label{algebraic-lattice}
A complete lattice $L$ is called \emph{algebraic} if every element
of $L$ is the least upper bound of a set of compact elements of $L$.
\end{definition}
The name ``algebraic lattice'' is due to G. Birkhoff~\cite{Bir67}
(who did not assume completeness at that time). In the literature,
algebraic lattices are also called {\em compactly generated
lattices}.
\begin{example}\label{zero-one}
The set $L=\{0,1\}$ under the usual numerical ordering is an
algebraic lattice with ${\cal K}(L) = \{0,1\}$.

Let $S$ be a set. Then, $2^S$, the set of all subsets of $S$, forms
a complete lattice under set inclusion. It is easy to see that this
is an algebraic lattice whose compact elements are the finite
subsets of $S$.\qed
\end{example}
Let $P$ be a poset. Given $B \subseteq P$ and $x \in P$, we write
$B_{[x]} = \{ b \in B \mid b \sqsubseteq_P x\}$. We have the
following easy proposition:
\begin{proposition}\label{characterization-of-al}
Let $L$ be an algebraic lattice. Then, for every $x \in L$, $x=
\bigsqcup {\cal K}(L)_{[x]}$.
\end{proposition}
\begin{proof}
Obviously it holds that $\bigsqcup {\cal K}(L)_{[x]} \sqsubseteq x$.
We show that $x \sqsubseteq \bigsqcup {\cal K}(L)_{[x]}$. By
Definition~\ref{algebraic-lattice}, there exists $A \subseteq {\cal
K}(L)$ such that $x = \bigsqcup A$. Obviously, $A \subseteq {\cal
K}(L)_{[x]}$. Therefore, by Proposition~\ref{set-of-sets},
$\bigsqcup A \sqsubseteq \bigsqcup {\cal K}(L)_{[x]}$ and
consequently $x \sqsubseteq \bigsqcup {\cal K}(L)_{[x]}$.\qed
\end{proof}
Given an algebraic lattice $L$, the set ${\cal K}(L)$ will be called
\emph{the basis of $L$}. If additionally, ${\cal K}(L)$ is
countable, then $L$ will be called an \emph{$\omega$-algebraic
lattice}.

In the rest of this section we will define a particular class of
algebraic lattices that arise in our semantics of higher-order logic
programming. This class will be characterized by
Lemma~\ref{algebraicity-composition} that follows. We first need to
define the notion of ``step functions'' (see for
example~\cite{AJ94}) which are used to build the compact elements of
our algebraic lattices.
\begin{definition}\label{step-functions}
Let $A$ be a poset and L be an algebraic lattice. For each $a \in A$
and $c \in {\cal K}(L)$, we define the function $(a \searrow c) : A
\rightarrow L$ as
\[ (a \searrow c)(x) = \left\{\begin{array}{ll}
                       c,  &\mbox{if $a \sqsubseteq_{A} x$}\\
                       \perp_{L}, &\mbox{otherwise}
                       \end{array}
                       \right.
\]
The functions of the above form will be called the \emph{step
functions} of $A \rightarrow L$.
\end{definition}
\begin{example}\label{example-of-step-functions}
Consider a non-empty set $A$ equipped with the trivial partial order
that relates every element of $A$ to itself, ie., $a \sqsubseteq_A
a$ for all $a$. Moreover, let $L=\{0,1\}$ (which by
Example~\ref{zero-one} is an algebraic lattice). Then, for every $a
\in A$, $(a \searrow 1)$ is the function that returns 1 iff its
argument is equal to $a$. In other words $(a \searrow 1)$ is the
singleton set $\{a\}$. On the other hand, for every $a$, $(a
\searrow 0)$ corresponds to the empty set.

As a second example, assume that $A$ is the set of finite subsets of
$\mathbb{N}$ and that $L=\{0,1\}$. Then, for any finite set $a\in
A$, $(a \searrow 1)$ is the function that given any finite set $x$
such that $x\supseteq a$, $(a \searrow 1)(x) = 1$. In other words,
$(a \searrow 1)$ is a set consisting of $a$ and all its (finite)
supersets. On the other hand, for every $a$, $(a \searrow 0)$ is the
function that given any finite set $x$, $(a \searrow 0)(x) = 0$,
ie., it corresponds to the empty set (of sets).\qed
\end{example}
The following lemma (which we have not seen explicitly stated
before) identifies a class of algebraic lattices that will play the
central role in the development of the semantics of higher-order
logic programming. An important characteristic of these lattices is
that they have a simple characterization of their basis. The proof
of the lemma is given in
Appendix~\ref{appendix-algebraicity-composition}.
\begin{lemma}\label{algebraicity-composition}
Let $A$ be a poset and $L$ be an algebraic lattice. Then, $[A
\stackrel{m}{\rightarrow} L]$ is an algebraic lattice whose basis is
the set of all least upper bounds of finitely many step functions
from $A$ to $L$. If, additionally, $A$ is countable and $L$ is an
$\omega$-algebraic lattice then $[A \stackrel{m}{\rightarrow} L]$ is
an $\omega$-algebraic lattice.
\end{lemma}

We can now outline the reasons why algebraic lattices play such an
important role in our context. As we have already mentioned, one of
the contributions of the paper is that it allows the treatment of
queries with uninstantiated predicate variables. The results
of~\cite{Wa91a} indicate that (due to continuity), if a relation
satisfies a predicate, then some ``finite representative'' of this
relation also satisfies it. This gives the idea of defining a
semantics which makes these ``finite representatives'' more
explicit. Intuitively, these finite representatives are the compact
elements of an algebraic lattice. From an operational point of view,
restricting attention to the compact elements allows us to answer
queries with uninstantiated variables: if the set of compact
elements is enumerable then we can try them one by one examining in
each case whether the query is satisfied.

More formally, since our lattices are algebraic and satisfy the
conditions of Lemma~\ref{algebraicity-composition}, we have a
relatively easy characterization of their sets of compact elements
(as suggested by Lemma~\ref{algebraicity-composition}). Moreover, as
we are going to see, if we restrict attention to Herbrand
interpretations (see Section~\ref{min-Herbrand}), then the lattices
that we have to consider are all $\omega$-algebraic and therefore
their sets of compact elements are countable. For these lattices it
turns out that we can devise an effective procedure for enumerating
their compact elements which leads us to an effective proof
procedure for our higher-order language.

\section{The Semantics of ${\cal H}$}\label{semantics-of-H}
The semantics of ${\cal H}$ is built upon the notion of algebraic
lattice. Recall that an algebraic lattice is a complete lattice $L$
with the additional property that every element $x$ of $L$ is the
least upper bound of ${\cal K}(L)_{[x]}$.

\subsection{The Semantics of Types}
Before specifying the semantics of expressions of ${\cal H}$ we need
to provide the set-theoretic meaning of the types of expressions of
${\cal H}$ with respect to a set $D$ (where $D$ is later going to be
the domain of our interpretations). The fact that a given type
$\tau$ denotes a set $\lsem \tau \rsem_D$ will mean that an
expression of type $\tau$ denotes an element of $\lsem \tau\rsem_D$.
In other words, the semantics of types help us understand what are
the meanings of the expressions of our language. In the following
definition we define simultaneously and recursively two things: the
semantics $\lsem \tau \rsem_D$ of a type $\tau$ and the
corresponding partial order $\sqsubseteq_{\tau}$\footnote{Notice
that we are writing $\sqsubseteq_{\tau}$ instead of the more
accurate $\sqsubseteq_{\lsem \tau\rsem_D}$. In the following, for
brevity reasons we will often use the former (simpler) notation.
Similarly, we will often write $\perp_{\pi}$ instead of
$\perp_{\lsem \pi\rsem_D}$.}.
\begin{definition}\label{semantics-of-types}
Let $D$ be a non-empty set. Then:
\begin{itemize}
\item $\lsem \iota \rsem_D = D$, and $\sqsubseteq_\iota$
is the trivial partial order such that $d \sqsubseteq_\iota d$, for
all $d \in D$.
\item $\lsem \iota^n \rightarrow \iota \rsem_D = D^n\rightarrow D$.
A partial order for this case will not be needed.
\item $\lsem o \rsem_D = \{0, 1\}$, and $\sqsubseteq_o$
is the partial order defined by the numerical ordering on $\{0,
1\}$.
\item $\lsem \iota \rightarrow \pi \rsem_D = D \rightarrow \lsem \pi
\rsem_D$, and $\sqsubseteq_{\iota \rightarrow \pi}$ is the partial
order defined as follows: for all $f,g \in \lsem \iota \rightarrow
\pi \rsem_D$, $f\sqsubseteq_{\iota \rightarrow \pi} g$ if and only
if $f(d) \sqsubseteq_{\pi} g(d)$, for all $d \in D$.

\item $\lsem \pi_1 \rightarrow \pi_2 \rsem_D =
[{\cal K}(\lsem \pi_1 \rsem_D) \stackrel{m}{\rightarrow} \lsem \pi_2
\rsem_D]$, and $\sqsubseteq_{\pi_1 \rightarrow \pi_2}$ is the
partial order defined as follows: for all $f,g \in \lsem \pi_1
\rightarrow \pi_2 \rsem_D$, $f\sqsubseteq_{\pi_1 \rightarrow \pi_2}
g$ if and only if $f(d) \sqsubseteq_{\pi_2} g(d)$, for all $d \in
{\cal K}(\lsem \pi_1 \rsem_D)$.
\end{itemize}
\end{definition}

It is not immediately obvious that the last case in the above
definition is well-defined. More specifically, in order for the
quantity ${\cal K}(\lsem \pi_1 \rsem_D)$ to make sense, $\lsem \pi_1
\rsem_D$ must be a complete lattice. This is ensured by the
following lemma:
\begin{lemma}\label{pi-is-algebraic-lattice}
Let $D$ be a non-empty set. Then, for every $\pi$, $\lsem \pi
\rsem_D$ is an algebraic lattice ($\omega$-algebraic if $D$ is
countable).
\end{lemma}
\begin{proof}
The proof is by a simple induction on the structure of $\pi$. The
basis case is for $\pi = o$ and holds trivially (see
Example~\ref{zero-one}). For the induction step, we distinguish two
cases.  The first case is for $\pi = \iota \rightarrow \pi_1$. Then,
$\lsem \iota \rightarrow \pi_1 \rsem_D = D \rightarrow \lsem \pi_1
\rsem_D$. Notice now that $D$ is partially ordered by the trivial
partial order $\sqsubseteq_{\iota}$, and it holds that $D
\rightarrow \lsem \pi_1\rsem_D = [D \stackrel{m}{\rightarrow} \lsem
\pi_1 \rsem_D]$ (monotonicity is trivial in this case). By the
induction hypothesis and Lemma~\ref{algebraicity-composition} it
follows that $\lsem \pi \rsem_D$ is an algebraic lattice
($\omega$-algebraic if $D$ is countable). The second case is for
$\pi = \pi_1 \rightarrow \pi_2 $, and the result follows by the
induction hypothesis and Lemma~\ref{algebraicity-composition}.\qed
\end{proof}
The following definition gives us a convenient shorthand when we
want to refer to an object that is either a compact element or a
member of the domain $D$ of our interpretations. This shorthand will
be used in various places of the paper.
\begin{definition}\label{definition-basic-elements}
Let $D$ be a non-empty set and let $\rho$ be an argument type.
Define:
\[ {\cal F}_D(\rho) = \left\{\begin{array}{ll}
                       D,  &\mbox{if $\rho = \iota$}\\
                       {\cal K}(\lsem \rho \rsem_D), &\mbox{otherwise}
                       \end{array}
                       \right.
\]
The set ${\cal F}_D(\rho)$ will be called the \emph{set of basic
elements of type $\rho$} (with respect to the set $D$).
\end{definition}
\begin{example}\label{example-of-domains}
Consider the type $\iota \rightarrow o$ (a first-order predicate
with one argument has this type). By
Definition~\ref{semantics-of-types}, $\lsem \iota \rightarrow o
\rsem_D$ is the set of all functions from $D$ to $\{0,1\}$ (or
equivalently, of arbitrary subsets of $D$).

As a second example, consider the type $(\iota \rightarrow
o)\rightarrow o$. This is the type of a predicate which takes as its
only parameter another predicate which is first-order; for example,
{\tt p} in Example~\ref{the-ultimate-example} has this type. Then,
it can be verified using Lemma~\ref{algebraicity-composition} and
Example~\ref{example-of-step-functions} that the set ${\cal K}(\lsem
\iota \rightarrow o\rsem_D)$ is the set of all {\em finite}
functions from $D$ to $\{0,1\}$ (or equivalently, of finite subsets
of $D$). By Definition~\ref{semantics-of-types}, $\lsem (\iota
\rightarrow o) \rightarrow o\rsem_D$ is the set of all monotonic
functions from finite subsets of $D$ (ie., elements of ${\cal
K}(\lsem \iota \rightarrow o\rsem_D)$) to $\{0,1\}$. In other words,
in the semantics of ${\cal H}$, a predicate of type $(\iota
\rightarrow o)\rightarrow o$ will denote a monotonic function from
finite subsets of $D$ to $\{0,1\}$. The role that monotonicity plays
in this context can be intuitively explained by considering again
Example~\ref{the-ultimate-example}: if {\tt p} is true of a finite
set, then this set must contain both {\tt 0} and {\tt s(0)}. But
then, {\tt p} will also be true for every superset of this set
(since every superset also contains both {\tt 0} and {\tt s(0)}). As
we are going to see, the meaning of all the higher-order predicates
that are defined in a program will possess the monotonicity
property.\qed
\end{example}
It should be noted at this point that the semantics of types of our
language is in some sense a finitary version of the one given
in~\cite{Wa91a}, where the denotation of a type of the form $\pi_1
\rightarrow \pi_2$ is the set of all continuous functions from the
denotation of $\pi_1$ to the denotation of $\pi_2$ (more details on
the connections between the two approaches will be given in
Section~\ref{comparison-with-Wadge}). Notice now that in our
interpretation of types, only monotonicity is required; actually,
continuity is not applicable in our interpretation: given a type
$\pi_1\rightarrow \pi_2$, it would be meaningless to talk about the
continuous functions from ${\cal K}(\lsem \pi_1 \rsem_D)$ to $\lsem
\pi_2 \rsem_D$ because ${\cal K}(\lsem \pi_1 \rsem_D)$ is not in
general a complete lattice\footnote{To see this, take $\pi_1=\iota
\rightarrow o$ and let $D$ be an infinite set. Then, ${\cal K}(\lsem
\iota \rightarrow o \rsem_D)$ consists of all finite subsets of $D$
and is not a complete lattice (since the least upper bound of a set
of finite sets can itself be infinite).} as required by the
definition of continuity. However, as we are going to see,
monotonicity suffices in order to establish that the immediate
consequence operator of every program is continuous
(Lemma~\ref{T_P_continuous}) and therefore has a least fixed-point.

As a last remark, we should mention that the interpretation of types
given in Definition~\ref{semantics-of-types}, does not apply to the
inverse implication operator $\leftarrow_{\pi}$ of ${\cal H}$, whose
denotation is not monotonic (for example, notice that negation can
be implicitly defined with the use of implication). However, since
the use of $\leftarrow_{\pi}$ is not allowed inside positive
expressions, the non-monotonicity of $\leftarrow_{\pi}$ does not
create any semantic problems.

\subsection{The Semantics of Expressions}\label{subsection-semantics-of-expressions}
We can now proceed to give meaning to the expressions of ${\cal H}$.
This is performed by first defining the notions of
\emph{interpretation} and \emph{state} for ${\cal H}$:
\begin{definition}
An interpretation $I$ of ${\cal H}$ consists of:
\begin{enumerate}\label{definition-of-interpretation}
\item a nonempty set $D$, called the \emph{domain} of $I$

\item an assignment to each individual constant symbol $\mathsf{c}$,
      of an element $I(\mathsf{c})\in D$

\item an assignment to each predicate constant $\mathsf{p}$ of type $\pi$, of an
      element $I(\mathsf{p})\in \lsem \pi \rsem_D$

\item an assignment to each function symbol $\mathsf{f}$ of type $\iota^n
      \rightarrow \iota$, of a function $I(\mathsf{f})\in D^n \rightarrow D$.
\end{enumerate}
\end{definition}
\begin{definition}
Let $D$ be a nonempty set. Then, a \emph{state $s$ of ${\cal H}$
over $D$} is a function that assigns to each argument variable
$\mathsf{V}$ of type $\rho$ of ${\cal H}$ an element $s(\mathsf{V})
\in {\cal F}_D(\rho)$.
\end{definition}

In the following, $s[d/\mathsf{V}]$ is used to denote a state that
is identical to $s$ the only difference being that the new state
assigns to $\mathsf{V}$ the value $d$.

Before we proceed to formally define the semantics of expressions of
${\cal H}$, a short discussion on the semantics of {\em application}
is needed. The key technical difficulty we have to confront can be
explained by reconsidering Example~\ref{the-ultimate-example} in the
more formal context that we have now developed.
\begin{example}
Consider again the program from Example~\ref{the-ultimate-example}:
\[
\begin{array}{l}
\mbox{\tt p(Q):-Q(0),Q(s(0)).} \\
\mbox{\tt nat(0).}\\
\mbox{\tt nat(s(X)):-nat(X).}
\end{array}
\]
Consider also the query {\tt $\leftarrow$ p(nat)}. The type of {\tt
p} is $(\iota \rightarrow o)\rightarrow o$, while the type of {\tt
nat} is $\iota \rightarrow o$. Let $I$ be an interpretation  with
underlying domain $D$. Then, according to
Definition~\ref{definition-of-interpretation}, $I({\tt p})$ must be
a monotonic function from ${\cal K}(\lsem \iota \rightarrow
o\rsem_D)$ to $\{0,1\}$. Moreover, according to
Example~\ref{example-of-domains}, ${\cal F}_D(\iota \rightarrow o)$
consists of all the {\em finite} sets of elements of $D$. But
$I({\tt nat})$ is a member of $\lsem \iota \rightarrow o\rsem_D$ and
can therefore be an infinite set. How can we apply $I({\tt p})$ to
$I({\tt nat})$? To overcome this problem, observe that in order for
the predicate {\tt p} to succeed for its argument {\tt Q}, it only
has to examine a ``finite number of facts'' about {\tt Q} (namely
whether {\tt Q} is true of {\tt 0} and {\tt s(0)}). This remark
suggests that the meaning of {\tt p(nat)} can be established
following a non-standard interpretation of application: we apply
$I({\tt p})$ to all the ``finite approximations'' of $I({\tt nat})$,
ie., to all elements of ${\cal K}(\lsem \iota \rightarrow
o\rsem_D)_{[I({\tt nat})]}$, and then take the least upper bound of
the results. Notice that our approach heavily relies on the fact
that our semantic domains are algebraic lattices: every element of
such a lattice (like $I({\tt nat})$ in our example) is the least
upper bound of the compact elements of the lattice that are below it
(the finite subsets of $I({\tt nat})$ in our case).\qed
\end{example}

We can now proceed to present the semantics of ${\cal H}$:
\begin{definition}\label{definition-semantics-of-expressions}
Let $I$ be an interpretation of ${\cal H}$, let $D$ be the domain of
$I$, and let $s$ be a state over $D$. Then, the semantics of
expressions of ${\cal H}$ with respect to $I$ and $s$, is defined as
follows:
\begin{enumerate}
\item $\lsem \mathsf{0} \rsem_s (I) = 0$

\item $\lsem \mathsf{1} \rsem_s (I) = 1$

\item $\lsem \mathsf{c} \rsem_s (I) = I(\mathsf{c})$, for every
      individual constant $\mathsf{c}$

\item $\lsem \mathsf{p} \rsem_s (I) = I(\mathsf{p})$, for every
      predicate constant $\mathsf{p}$

\item $\lsem \mathsf{V} \rsem_s (I) = s(\mathsf{V})$, for every
      argument variable $\mathsf{V}$

\item $\lsem (\mathsf{f}\,\,\mathsf{E}_1\cdots \mathsf{E}_n) \rsem_s (I) =
      I(\mathsf{f})\,\,\lsem \mathsf{E}_1\rsem_s (I) \cdots \lsem \mathsf{E}_n\rsem_s (I)$,
      for every $n$-ary function symbol $\mathsf{f}$

\item $\lsem \mathsf{(}\mathsf{E}_1\mathsf{E}_2\mathsf{)} \rsem_s
      (I)= \bigsqcup_{b \in B} (\lsem \mathsf{E}_1 \rsem_s (I)(b))$,
      where $B = {\cal F}_D(type(\mathsf{E}_2))_{[\lsem \mathsf{E}_2\rsem_s(I)]}$

\item $\lsem \mathsf{(\lambda V.E)} \rsem_s (I) =\lambda d.\lsem
      \mathsf{E}\rsem_{s[d/\mathsf{V}]}(I)$, where $d$ ranges over
      ${\cal F}_D(type(\mathsf{V}))$

\item $\lsem (\mathsf{E}_1 \bigvee_{\pi} \mathsf{E}_2)\rsem_s (I) =
      \bigsqcup_{\pi}\{\lsem \mathsf{E}_1\rsem_s(I),\lsem \mathsf{E}_2\rsem_s(I)\}$, where
      $\bigsqcup_{\pi}$ is the least upper bound function on $\lsem \pi
      \rsem_D$

\item $\lsem (\mathsf{E}_1 \bigwedge_{\pi} \mathsf{E}_2)\rsem_s (I) =
      \bigsqcap_{\pi}\{\lsem \mathsf{E}_1\rsem_s(I),\lsem \mathsf{E}_2\rsem_s(I)\}$, where
      $\bigsqcap_{\pi}$ is the greatest lower bound function on $\lsem \pi
      \rsem_D$

\item $\lsem (\mathsf{E}_1 \mathsf{\approx} \mathsf{E}_2)\rsem_s (I) = \left\{\begin{array}{ll}
                                               1, & \mbox{if $\lsem \mathsf{E}_1 \rsem_s
                                                    (I) = \lsem \mathsf{E}_2 \rsem_s
                                                    (I)$}\\
                                               0, & \mbox{otherwise}
                                                   \end{array} \right. $

\item $\lsem (\exists \mathsf{V}\, \mathsf{E}) \rsem_s (I)=
                                 \left\{\begin{array}{ll}
                                      1, & \mbox{if there exists $d \in {\cal F}_D(type(\mathsf{V}))$
                                                 such that $\lsem \mathsf{E}
                                                 \rsem_{s[d/\mathsf{V}]}
                                                 (I)=1$}\\
                                      0, & \mbox{otherwise}
                                         \end{array} \right. $

\item $\lsem (\mathsf{p} \leftarrow_{\pi} \mathsf{E})\rsem_s (I)=
       \left\{\begin{array}{ll}
         1, & \mbox{if $\lsem \mathsf{E} \rsem_s (I)
         \sqsubseteq_{\pi} I(\mathsf{p})$}\\
         0, & \mbox{otherwise}
              \end{array} \right. $

\item $\lsem (\leftarrow \mathsf{E}) \rsem_s (I)= \left\{\begin{array}{ll}
                                               1, & \mbox{if $\lsem \mathsf{E} \rsem_s
                                                    (I) = 0$}\\
                                               0, & \mbox{otherwise}
                                                   \end{array} \right. $

\end{enumerate}
\end{definition}

For closed expressions $\mathsf{E}$ we will often write $\lsem
\mathsf{E} \rsem(I)$ instead of $\lsem \mathsf{E} \rsem_s(I)$
(since, in this case, the meaning of $\mathsf{E}$ is independent of
$s$).

We need to demonstrate that the semantic valuation function $\lsem
\cdot \rsem$ assigns to every expression of ${\cal H}$ an element of
the corresponding semantic domain. More formally, we need to
establish that for every interpretation $I$ with domain $D$, for
every state $s$ over $D$ and for all expressions $\mathsf{E}:\rho$,
it holds that $\lsem \mathsf{E} \rsem_s(I) \in \lsem \rho \rsem_D$.
In order to prove this, the following definition is needed:
\begin{definition}\label{partial-order-on-states}
Let ${\cal S}_{{\cal H},D}$ be the set of states of ${\cal H}$ over
the nonempty set $D$. We define the following partial order on
${\cal S}_{{\cal H},D}$: for all $s_1,s_2 \in {\cal S}_{{\cal
H},D}$, $s_1 \sqsubseteq_{{\cal S}_{{\cal H},D}} s_2$ iff for every
argument variable $\mathsf{V}:\rho$ of ${\cal H}$, $s_1(\mathsf{V})
\sqsubseteq_{\rho} s_2(\mathsf{V})$.
\end{definition}

The following lemma states that
Definition~\ref{definition-semantics-of-expressions} assigns to
expressions elements of the corresponding semantic domain. Notice
that in order to establish this, we must also prove simultaneously
that the meaning of positive expressions is monotonic with respect
to states.
\begin{lemma}\label{monotonicity-of-semantics-wrt-state}
Let $\mathsf{E}:\rho$ be an expression of ${\cal H}$ and let $D$ be
a nonempty set. Moreover, let $s,s_1,s_2$ be states over $D$ and let
$I$ be an interpretation over $D$. Then:
\begin{enumerate}
\item $\lsem \mathsf{E} \rsem_s(I) \in \lsem \rho \rsem_D$.

\item If $\mathsf{E}$ is positive and $s_1 \sqsubseteq_{{\cal S}_{{\cal H},D}} s_2$
      then $\lsem \mathsf{E} \rsem_{s_1}(I) \sqsubseteq_{\rho} \lsem \mathsf{E} \rsem_{s_2}(I)$.
\end{enumerate}
\end{lemma}
The proof of the lemma is given in
Appendix~\ref{appendix-monotonicity-state}.

We can now define the important notion of a \emph{model} of a set of
formulas:
\begin{definition}
Let $S$ be a set of formulas of ${\cal H}$ and let $I$ be an
interpretation of ${\cal H}$. We say that $I$ is a \emph{model} of
$S$ if for every $\mathsf{F} \in S$ and for every state $s$ over the
domain of $I$, $\lsem \mathsf{F} \rsem_s (I) = 1$.
\end{definition}
We close this section with the definitions of the notions of
unsatisfiability and of logical consequence of a set of formulas.
\begin{definition}
Let $S$ be a set of formulas  of ${\cal H}$. We say that $S$ is
\emph{unsatisfiable} if no interpretation of ${\cal H}$ is a model
for $S$.
\end{definition}
\begin{definition}
Let $S$ be a set of formulas and $\mathsf{F}$ a formula of ${\cal
H}$. We say that $\mathsf{F}$ is a \emph{logical consequence} of $S$
if, for every interpretation $I$ of ${\cal H}$, $I$ is a model of
$S$ implies that $I$ is a model of $\mathsf{F}$.
\end{definition}

\subsection{A Comparison with the Continuous
Semantics}\label{comparison-with-Wadge}
In this subsection we give a brief comparison of the proposed
semantics with the semantics introduced in~\cite{Wa91a}. A complete
presentation of such a comparison would require a detailed
presentation of the approach introduced in~\cite{Wa91a} and its
adaptation to the richer language ${\cal H}$ introduced in this
paper. We avoid such an extensive comparison by outlining the main
points in an intuitive way.

As already mentioned, the source language considered in~\cite{Wa91a}
is restricted compared to ${\cal H}$. However, the semantics
of~\cite{Wa91a} can be appropriately extended to apply to ${\cal H}$
as well. Given a non-empty set $D$, let us denote by $\lsem \rho
\rsem_D^*$ the semantics of an argument type $\rho$ in $D$ under the
approach of~\cite{Wa91a}. Then, the semantics of types is defined as
follows:
\begin{itemize}
\item $\lsem \iota \rsem_D^* = D$.
\item $\lsem \iota^n \rightarrow \iota \rsem_D^* = D^n\rightarrow D$.
\item $\lsem o \rsem_D^* = \{0, 1\}$.
\item $\lsem \iota \rightarrow \pi \rsem_D^* = D \rightarrow \lsem \pi
\rsem_D^*$.

\item $\lsem \pi_1 \rightarrow \pi_2 \rsem_D^* =
[\lsem \pi_1 \rsem_D^* \stackrel{c}{\rightarrow} \lsem \pi_2
\rsem_D^*]$.
\end{itemize}
where by $[A \stackrel{c}{\rightarrow} B]$ we denote the set of
continuous functions from $A$ to $B$. The corresponding partial
orders can be easily defined as in
Definition~\ref{semantics-of-types}. The semantics of expressions
can be defined in an analogous way as in
Definition~\ref{definition-semantics-of-expressions}, the main
difference being that the semantics of application is the standard
one. Roughly speaking, one can say that the semantics
of~\cite{Wa91a} is the logic programming analogue of the standard
denotational semantics of functional programming
languages~\cite{Ten91}. In the following, we will refer to the
semantics of~\cite{Wa91a} as the ``{\em continuous semantics}''.

It is relatively easy to show that for every argument type $\rho$ of
${\cal H}$ there is a bijection between the sets $\lsem \rho
\rsem_D$ and $\lsem \rho \rsem_D^*$. Similarly, there is a bijection
between the set of interpretations of ${\cal H}$ under the proposed
semantics and the set of interpretations of ${\cal H}$ under the
continuous semantics. Then, the following proposition can be
established:
\begin{proposition}
Let $\mathsf{P}$ be a program and let $\mathsf{F}$ be a formula of
${\cal H}$. Then, $\mathsf{F}$ is a logical consequence of
$\mathsf{P}$ under the proposed semantics iff $\mathsf{F}$ is a
logical consequence of $\mathsf{P}$ under the continuous semantics.
\end{proposition}
In other words, the two semantics, despite their differences, are
closely related. The key advantage of the proposed semantics is that
it is much closer to the SLD-resolution proof procedure that will be
introduced in Section~\ref{proof-procedure-section}. More
specifically:
\begin{itemize}
\item The compact elements of our algebraic lattices correspond to
      the {\em basic expressions} that are a vital characteristic of the
      proposed proof procedure (see
      Subsection~\ref{subsection-basic-expressions}).

\item The notion of {\em answer} and {\em correct answer} for a query
      (see Definitions~\ref{answer} and~\ref{correct-answer}) can
      now be accurately defined. Notice that the notion of correct
      answer must be quite close to that of {\em computed answer} in
      order to be able to state the main completeness theorem.

\end{itemize}
In conclusion, the proposed semantics allows us to define an
SLD-resolution proof procedure and it helps us formalize and prove
its completeness. It is unclear to us whether (and how) this could
have been accomplished by relying on the continuous semantics.

\section{Minimum Herbrand Model Semantics}\label{min-Herbrand}
Herbrand interpretations constitute a special form of
interpretations that have proven to be a cornerstone of first-order
logic programming. Analogously, we have:
\begin{definition}
The Herbrand universe $U_{\cal H}$ of ${\cal H}$ is the set of all
terms that can be formed out of the individual constants and the
function symbols of ${\cal H}$.
\end{definition}
\begin{definition}
A Herbrand interpretation $I$ of ${\cal H}$ is an interpretation
such that:
\begin{enumerate}
\item The domain of $I$ is the Herbrand universe $U_{\cal H}$ of ${\cal
      H}$.

\item For every individual constant $\mathsf{c}$, $I(\mathsf{c}) =
      \mathsf{c}$.

\item For every predicate constant $\mathsf{p}$ of type $\pi$,
      $I(\mathsf{p}) \in \lsem \pi \rsem_{U_{\cal H}}$.

\item For every $n$-ary function symbol $\mathsf{f}$ and all terms
      $\mathsf{t}_1,\ldots,\mathsf{t}_n$ of $U_{\cal H}$,
      $I(\mathsf{f})\, \mathsf{t}_1 \cdots \mathsf{t}_n =
      \mathsf{f}\, \mathsf{t}_1 \cdots \mathsf{t}_n$.
\end{enumerate}
\end{definition}
Since all Herbrand interpretations have the same underlying domain,
we will often refer to a ``Herbrand state $s$'', meaning a state
whose underlying domain is $U_{\cal H}$. As it is a standard
practice in logic programming, we will often refer to an
``interpretation of a set of formulas $S$'' rather than of the
underlying language ${\cal H}$. In this case, we will implicitly
assume that the set of individual constants and function symbols are
those that appear in $S$. Under this assumption, we will often talk
about the ``Herbrand universe $U_S$ of a set of formulas $S$''.

We should also note that since the Herbrand universe is a countable
set, by Lemma~\ref{pi-is-algebraic-lattice}, for every predicate
type $\pi$, $\lsem \pi \rsem_{U_{\cal H}}$ is an $\omega$-algebraic
lattice (ie., it has a countable basis).

We can now proceed to examine properties of Herbrand
interpretations. In the following we denote the set of Herbrand
interpretations of a program $\mathsf{P}$ with ${\cal
I}_{\mathsf{P}}$ .
\begin{definition}
Let $\mathsf{P}$ be a program. We define the following partial order
on ${\cal I}_{\mathsf{P}}$: for all $I,J \in {\cal I}_\mathsf{P}$,
$I \sqsubseteq_{{\cal I}_\mathsf{P}} J$ iff for every $\pi$ and for
every predicate constant $\mathsf{p}:\pi$ of $\mathsf{P}$,
$I(\mathsf{p}) \sqsubseteq_{\pi} J(\mathsf{p})$.
\end{definition}
\begin{lemma}\label{herbrand-interpretations-make-a-lattice}
Let $\mathsf{P}$ be a program and let ${\cal I} \subseteq {\cal
I}_{\mathsf{P}}$. Then, for every predicate $\mathsf{p}$ of
$\mathsf{P}$, $(\bigsqcup {\cal I})(\mathsf{p}) = \bigsqcup_{I\in
{\cal I}} I(\mathsf{p})$ and $(\bigsqcap {\cal I})(\mathsf{p}) =
\bigsqcap_{I\in {\cal I}} I(\mathsf{p})$. Therefore, ${\cal
I}_\mathsf{P}$ is a complete lattice under $\sqsubseteq_{{\cal
I}_\mathsf{P}}$.
\end{lemma}
\begin{proof}
We give the proof for $\bigsqcup$; the proof for $\bigsqcap$ is
symmetrical and omitted. Let $J\in {\cal I}_\mathsf{P}$ such that
for every $\mathsf{p}:\pi$ in $\mathsf{P}$, $J(\mathsf{p})=
\bigsqcup_{I\in {\cal I}} I(\mathsf{p})$. Notice that
$\bigsqcup_{I\in {\cal I}} I(\mathsf{p})$ is well-defined since
$\lsem \pi\rsem_{U_{\mathsf{P}}}$ is a complete lattice. Notice also
that $J$ is an upper-bound for ${\cal I}$ because for every $I \in
{\cal I}$, $I \sqsubseteq_{{\cal I}_{\mathsf{P}}} J$. Let $J'$ be an
arbitrary upper bound of ${\cal I}$. Then, for every
$\mathsf{p}:\pi$, it holds that $J'(\mathsf{p})$ is an upper bound
of $\{I(\mathsf{p}) \mid I \in {\cal I}\}$, and therefore
$\bigsqcup_{I \in {\cal I}}I(\mathsf{p}) \sqsubseteq_{\pi}
J'(\mathsf{p})$, which implies that $J \sqsubseteq_{{\cal
I}_{\mathsf{P}}} J'$.\qed
\end{proof}
In the following we denote with $\perp_{{\cal I}_\mathsf{P}}$ the
greatest lower bound of ${\cal I}_\mathsf{P}$, ie., the
interpretation which for every $\pi$, assigns to each predicate
$\mathsf{p}:\pi$ of $\mathsf{P}$ the element $\perp_{\pi}$.

The properties of monotonicity and continuity of the semantic
valuation function will prove vital:
\begin{lemma}[Monotonicity of Semantics]\label{monotonicity-of-semantics}
Let $\mathsf{P}$ be a program and let $\mathsf{E}:\rho$ be a
positive expression of $\mathsf{P}$. Let $I,J$ be Herbrand
interpretations and $s$ a Herbrand state of $\mathsf{P}$. If $I
\sqsubseteq_{{\cal I}_\mathsf{P}} J$ then $\lsem \mathsf{E}
\rsem_s(I) \sqsubseteq_{\rho} \lsem \mathsf{E} \rsem_s(J)$.
\end{lemma}
The proof of the lemma is given in
Appendix~\ref{appendix-monotonicity}.
\begin{lemma}[Continuity of Semantics]\label{continuity-of-semantics}
Let $\mathsf{P}$ be a program and let $\mathsf{E}$ be a positive
expression of $\mathsf{P}$. Let ${\cal I}$ be a directed set of
Herbrand interpretations and $s$ a Herbrand state of $\mathsf{P}$.
Then, $\lsem \mathsf{E} \rsem_s (\bigsqcup {\cal I}) = \bigsqcup_{I
\in {\cal I}} \lsem \mathsf{E} \rsem_s (I)$.
\end{lemma}
The proof of the lemma is given in
Appendix~\ref{appendix-continuity}.

All the basic properties of first-order logic programming
extend naturally to the higher-order case:
\begin{theorem}[Model Intersection Theorem]
Let $\mathsf{P}$ be a program and ${\cal M}$ a non-empty set of
Herbrand models of $\mathsf{P}$. Then, $\bigsqcap {\cal M}$ is a
Herbrand model for $\mathsf{P}$.
\end{theorem}
\begin{proof}
By Lemma~\ref{herbrand-interpretations-make-a-lattice}, $\bigsqcap
{\cal M}$ is well-defined. Assume that $\bigsqcap {\cal M}$ is not a
model for $\mathsf{P}$. Then, there exists a rule $\mathsf{p}
\leftarrow_{\pi} \mathsf{E}$ in $\mathsf{P}$ and $b_1,\ldots,b_n$ of
the appropriate types such that $(\bigsqcap {\cal
M})(\mathsf{p})\,\,b_1\cdots b_n = 0$ while $\lsem \mathsf{E} \rsem
(\bigsqcap {\cal M})\, b_1 \cdots b_n = 1$. Since for every $M \in
{\cal M}$ we have $\bigsqcap {\cal M} \sqsubseteq M$, using
Lemma~\ref {monotonicity-of-semantics} we conclude that for all $M
\in {\cal M}$, $\lsem \mathsf{E} \rsem (M)\, b_1 \cdots b_n = 1$.
Moreover, since $(\bigsqcap {\cal M})(\mathsf{p})\,\,b_1\cdots b_n =
0$, by Lemma~\ref{herbrand-interpretations-make-a-lattice} we get
that  $(\bigsqcap {\cal M}(\mathsf{p}))\,\,b_1\cdots b_n = 0$. By
Proposition~\ref{monotonic-functions-make-a-complete-lattice} we
conclude that for some $M \in {\cal M}$, $M
(\mathsf{p})\,\,b_1\cdots b_n = 0$. But then there exists $M \in
{\cal M}$ that does not satisfy the rule $\mathsf{p}
\leftarrow_{\pi} \mathsf{E}$, and therefore is not a model of
$\mathsf{P}$ (contradiction).\qed
\end{proof}

It is straightforward to check that every higher-order program
$\mathsf{P}$ has at least one Herbrand model $I$, namely the one
which for every predicate constant $\mathsf{p}$ of $\mathsf{P}$ and
for all basic elements $b_1,\ldots,b_n$ of the appropriate types,
$I(\mathsf{p})\,b_1 \cdots b_n = 1$. Notice that this model
generalizes the familiar idea of ``Herbrand Base'' that is used in
the theory of first-order logic programming.

Since the set of models of a higher-order logic program is
non-empty, the intersection ($\mathit{glb}$) of all Herbrand models
is well-defined, and by the above theorem is a model of the program.
We will denote this model by $M_{\mathsf{P}}$.

\begin{definition}
Let $\mathsf{P}$ be a program. The mapping $T_{\mathsf{P}}: {\cal
I}_\mathsf{P} \rightarrow {\cal I}_\mathsf{P}$ is defined as follows
for every $\mathsf{p}:\pi$ in $\mathsf{P}$ and for every $I \in
{\cal I}_\mathsf{P}$:
$$T_{\mathsf{P}}(I)(\mathsf{p}) = \bigsqcup_{(\mathsf{p} \leftarrow_{\pi}
\mathsf{E}) \in \mathsf{P}}\lsem \mathsf{E}\rsem (I)$$
The mapping $T_{\mathsf{P}}$ will be called the \emph{immediate
consequence operator} for $\mathsf{P}$.
\end{definition}
The fact that $T_{\mathsf{P}}$ is well-defined is verified by the
following lemma:
\begin{lemma}
Let $\mathsf{P}$ be a program and let $\mathsf{p}:\pi$ be a
predicate constant of $\mathsf{P}$. Then, for every $I \in {\cal
I}_{\mathsf{P}}$, $T_{\mathsf{P}}(I)(\mathsf{p}) \in \lsem \pi
\rsem_{U_{\mathsf{P}}}$.
\end{lemma}
\begin{proof}
The result follows directly by the definition of $T_{\mathsf{P}}$,
Lemma~\ref{monotonicity-of-semantics-wrt-state} and the fact that
$\lsem \pi\rsem_{U_{\mathsf{P}}}$ is a complete lattice.\qed
\end{proof}
The key property of $T_{\mathsf{P}}$ is that it is continuous:
\begin{lemma}\label{T_P_continuous}
Let $\mathsf{P}$ be a program. Then the mapping $T_{\mathsf{P}}$ is
continuous.
\end{lemma}
\begin{proof}
Straightforward using Lemma~\ref{continuity-of-semantics}.\qed
\end{proof}
The following property of $T_{\mathsf{P}}$ generalizes the
corresponding well-known property from first-order logic
programming:
\begin{lemma}
Let $\mathsf{P}$ be a program and let $I \in {\cal I}_{\mathsf{P}}$.
Then $I$ is a model of $\mathsf{P}$ if and only if
$T_{\mathsf{P}}(I) \sqsubseteq_{{\cal I}_\mathsf{P}}I$.
\end{lemma}
\begin{proof}
An interpretation $I \in {\cal I}_{\mathsf{P}}$ is a model of
$\mathsf{P}$ iff $\lsem \mathsf{E} \rsem(I) \sqsubseteq_{\pi}
I(\mathsf{p})$ for every clause $\mathsf{p} \leftarrow_{\pi}
\mathsf{E}$ in $\mathsf{P}$ iff $\bigsqcup_{(\mathsf{p}
\leftarrow_{\pi} \mathsf{E}) \in \mathsf{P}}\lsem \mathsf{E}\rsem
(I) \sqsubseteq_{\pi} I(\mathsf{p})$ iff
$T_{\mathsf{P}}(I)(\mathsf{p})  \sqsubseteq_{\pi}I(\mathsf{p})$.\qed
\end{proof}

Define now the following sequence of interpretations:
\[
\begin{array}{lll}
T_{\mathsf{P}} \uparrow 0  & = & \perp_{{\cal I}_{\mathsf{P}}}\\
T_{\mathsf{P}} \uparrow (n+1)  & = & T_{\mathsf{P}}(T_{\mathsf{P}} \uparrow n)\\
T_{\mathsf{P}} \uparrow \omega & = & \bigsqcup \{ T_{\mathsf{P}} \uparrow n \mid n<\omega\}
\end{array}
\]
We have the following theorem (which is entirely analogous to the one
for the first-order case):
\begin{theorem}
Let $\mathsf{P}$ be a program. Then $M_{\mathsf{P}}
=\mathit{lfp}(T_\mathsf{P})= T_{\mathsf{P}}\uparrow \omega$.
\end{theorem}
\begin{proof}
Using exactly the same reasoning as in the first-order case (see for
example the corresponding proof in~\cite{lloyd}).\qed
\end{proof}
%

\section{Proof Procedure}\label{proof-procedure-section}
In this section we propose a sound and complete proof-procedure for
${\cal H}$. One important aspect we initially have to resolve, is
how to represent basic elements (see
Definition~\ref{definition-basic-elements}) in our source language.
In the following subsection we introduce a class of positive
expressions, namely {\em basic expressions}, which are the syntactic
analogues of basic elements. Basic expressions will be used in order
to formalize the notion of {\em answer} (to a given query)
as-well-as in our development of the SLD-resolution for ${\cal H}$.

\subsection{Basic Expressions}\label{subsection-basic-expressions}
As we have already seen, basic elements have played an important
role in the development of the semantics of our higher-order logic
programming language. In order to devise a sound and complete proof
procedure for our language, we first need to find a syntactic
representation for basic elements. Since the definition of basic
elements uses the operator $\sqsubseteq$ (see
Lemma~\ref{algebraicity-composition},
Definition~\ref{step-functions} and
Definition~\ref{definition-basic-elements}), it is not immediately
obvious how one can construct a positive expression whose meaning
coincides with a given basic element. Basic expressions introduced
below, solve this apparent difficulty:
\begin{definition}\label{basic-expressions}
The set of \emph{basic expressions of} ${\cal H}$ is recursively
defined as follows. Every expression of ${\cal H}$ of type $\iota$
is a basic expression of type $\iota$. Every predicate variable of
${\cal H}$ of type $\pi$ is a basic expression of type $\pi$. The
propositional constants $\mathsf{0}$ and $\mathsf{1}$ are basic
expressions of type $o$. A non-empty finite union of expressions
each one of which has the following form, is a basic expression of
type $\rho_1 \rightarrow \cdots \rightarrow \rho_n \rightarrow o$
(where $\mathsf{V}_1:\rho_1,\ldots,\mathsf{V}_n:\rho_n$):
\begin{enumerate}
 \item $\lambda \mathsf{V}_1 . \cdots \lambda \mathsf{V}_n . \mathsf{0}$

 \item $\lambda \mathsf{V}_1 . \cdots \lambda \mathsf{V}_n . (\mathsf{A}_1
       \wedge \cdots \wedge \mathsf{A}_n)$, where each $\mathsf{A}_i$ is either
 \begin{enumerate}
  \item $(\mathsf{V}_i \approx \mathsf{B}_i)$, if $\mathsf{V}_i:\iota$
       and $\mathsf{B}_i:\iota$ is a basic expression where
       $\mathsf{V}_j \not\in FV(\mathsf{B}_i)$ for all $j$, or

  \item the constant $\mathsf{1}$ or $\mathsf{V}_i$,
        if $\mathsf{V}_i:o$, or

  \item the constant $\mathsf{1}$ or $\mathsf{V}_i(\mathsf{B}_{11})\cdots(\mathsf{B}_{1r})\wedge \cdots \wedge
       \mathsf{V}_i(\mathsf{B}_{m1})\cdots(\mathsf{B}_{mr})$, $m>0$, if $type(\mathsf{V}_i)\neq \iota, o$ and
       the $\mathsf{B}_{kl}$'s are basic expressions of the
       appropriate types, where $\mathsf{V}_j \not\in FV(\mathsf{B}_{kl})$
       for all $j,k,l$.
 \end{enumerate}
\end{enumerate}
The $\mathsf{B}_i$ and $\mathsf{B}_{kl}$ above will be called the
{\em basic subexpressions} of $\mathsf{B}$.
\end{definition}
The following example illustrates the ideas behind the above
definition.
\begin{example}
We consider various cases of the above definition:
\begin{itemize}
\item  The terms {\tt a}, {\tt f(a,b)}, {\tt X} and {\tt f(X,h(Y))},
       are basic expressions of type $\iota$.

\item Assume ${\tt X}:\rho$. Then, ${\tt \lambda X.0}$ is a basic expression of type $\rho\rightarrow o$.
      Intuitively, it corresponds to the basic element $\perp_{\rho \rightarrow o}$.

\item Assume ${\tt X}:\iota$. Then, {\tt $\lambda$X.(X$\approx$a)}
      is a basic expression of type $\iota\rightarrow o$.  Intuitively,
      it corresponds to the basic element $({\tt a}\searrow{\tt 1})$ or more simply
      to the finite set $\{{\tt a}\}$.

\item Assume ${\tt X}:\iota$ and ${\tt Y}:\iota$. Then, {\tt
      $\lambda$X.$\lambda$Y.(X$\approx$a)$\wedge$(Y$\approx$b)} is a basic
      expression of type $\iota\rightarrow \iota \rightarrow o$.
      Intuitively, it corresponds to the basic element
      $({\tt a}\searrow ({\tt b}\searrow{\tt 1}))$ or more simply to the singleton binary
      relation $\{{\tt (a,b)}\}$.

\item Assume ${\tt X}:\iota$. Then, {\tt ($\lambda$X.(X$\approx$a))
      $\bigvee_{\iota\rightarrow o}$($\lambda$X.(X$\approx$b))}
      is a basic expression of type $\iota \rightarrow o$. It
      corresponds to the basic element
      $\bigsqcup\{({\tt a}\searrow{\tt 1}),({\tt b}\searrow{\tt 1})\}$, or more
      simply to the finite set $\{{\tt a},{\tt b}\}$.

\item Assume ${\tt Q}:\iota\rightarrow o$. Then,
      {\tt $\lambda$Q.(Q(a)$\wedge$Q(b))}is a basic expression of type
      $(\iota \rightarrow o) \rightarrow o$. Intuitively, it corresponds to the
      basic element $(\bigsqcup\{({\tt a}\searrow 1),({\tt
      b}\searrow 1)\})\searrow 1$. More simply, it corresponds to the set
      of all finite sets that contain both {\tt a} and {\tt b}.

\qed
\end{itemize}
\end{example}

The proof procedure that will be developed later in this section,
relies on a special form of basic expressions:
\begin{definition}\label{definition-basic-template}
The set of \emph{basic templates of} ${\cal H}$ is the subset of the
set of basic expressions of ${\cal H}$ defined as follows:
\begin{itemize}
\item The propositional constants $\mathsf{0}$ and $\mathsf{1}$ are
      basic templates.

\item Every non-empty finite union of basic expressions (of the form
      presented in items 1 and 2 of Definition~\ref{basic-expressions})
      in which all the basic subexpressions involved are {\em distinct} variables,
      is a basic template.
\end{itemize}
The variables mentioned above, will be called {\em template
variables}.
\end{definition}
\begin{example}
Assume in the following expressions that ${\tt X},{\tt Y},{\tt
Z},{\tt W}:\iota$, ${\tt Q},{\tt Q}_1,{\tt Q}_2:\iota\rightarrow o$
and ${\tt R}:((\iota \rightarrow o)\rightarrow o)\rightarrow o$. The
expression {\tt $\lambda$X.(X$\approx$Z)} is a basic template of
type $\iota\rightarrow o$. The expression {\tt
$\lambda$X.$\lambda$Y.(X$\approx$Z)$\wedge$(Y$\approx$W)} is a basic
template of type $\iota\rightarrow \iota \rightarrow o$; the
template variables in this case are {\tt Z} and {\tt W}. The
expression {\tt $\lambda$Q.(Q(Z)$\wedge$Q(W))}is a basic template of
type $(\iota \rightarrow o) \rightarrow o$ with template variables
{\tt Z} and {\tt W}. The expression {\tt
$\lambda$R.(R(Q$_1$)$\wedge$R(Q$_2$))} is a basic template of type
$((\iota \rightarrow o)\rightarrow o)\rightarrow o$ with template
variables {\tt Q}$_1$ and {\tt Q}$_2$.\qed
\end{example}
Notice from the above example that the structure of basic templates
is in general much simpler than that of basic expressions (due to
the fact that a template variable can represent an arbitrary basic
expression of the same type). For this reason, basic templates are
much simpler to enumerate than arbitrary basic expressions.

The following two lemmas establish the connections between basic
elements and basic expressions.
\begin{lemma}
For every basic expression $\mathsf{B}:\rho$, for every Herbrand
interpretation $I$ of ${\cal H}$, and for every Herbrand state $s$,
$\lsem \mathsf{B} \rsem_{s}(I) \in {\cal F}_{U_{\cal H}}(\rho)$.
\end{lemma}
\begin{proof}
The proof is by induction on the type of $\mathsf{B}$. The basis
case is for basic expressions of type $\iota$ and $o$ and holds
trivially. We demonstrate that the lemma holds for basic expressions
of type $\rho = \rho_1 \rightarrow \cdots \rightarrow \rho_n
\rightarrow o$, assuming that it holds for all basic expressions
that have simpler types than $\rho$. If the basic expression is a
predicate variable, the result is immediate; otherwise, we have to
distinguish the following cases:

\vspace{0.2cm} \noindent {\em Case 1:} $\mathsf{B} = \lambda
\mathsf{V}_1. \cdots \lambda \mathsf{V}_n.\mathsf{0}$. Then, the
corresponding basic element in ${\cal F}_{U_{\cal H}}(\rho)$ is the
bottom element of type $\rho_1 \rightarrow \cdots \rightarrow \rho_n
\rightarrow o$ (ie., $\perp_{\rho_1 \rightarrow \cdots \rightarrow
\rho_n \rightarrow o}$).

\vspace{0.2cm} \noindent {\em Case 2:} $\mathsf{B} = \lambda
\mathsf{V}_1 . \cdots \lambda \mathsf{V}_n . (\mathsf{A}_1 \wedge
\cdots \wedge \mathsf{A}_n)$. Then, the corresponding basic element
is the element $b_1 \searrow (b_2 \searrow \cdots \searrow (b_n
\searrow 1)\cdots)$, where the $b_i$ are defined as follows:
\begin{itemize}
\item If $\mathsf{V}_i:\iota$, then by Definition~\ref{basic-expressions},
      $\mathsf{A}_i=(\mathsf{V}_i\approx \mathsf{B}_i)$.
      In this case, $b_i=\lsem \mathsf{B}_i \rsem_s(I)$.

\item If $\mathsf{V}_i:o$ then $\mathsf{A}_i$ is either equal to $\mathsf{1}$ or to
      $\mathsf{V}_i$; in the former case $b_i=0$ and in the latter case $b_i=1$.

\item If $\mathsf{V}_i$ is of any other type then $\mathsf{A}_i$ is either equal
      to $\mathsf{1}$ or to $\mathsf{V}_i(\mathsf{B}_{11})\cdots(\mathsf{B}_{1r})\wedge \cdots \wedge
      \mathsf{V}_i(\mathsf{B}_{m1})\cdots(\mathsf{B}_{mr})$, where $m>0$. In the
      former case it is $b_i=0$; in the latter case
      $b_i = \bigsqcup_{1 \leq j \leq m} (\lsem \mathsf{B}_{j1} \rsem_s(I)
      \searrow (\lsem \mathsf{B}_{j2}\rsem_s(I)\searrow \cdots \searrow (\lsem \mathsf{B}_{jr} \rsem_{s}(I)
      \searrow 1)\cdots ))$.
\end{itemize}

\noindent {\em Case 3:} $\mathsf{B}$ is a finite union of lambda
abstractions. Then, for each term of the finite union we can create
(as above) a basic element. By taking the finite union of these
elements, we create the basic element that corresponds to
$\mathsf{B}$.
It can be easily verified that for every basic expression
$\mathsf{B}$, $\lsem \mathsf{B} \rsem_s(I)$ coincides with the
corresponding basic element defined as above.\qed
\end{proof}
The converse of the above lemma holds, as the following lemma
demonstrates.
\begin{lemma}\label{for_every_b_exists_B}
Let $\rho$ be any argument type and let $b \in {\cal F}_{U_{\cal
H}}(\rho)$. Then, there exists a closed basic expression
$\mathsf{B}:\rho$ such that for every Herbrand interpretation $I$,
$\lsem \mathsf{B} \rsem(I) = b$.
\end{lemma}
\begin{proof}
The proof is by induction on the structure of argument types. The
basis case is for argument types $\iota$ and $o$, and holds
trivially. We demonstrate that the lemma holds for type $\rho =
\rho_1 \rightarrow \cdots \rightarrow \rho_n \rightarrow o$,
assuming that it holds for all subtypes of $\rho$. Assume now that
$b$ is a basic element of type $\rho$, consisting of a finite union
of step functions.

If the union is empty, then $\mathsf{B} = \lambda
\mathsf{V}_1.\cdots \lambda \mathsf{V}_n.\mathsf{0}$. Assume now
that the union is non-empty. Then, the basic expression
corresponding to $b$ will simply be the union of the basic
expressions corresponding to the step functions that comprise $b$.

Let $b_1 \searrow (b_2\searrow  \cdots \searrow (b_n \searrow
1)\cdots )$ be one of the step functions that constitute $b$. We
create the basic expression:
$\mathsf{B} = \lambda \mathsf{V}_1. \cdots \lambda
\mathsf{V}_n.(\mathsf{A}_1\wedge \cdots \mathsf{A}_n)$
where each $\mathsf{A}_i$ can be created as follows:
\begin{itemize}
\item If $b_i$ is of type $\iota$ and $b_i = \mathsf{t}\in U_{\cal H}$,
      then $\mathsf{A}_i=(\mathsf{V}_i \approx \mathsf{t})$.

\item If $b_i$ is of type $o$ and $b_i = 0$, then
      $\mathsf{A}_i=\mathsf{1}$.

\item If $b_i$ is of type $o$ and $b_i = 1$, then $\mathsf{A}_i=\mathsf{0}$.

\item Otherwise, $b_i$ is a finite union of $m>0$ basic
      elements of the form $b_{j1} \searrow (b_{j2} \searrow \cdots \searrow (b_{jr}
      \searrow 1)\cdots)$, $1\leq j \leq m$. Then, $\mathsf{A}_i =
      \mathsf{V}_i(\mathsf{B}_{11})\cdots(\mathsf{B}_{1r})\wedge \cdots \wedge
      \mathsf{V}_i(\mathsf{B}_{m1})\cdots(\mathsf{B}_{mr})$,
      where $\mathsf{B}_{j1},\ldots,\mathsf{B}_{jm}$ are the expressions that correspond (by the
      induction hypothesis) to $b_{j1},\ldots, b_{jm}$.
\end{itemize}
It is easy to verify that the resulting basic expression
$\mathsf{B}$ satisfies $\lsem \mathsf{B} \rsem(I) = b$.\qed
\end{proof}
The above two lemmas suggest that basic expressions are the syntactic analogues
of basic elements.
\subsection{Substitutions and Unifiers}
Substitutions are vital in the development of the proof procedure
for ${\cal H}$:
\begin{definition}
A \emph{substitution} $\theta$ is a finite set of the form
$\{\mathsf{V}_1/\mathsf{E}_1,\ldots,\mathsf{V}_n/\mathsf{E}_n\}$,
where the $\mathsf{V}_i$'s are different argument variables of
${\cal H}$ and each $\mathsf{E}_i$ is a positive expression of
${\cal H}$ having the same type as $\mathsf{V}_i$. We write
$dom(\theta) = \{\mathsf{V}_1,\ldots,\mathsf{V}_n\}$ and
$range(\theta) = \{\mathsf{E}_1,\ldots,\mathsf{E}_n\}$. A
substitution is called \emph{basic} if all $\mathsf{E}_i$ are basic
expressions. A substitution is called {\em zero-order}, if
$type(\mathsf{V}_i) = \iota$, for all $i\in \{1,\ldots,n\}$ (notice
that every zero-order substitution is also basic). The substitution
corresponding to the empty set will be called the \emph{identity
substitution} and will be denoted by $\epsilon$.
\end{definition}
We are now ready to define what it means to apply a substitution
$\theta$ to an expression $\mathsf{E}$. Such definitions are usually
complicated by the fact that one has to often rename the bound
variable before applying $\theta$ to the body of a lambda
abstraction. In order to simplify matters, we follow the simple
approach suggested in~\cite{Bar84}[pages 26-27], which consists of
the following two conventions:
\begin{itemize}
\item {\bf The $\alpha$-congruence convention}: Expressions that are
      {\em $\alpha$-congruent} will be considered identical (expression
      $\mathsf{E}_1$ is $\alpha$-congruent with expression $\mathsf{E}_2$ if
      $\mathsf{E}_2$ results from $\mathsf{E}_1$ by a series of changes of bound
      variables). For example, {\tt $\lambda$Q.Q(a)} is $\alpha$-congruent to
      {\tt $\lambda$R.R(a)}.

\item {\bf The variable convention}: If expressions $\mathsf{E}_1,\ldots,\mathsf{E}_n$
      occur in a certain mathematical context (eg., definition, proof), then in these
      expressions all bound variables are chosen to be different from the free
      variables.
\end{itemize}
Using the variable convention, we have the following simple
definition:
\begin{definition}
Let $\theta$ be a substitution and let $\mathsf{E}$ be a positive
expression. Then, $\mathsf{E}\theta$ is an expression obtained from
$\mathsf{E}$ as follows:
\begin{itemize}
\item $\mathsf{E} \theta = \mathsf{E}$, if $\mathsf{E}$ is
$\mathsf{0}$, $\mathsf{1}$, $\mathsf{c}$, or $\mathsf{p}$.
\item $\mathsf{V}\theta =\theta(\mathsf{V})$ if $\mathsf{V}\in
dom(\theta)$; otherwise, $\mathsf{V}\theta =\mathsf{V}$.
\item $\mathsf{(f}\,\, \mathsf{E}_1 \cdots \mathsf{E}_n)\theta =
\mathsf{(f}\,\, \mathsf{E}_1\theta \cdots \mathsf{E}_n\theta)$.
\item $\mathsf{(E}_1\mathsf{E}_2)\theta = \mathsf{(E}_1\theta\,
\mathsf{E}_2\theta)$.
\item $(\lambda \mathsf{V}.\mathsf{E}_1)\theta = (\lambda
\mathsf{V}.(\mathsf{E}_1\theta))$.
\item $\mathsf{(E}_1\bigvee_{\pi} \mathsf{E}_2)\theta = \mathsf{(E}_1\theta
\bigvee_{\pi} \mathsf{E}_2\theta)$.
\item $\mathsf{(E}_1\bigwedge_{\pi} \mathsf{E}_2)\theta = \mathsf{(E}_1\theta
\bigwedge_{\pi} \mathsf{E}_2\theta)$.
\item $\mathsf{(E}_1\approx \mathsf{E}_2)\theta = \mathsf{(E}_1\theta
\approx \mathsf{E}_2\theta)$.
\item $(\exists\mathsf{V}\,\mathsf{E}_1)\theta = (\exists\mathsf{V}\,
(\mathsf{E}_1\theta))$.
\end{itemize}
\end{definition}

Notice that in the case of lambda abstraction (and similarly in the
case of existential quantification), it is not needed to say
``provided $\mathsf{V} \not\in FV(range(\theta))$ and $\mathsf{V}
\not\in dom(\theta)$''. By the variable convention this is the case.
\begin{definition}\label{composition-of-substitutions}
Let $\theta=
\{\mathsf{V}_1/\mathsf{E}_1,\ldots,\mathsf{V}_m/\mathsf{E}_m\}$ and
$\sigma=
\{\mathsf{V}'_1/\mathsf{E}'_1,\ldots,\mathsf{V}'_n/\mathsf{E}'_n\}$
be substitutions. Then the composition $\theta\sigma$ of $\theta$
and $\sigma$ is the substitution obtained from the set
$$\{\mathsf{V}_1/\mathsf{E}_1\sigma,\ldots,\mathsf{V}_m/\mathsf{E}_m\sigma,
\mathsf{V}'_1/\mathsf{E}'_1,\ldots,\mathsf{V}'_n/\mathsf{E}'_n\}$$
by deleting any $\mathsf{V}_i/\mathsf{E}_i\sigma$ for which
$\mathsf{V}_i = \mathsf{E}_i\sigma$ and deleting any
$\mathsf{V}'_j/\mathsf{E}'_j$ for which $\mathsf{V}'_j \in
\{\mathsf{V}_1,\ldots,\mathsf{V}_m\}$.
\end{definition}
The following proposition is easy to establish:
\begin{proposition}
Let $\theta,\sigma$ and $\gamma$ be substitutions. Then:
\begin{enumerate}
\item $\theta \epsilon = \epsilon \theta = \theta$.
\item For all positive expressions $\mathsf{E}$,
      $(\mathsf{E} \theta) \sigma = \mathsf{E} (\theta \sigma)$.
\item $(\theta \sigma) \gamma = \theta (\sigma \gamma)$.
\end{enumerate}
\end{proposition}
We will use the notions of unifier and most general unifier, which
in our case have exactly the same meaning as in the case of
classical (first-order) logic programming:
\begin{definition}
Let $S$ be a set of terms of ${\cal H}$ (ie., expressions of type
$\iota$). A zero-order substitution $\theta$ will be called a
\emph{unifier} of the expressions in $S$ if the set $S\theta =
\{\mathsf{E}\theta \mid \mathsf{E} \in S\}$ is a singleton. The
substitution $\theta$ will be called a \emph{most general unifier of
$S$} (denoted by $\mathit{mgu}(S)$), if for every unifier $\sigma$
of the expressions in $S$, there exists a zero-order substitution
$\gamma$ such that $\sigma = \theta \gamma$.
\end{definition}
We now have the following {\em Substitution Lemma} (see for
example~\cite{Ten91} for a corresponding lemma in the case of
functional programming). The Substitution Lemma shows that given a
basic substitution $\theta$, the meaning of $\mathsf{E}\theta$ is
that of $\mathsf{E}$ in a certain state definable from $\theta$. The
lemma will be later used in the proof of soundness of the proposed
proof procedure.
\begin{lemma}[Substitution Lemma]\label{substitution-lemma}
Let $\mathsf{P}$ be a program, let $I$ be an interpretation of
$\mathsf{P}$ and let $s$ be a state over the domain of $I$. Let
$\theta$ be a basic substitution and $\mathsf{E}$ be a positive
expression. Then, $\lsem \mathsf{E}\theta\rsem_s(I) = \lsem
\mathsf{E}\rsem_{s'}(I)$, where $s'(\mathsf{V}) = \lsem
\theta(\mathsf{V})\rsem_s(I)$ if $\mathsf{V}\in dom(\theta)$ and
$s'(\mathsf{V}) = s(\mathsf{V})$, otherwise.
\end{lemma}
\begin{proof}
By structural induction on $\mathsf{E}$.\qed
\end{proof}
The following lemmas, that also involve the notion of substitution,
can be easily demonstrated and will prove useful in the sequel.
\begin{lemma}\label{composition-of-basic}
Let $\theta_1,\ldots,\theta_n$ be basic substitutions. Then,
$\theta_1\cdots \theta_n$ is also a basic substitution.
\end{lemma}
\begin{proof}
By induction on $n$ and using Definitions~\ref{basic-expressions}
and~\ref{composition-of-substitutions}.\qed
\end{proof}
\begin{lemma}\label{beta-reduction-lemma}
Let $\mathsf{P}$ be a program, let $I$ be an interpretation of
$\mathsf{P}$ and let $s$ be a state over the domain of $I$. Let
$\lambda\mathsf{V}.\mathsf{E}_1$ and $\mathsf{E}_2$ be positive
expressions of type $\rho \rightarrow \pi$ and $\rho$ respectively.
Then, $\lsem (\lambda \mathsf{V}.\mathsf{E}_1)\mathsf{E}_2\rsem_s(I)
= \lsem \mathsf{E}_1\{\mathsf{V}/\mathsf{E}_2\}\rsem_s(I)$.
\end{lemma}
\begin{proof}
By structural induction on $\mathsf{E}_1$.\qed
\end{proof}
\begin{lemma}\label{lemma-transform-state-to-subst}
Let $\mathsf{P}$ be a program, $I$ a Herbrand interpretation of
$\mathsf{P}$ and $s$ a Herbrand state. Let $\mathsf{E}$ be a
positive expression. Then, there exists a basic substitution
$\theta$ such that $\lsem \mathsf{E}\rsem_s(I) = \lsem
\mathsf{E}\theta\rsem_{s'}(I)$ for every Herbrand state $s'$.
\end{lemma}
\begin{proof}
Define $\theta$ such that if $\mathsf{V}\in FV(\mathsf{E})$,
$\theta(\mathsf{V}) = \mathsf{B}$, where $\mathsf{B}$ is a closed
basic expression such that $\lsem \mathsf{B}\rsem(I) =
s(\mathsf{V})$ (the existence of such a $\mathsf{B}$ is ensured by
Lemma~\ref{for_every_b_exists_B}). The lemma follows by a structural
induction on $\mathsf{E}$.\qed
\end{proof}
It is important to note that in the rest of the paper, the
substitutions that we will use will be {\em basic} ones (unless
otherwise stated). Actually, the only place where a non-basic
substitution will be needed, is when we perform a {\em
$\beta$-reduction} step (see for example the rule for $\lambda$ in
the forthcoming Definition~\ref{derives-in-one-step}).
\subsection{SLD-Resolution}
We now proceed to define the notions of \emph{answer} and
\emph{correct answer}.
\begin{definition}\label{answer}
Let $\mathsf{P}$ be a program and $\mathsf{G}$ a goal. An
\emph{answer} for $\mathsf{P} \cup \{\mathsf{G}\}$ is a basic
substitution for (certain of the) free variables of $\mathsf{G}$.
\end{definition}
\begin{definition}\label{correct-answer}
Let $\mathsf{P}$ be a program, $\mathsf{G}= \leftarrow \mathsf{A}$ a
goal clause and $\theta$ an answer for $\mathsf{P} \cup
\{\mathsf{G}\}$. We say that $\theta$ is \emph{a correct answer} for
$\mathsf{P} \cup \{\mathsf{G}\}$ if for every model $M$ of
$\mathsf{P}$ and for every state $s$ over the domain of $M$, $\lsem
\mathsf{A}\theta \rsem_s(M)=1$.
\end{definition}
\begin{definition}\label{derives-in-one-step}
Let $\mathsf{P}$ be a program and let $\mathsf{G}=\leftarrow
\mathsf{A}$ and $\mathsf{G}'=\leftarrow \mathsf{A}'$ be goal
clauses. Then, we say that $\mathsf{A}'$ is derived in one step from
$\mathsf{A}$ using basic substitution $\theta$ (or equivalently that
$\mathsf{G}'$ is derived in one step from $\mathsf{G}$ using
$\theta$), and we denote this fact by $\mathsf{A}
\stackrel{\theta}{\rightarrow} \mathsf{A}'$ (respectively,
$\mathsf{G} \stackrel{\theta}{\rightarrow} \mathsf{G}'$) if one of
the following conditions applies:
\begin{enumerate}
\item $\mathsf{p} \,\, \mathsf{E}_1 \cdots \mathsf{E}_n
\stackrel{\epsilon}{\rightarrow} \mathsf{E} \,\, \mathsf{E}_1 \cdots
\mathsf{E}_n$, where $\mathsf{p} \leftarrow_{\pi} \mathsf{E}$ is a
rule in $\mathsf{P}$.

\item $\mathsf{Q} \,\, \mathsf{E}_1 \cdots \mathsf{E}_n
\stackrel{\theta}{\rightarrow} (\mathsf{Q} \,\, \mathsf{E}_1 \cdots
\mathsf{E}_n)\theta$, where $\theta= \{\mathsf{Q}/\mathsf{B}_t\}$
and $\mathsf{B}_t$ a basic template.

\item $(\lambda \mathsf{V}.\mathsf{E})\,\, \mathsf{E}_1 \cdots
\mathsf{E}_n \stackrel{\epsilon}{\rightarrow}
(\mathsf{E}\{\mathsf{V}/\mathsf{E}_1\})\, \mathsf{E}_2 \cdots
\mathsf{E}_n$.

\item $(\mathsf{E}' \bigvee_{\pi} \mathsf{E}'') \,\, \mathsf{E}_1 \cdots \mathsf{E}_n
\stackrel{\epsilon}{\rightarrow} \mathsf{E}' \,\, \mathsf{E}_1
\cdots \mathsf{E}_n$.

\item $(\mathsf{E}' \bigvee_{\pi} \mathsf{E}'') \,\, \mathsf{E}_1 \cdots \mathsf{E}_n
\stackrel{\epsilon}{\rightarrow} \mathsf{E}'' \,\, \mathsf{E}_1
\cdots \mathsf{E}_n$.

\item $(\mathsf{E}' \bigwedge_{\pi} \mathsf{E}'') \,\, \mathsf{E}_1 \cdots \mathsf{E}_n
\stackrel{\epsilon}{\rightarrow} (\mathsf{E}'\, \mathsf{E}_1 \cdots
\mathsf{E}_n) \wedge (\mathsf{E}''\, \mathsf{E}_1 \cdots
\mathsf{E}_n)$, where $\pi \neq o$.

\item $(\mathsf{E}_1 \wedge \mathsf{E}_2)
\stackrel{\theta}{\rightarrow} (\mathsf{E}'_1 \wedge (\mathsf{E}_2
\theta))$, if $\mathsf{E}_1 \stackrel{\theta}{\rightarrow}
\mathsf{E}'_1$.

\item $(\mathsf{E}_1 \wedge \mathsf{E}_2)
\stackrel{\theta}{\rightarrow} ((\mathsf{E}_1\theta) \wedge
\mathsf{E}'_2)$, if $\mathsf{E}_2 \stackrel{\theta}{\rightarrow}
\mathsf{E}'_2$.

\item $(\Box \wedge \mathsf{E})
\stackrel{\epsilon}{\rightarrow} \mathsf{E}$

\item $(\mathsf{E} \wedge \Box) \stackrel{\epsilon}{\rightarrow}
\mathsf{E}$

\item $(\mathsf{E}_1 \approx \mathsf{E}_2)
\stackrel{\theta}{\rightarrow} \Box$, where $\theta$ is an mgu of
$\mathsf{E}_1$ and $\mathsf{E}_2$.

\item $(\exists \mathsf{V} \,\mathsf{E})
\stackrel{\epsilon}{\rightarrow} \mathsf{E}$

\end{enumerate}
Moreover, we write
$\mathsf{A}\stackrel{\theta}{\twoheadrightarrow}\mathsf{A}'$ if
$\mathsf{A}=\mathsf{A}_0\stackrel{\theta_1}{\rightarrow}\mathsf{A}_1
\stackrel{\theta_2}{\rightarrow}\cdots
\stackrel{\theta_n}{\rightarrow} \mathsf{A}_n=\mathsf{A}'$,
$n\geq1$, where $\theta=\theta_1\cdots\theta_n$ (and similarly for
$\mathsf{G}\stackrel{\theta}{\twoheadrightarrow}\mathsf{G}'$).
\end{definition}
\begin{definition}\label{SLD-derivation}
Let $\mathsf{P}$ be a program and $\mathsf{G}$ a goal. An
\emph{SLD-derivation} of $\mathsf{P} \cup \{\mathsf{G}\}$ is a
(finite or infinite) sequence $\mathsf{G}_0 = \mathsf{G},
\mathsf{G}_1,\ldots$ of goals and a sequence
$\theta_1,\theta_2,\ldots$ of basic substitutions such that each
$\mathsf{G}_{i+1}$ is derived in one step from $\mathsf{G}_i$ using
$\theta_{i+1}$. Moreover, for all $i$, if $\theta_i =
\{\mathsf{V}/\mathsf{B}_t\}$, then the free variables of
$\mathsf{B}_t$ are disjoint from all the variables that have already
appeared in the derivation up to $\mathsf{G}_{i-1}$.
\end{definition}
\begin{definition}
Let $\mathsf{P}$ be a program and $\mathsf{G}$ a goal. Assume that
$\mathsf{P} \cup \{\mathsf{G}\}$ has a finite SLD-derivation
$\mathsf{G}_0 = \mathsf{G}, \mathsf{G}_1,\ldots,\mathsf{G}_n$ with
basic substitutions $\theta_1,\ldots,\theta_n$, such that
$\mathsf{G}_n=\Box$. Then, we will say that $\mathsf{P} \cup
\{\mathsf{G}\}$ has an \emph{SLD-refutation of length $n$ using
basic substitution $\theta=\theta_1\cdots\theta_n$}.
\end{definition}
\begin{definition}
Let $\mathsf{P}$ be a program, $\mathsf{G}$ a goal and assume that
$\mathsf{P}\cup\{\mathsf{G}\}$ has an SLD-refutation using basic
substitution $\theta$. Then, a {\em computed answer} $\sigma$ for
$\mathsf{P} \cup \{ \mathsf{G} \}$ is the basic substitution
obtained by restricting $\theta$ to the free variables of
$\mathsf{G}$.
\end{definition}

\begin{example}
Consider the program of Example~\ref{closure-example}. An
SLD-refutation of the goal $\leftarrow$ {\tt closure Q a b} is given
below (where we have omitted certain simple steps involving lambda
abstractions):
\[
\begin{array}[b]{@{}ll@{}}
\mbox{\tt closure Q a b} & \theta_1 = \epsilon\\
\mbox{\tt ($\lambda$R.$\lambda$X.$\lambda$Y.(R X Y)) Q a b} & \theta_2 = \epsilon \\
\mbox{\tt Q a b} & \theta_3 = \{ {\tt Q} / \mbox{\tt ($\lambda$X.$\lambda$Y.(X$\approx$X$_0$)$\wedge$(Y$\approx$Y$_0$))}\} \\
\mbox{\tt ($\lambda$X.$\lambda$Y.(X$\approx$X$_0$)$\wedge$(Y$\approx$Y$_0$)) a b} & \theta_4 = \epsilon\\
\mbox{\tt (a$\approx$X$_0$)$\wedge$(b$\approx$Y$_0$)} & \theta_5 = \{{\tt X}_0 /{\tt a} \}\\
\mbox{\tt $\square$$\wedge$(b$\approx$Y$_0$)} & \theta_6 = \epsilon\\
\mbox{\tt (b$\approx$Y$_0$) } & \theta_7 = \{ {\tt Y}_0 / {\tt b} \}\\
\square & \\
\end{array}
\]
If we restrict the composition $\theta_1\cdots \theta_7$ to the free
variables of the goal, we get the computed answer $\sigma_1 = \{
{\tt Q} / \mbox{\tt
$\lambda$X.$\lambda$Y.(X$\approx$a)$\wedge$(Y$\approx$b)} \}$.
Intuitively, $\sigma_1$ assigns to {\tt Q} the relation $\{({\tt
a},{\tt b})\}$ (for which the query is obviously true). Notice that
by substituting {\tt Q} with different basic templates, one can get
answers that are ``similar'' to the above one, such as for example
$\{({\tt a},{\tt b}),({\tt Z1},{\tt Z2})\}$ or $\{({\tt a},{\tt
b}),({\tt Z1},{\tt Z2}),({\tt Z3},{\tt Z4}) \}$, and so on. Answers
of this type are in some sense ``represented'' by the answer
$\{({\tt a},{\tt b})\}$. Actually, one can easily optimize the proof
procedure so as to avoid enumerating such superfluous answers (see
the discussion in Section~\ref{future-work}).

However, there exist other answers to our original query that are
genuinely different from $\{({\tt a},{\tt b})\}$ and can be obtained
by making different clause choices. For example, another answer to
our query is $\sigma_2 = \{ {\tt Q} / \mbox{\tt
($\lambda$X.$\lambda$Y.(X$\approx$a)$\wedge$(Y$\approx$Z))}
\bigvee_{\pi} \mbox{\tt
($\lambda$X.$\lambda$Y.(X$\approx$Z)$\wedge$(Y$\approx$b))}\}$,
which corresponds to the relations of the form $\{({\tt a},{\tt
Z}),({\tt Z},{\tt b})\}$, for every {\tt Z} in the Herbrand
universe. Similarly, one can get the answer $\{({\tt a},{\tt
Z1}),({\tt Z1},{\tt Z2}),({\tt Z2},{\tt b})\}$, and so on.

In other words, we observe that by performing different choices in
the selection of a basic template for {\tt Q} and making an
appropriate use of the two rules of the program for {\tt closure},
we get an infinite (but countable) number of computed answers to our
original query.\qed
\end{example}

\subsection{Soundness of SLD-resolution}
In this subsection we establish the soundness of the SLD-resolution
proof procedure. The following lemmas are very useful in the proof
of the soundness theorem:
\begin{lemma}\label{less-or-equal-lemma}
Let $\mathsf{P}$ be a program, let $I$ be an interpretation of
$\mathsf{P}$ and let $s$ be a state over the domain of $I$. Let
$\mathsf{E}_1$ and $\mathsf{E}_2$ be positive expressions of type
$\rho \rightarrow \pi$ and $\mathsf{E}$ expression of  type $\rho$.
If $\lsem \mathsf{E}_1 \rsem_s(I) \sqsubseteq_{\rho \rightarrow \pi}
\lsem \mathsf{E}_2 \rsem_s(I)$, then $\lsem (\mathsf{E}_1\,
\mathsf{E}) \rsem_s(I) \sqsubseteq_{\pi} \lsem (\mathsf{E}_2\,
\mathsf{E}) \rsem_s(I)$.
\end{lemma}
\begin{proof}
Straightforward using the definition of application.\qed
\end{proof}
\begin{lemma}\label{G-greater-or-equal-G'}
Let $\mathsf{P}$ be a program, let $\mathsf{G}=\leftarrow\mathsf{A}$
and $\mathsf{G}'=\leftarrow\mathsf{A}'$ be goals and let $\theta$ be
a basic substitution such that
$\mathsf{A}\stackrel{\theta}{\rightarrow} \mathsf{A}'$. Then, for
every model $M$ of $\mathsf{P}$ and for every state $s$ over the
domain of $M$, it holds that $\lsem \mathsf{A}\theta\rsem_s(M)
\sqsupseteq \lsem \mathsf{A}'\rsem_s(M)$.
\end{lemma}
\begin{proof}
First, observe that in all cases $\mathsf{A}$ is of the form
$\mathsf{E}\ \mathsf{E}_1 \cdots \mathsf{E}_k$, $k\geq 0$, where
$\mathsf{E}$ is an expression of predicate type. We perform a
structural induction on $\mathsf{E}$.

\vspace{0.1cm}\noindent{\em Induction Basis:} We distinguish two
cases, namely $\mathsf{E} = \mathsf{p}$ and $\mathsf{E} =
\mathsf{Q}$. For the first case it suffices to show that $\lsem
(\mathsf{p}\, \mathsf{E}_1\cdots \mathsf{E}_k)\theta\rsem_s(M)
\sqsupseteq \lsem (\mathsf{E}_{\mathsf{p}}\, \mathsf{E}_1\cdots
\mathsf{E}_k)\theta \rsem_s(M)$, where $\theta=\epsilon$ and
$\mathsf{p} \leftarrow_{\pi} \mathsf{E}_{\mathsf{p}}$ is a clause in
$\mathsf{P}$. This follows easily by the fact that $M$ is a model of
$\mathsf{P}$ and using Lemma~\ref{less-or-equal-lemma}. The second
case is trivial.

\vspace{0.1cm}\noindent{\em Induction Step:} We examine the two most
interesting cases (the rest are straightforward):

\vspace{0.1cm}\noindent {\em Case 1:} $\mathsf{E} =
(\lambda\mathsf{V}.\mathsf{E}')$. In this case $\theta$ is the empty
substitution, and therefore it suffices to show that $\lsem (\lambda
\mathsf{V}.\mathsf{E}')\, \mathsf{E}_1\cdots \mathsf{E}_k\rsem_s(M)
\sqsupseteq \lsem \mathsf{E}'\{\mathsf{V}/\mathsf{E}_1\}\,
\mathsf{E}_2\cdots \mathsf{E}_k\rsem_s(M)$. By
Lemma~\ref{beta-reduction-lemma} we have that $\lsem (\lambda
\mathsf{V}.\mathsf{E}')\,\mathsf{E}_1\rsem_s(M) = \lsem
\mathsf{E}'\{\mathsf{V}/\mathsf{E}_1\}\rsem_s(M)$, and the result
follows by Lemma~\ref{less-or-equal-lemma}.

\vspace{0.1cm}\noindent {\em Case 2:} $\mathsf{E} = (\mathsf{E}'
\wedge \mathsf{E}'')$. Moreover, assume that $\mathsf{E}'
\stackrel{\theta}{\rightarrow} \mathsf{E}'_1$. Then, $(\mathsf{E}'
\wedge \mathsf{E}'')$ derives in one step the expression
$(\mathsf{E}'_1 \wedge (\mathsf{E}''\theta))$. It suffices to show
that $\lsem (\mathsf{E}' \wedge \mathsf{E}'')\theta\rsem_s(M)
\sqsupseteq \lsem \mathsf{E}'_1 \wedge
(\mathsf{E}''\theta)\rsem_s(M)$, or equivalently that $\lsem
(\mathsf{E}'\theta) \wedge (\mathsf{E}''\theta)\rsem_s(M)
\sqsupseteq \lsem \mathsf{E}'_1 \wedge
(\mathsf{E}''\theta)\rsem_s(M)$. But this holds since by the
induction hypothesis we have that $\lsem \mathsf{E}'\theta\rsem_s(M)
\sqsupseteq \lsem \mathsf{E}'_1\rsem_s(M)$.\qed
\end{proof}
\begin{lemma}\label{G-greater-or-equal-G'-n-times}
Let $\mathsf{P}$ be a program and $\mathsf{G}=\leftarrow \mathsf{A}$
be a goal. Let $\mathsf{G}_0
=\mathsf{G},\mathsf{G}_1=\leftarrow\mathsf{A}_1,\ldots,\mathsf{G}_n=\leftarrow\mathsf{A}_n$
be an SLD-refutation of length $n$ using basic substitutions
$\theta_1,\ldots,\theta_n$.  Then, for every model $M$ of
$\mathsf{P}$ and for every state $s$ over the domain of $M$, $\lsem
\mathsf{A}\theta_1\cdots\theta_n\rsem_s(M) \sqsupseteq \lsem
\mathsf{A}_n\rsem_s(M)$.
\end{lemma}
\begin{proof}
Using Lemma~\ref{G-greater-or-equal-G'},
Lemma~\ref{substitution-lemma} and induction on $n$.\qed
\end{proof}
\begin{theorem}[Soundness]\label{soundness}
Let $\mathsf{P}$ be a program and $\mathsf{G}=\leftarrow \mathsf{A}$
a goal. Then, every computed answer for $\mathsf{P} \cup
\{\mathsf{G}\}$ is a correct answer for $\mathsf{P} \cup
\{\mathsf{G}\}$.
\end{theorem}
\begin{proof}
The result is a direct consequence of
Lemma~\ref{G-greater-or-equal-G'-n-times} for
$\mathsf{G}_n=\Box$.\qed
\end{proof}
\subsection{Completeness of SLD-resolution}
In order to establish the completeness of the proposed
SLD-resolution, we need to first demonstrate a result that is
analogous to the {\em lifting lemma} of the first-order case
(see~\cite{lloyd}). We first state (and prove in the appendix) a
more technical lemma, which has as a special case the desired
lifting lemma.

In the rest of this subsection, whenever we refer to a
``substitution'' we mean a ``basic substitution''.
\begin{lemma}\label{pre-lifting-lemma} Let $\mathsf{P}$ be a program,
$\mathsf{G}$ a goal and $\theta$ a substitution. Suppose that there
exists an SLD-refutation of $\mathsf{P}\cup\{\mathsf{G}\theta\}$
using substitution $\sigma$. Then, there exists an SLD-refutation of
$\mathsf{P}\cup\{\mathsf{G}\}$ using a substitution $\delta$, where
for some substitution $\gamma$ it holds that $\delta\gamma \supseteq
\theta\sigma$ and $dom(\delta\gamma-\theta\sigma)$ is a (possibly
empty) set of template variables that are introduced during the
refutation of $\mathsf{P}\cup\{\mathsf{G}\}$.
\end{lemma}
The proof of the above lemma is by a straightforward (but tedious)
induction on the length of the SLD-refutation of
$\mathsf{P}\cup\{\mathsf{G}\theta\}$, and is given in
Appendix~\ref{appendix-lifting}.
\begin{lemma}[Lifting Lemma]\label{lifting-lemma}
Let $\mathsf{P}$ be a program, $\mathsf{G}$ a goal and $\theta$ a
substitution. Suppose that there exists an SLD-refutation of
$\mathsf{P}\cup\{\mathsf{G}\theta\}$ using substitution $\sigma$.
Then, there exists an SLD-refutation of
$\mathsf{P}\cup\{\mathsf{G}\}$ using a substitution $\delta$, where
for some substitution $\gamma$ it holds that $\mathsf{G}\delta\gamma
= \mathsf{G}\theta\sigma$.
\end{lemma}
\begin{proof}
By Lemma~\ref{pre-lifting-lemma}, $\delta\gamma$ and $\theta\sigma$
differ only in template variables that are introduced during the
refutation. By the restriction mentioned in
Definition~\ref{SLD-derivation}, these variables are different from
the variables in the goal $\mathsf{G}$. Therefore, $\delta\gamma$
and $\theta\sigma$ agree on the expressions they assign to the free
variables of $\mathsf{G}$.\qed
\end{proof}

Notice that the above lifting lemma differs slightly from the
corresponding lemma for classical logic programming, where we
actually have the equality $\delta\gamma = \theta\sigma$. This
difference is due to the existence of template variables in the
higher-order resolution proof procedure. Of course, if we restrict
the higher-order proof procedure to apply to first-order logic
programs, then it behaves like classical SLD-resolution and the
usual lifting lemma holds.
\begin{example}
Consider any program $\mathsf{P}$ of our higher-order language and
consider the goal clause $\mathsf{G} =\leftarrow \mbox{\tt R(Z)}$,
where {\tt Z} is of type $\iota$ and {\tt R} of type $\iota
\rightarrow o$. Let $\theta = \{{\tt R}/ \lambda \mbox{\tt X.(X
$\approx$ a)}, \mbox{\tt Z}/\mbox{\tt a}\}$. Then, $\mathsf{G}\theta
= \leftarrow \mbox{\tt ($\lambda$X.(X $\approx$ a))(a)}$. We have
the following SLD-refutation:
$$\mbox{\tt ($\lambda$X.(X $\approx$ a))(a)}
\stackrel{\epsilon}\rightarrow \mbox{\tt (a $\approx$ a)}
\stackrel{\epsilon}\rightarrow \Box$$
Therefore, $\mathsf{G}\theta$ has an SLD-refutation with
substitution $\sigma = \epsilon$. On the other hand, we have the
following SLD-refutation of $\mathsf{G}$:
$$\mbox{\tt R(Z)} \stackrel{\{\mbox{\tt R} /
\mbox{\tt $\lambda$X.(X $\approx$ X$_0$)}\}}\rightarrow \mbox{\tt
($\lambda$X.(X $\approx$ X$_0$))(Z)} \stackrel{\epsilon}\rightarrow
\mbox{\tt (Z $\approx$ X$_0$)} \stackrel{\{\mbox{\tt X$_0$} /
\mbox{\tt Z}\}}\rightarrow \Box$$
Therefore, $\mathsf{G}$ has an SLD-refutation with substitution
$\delta$ which is equal to the composition of the substitutions
$\{\mbox{\tt R} / \mbox{\tt $\lambda$X.(X $\approx$ X$_0$)}\}$,
$\epsilon$ and $\{\mbox{X$_0$/\mbox{\tt Z}}\}$, ie., $\delta =
\{\mbox{\tt R} / \mbox{\tt $\lambda$X.(X $\approx$ Z)}, \mbox{\tt
X}_0/\mbox{\tt Z}\}$. Let $\gamma = \{\mbox{\tt Z}/\mbox{\tt a}\}$.
Then, $\delta\gamma = \{\mbox{\tt R} / \mbox{\tt $\lambda$X.(X
$\approx$ a)}, \mbox{\tt X}_0/\mbox{\tt a}, \mbox{\tt Z}/\mbox{\tt
a}\}$ while $\theta = \{{\tt R}/ \lambda \mbox{\tt X.(X $\approx$
a)}, \mbox{\tt Z}/\mbox{\tt a}\}$. We see that $\delta\gamma
\supseteq \theta\sigma$ and $dom(\delta\gamma-\theta\sigma)= \{{\tt
X}_0\}$ (which is a template variable). Moreover, it holds
$\mathsf{G}\theta\sigma = \mathsf{G}\delta\gamma$.\qed
\end{example}
Before we derive the first completeness result, we need certain
definitions and lemmas.
\begin{definition}\label{def-of-s-g}
Let $\mathsf{P}$ be a program and let $\mathsf{E}$ be a positive
expression or a goal clause. We define $S_\mathsf{E}$ to be the set
of all expressions that can be obtained from $\mathsf{E}$ by
substituting zero or more occurrences of every predicate constant
$\mathsf{p}$ in $\mathsf{E}$ with the expression $\mathsf{E}_1
\bigvee_{\pi} \cdots \bigvee_{\pi}\mathsf{E}_k$, where $\mathsf{p}
\leftarrow_{\pi} \mathsf{E}_i$ are all the clauses\footnote{We may
assume without loss of generality that each predicate symbol
$\mathsf{p}$ that is used in $\mathsf{P}$, has a definition in
$\mathsf{P}$: if no such definition exists, we can add to the
program the clause $\mathsf{p}\leftarrow_{\pi} \mathsf{E}$, where
$\mathsf{E}$ is a basic expression corresponding to the basic
element $\perp_{\pi}$.} for $\mathsf{p}$ in $\mathsf{P}$. Moreover,
$\widehat{\mathsf{E}} \in S_\mathsf{E}$ is the expression obtained
from $\mathsf{E}$ by substituting {\em every} predicate symbol
occurrence with the corresponding expression.
\end{definition}
\begin{lemma}\label{lemma-equal-tp-goal-hat}
Let $\mathsf{P}$ be a program, $\mathsf{E}$ a positive expression or
a goal clause, $I$ a Herbrand interpretation of $\mathsf{P}$ and $s$
a Herbrand state. Then, \(\lsem \mathsf{E} \rsem_s(T_\mathsf{P}(I))
= \lsem \widehat{\mathsf{E}} \rsem_s(I)\).
\end{lemma}
\begin{proof}
The proof is by structural induction on $\mathsf{E}$. Assume that
$\mathsf{E}$ is a positive expression (the proof for the case of
goal clause is similar). For the induction basis we need to consider
the cases where $\mathsf{E}$ is an argument variable $\mathsf{V}$,
an individual constant $\mathsf{c}$, a propositional constant
($\mathsf{0}$, $\mathsf{1}$), or a predicate constant $\mathsf{p}$.
Except for the last one, all other cases are straightforward because
the meaning of $\mathsf{E}$ is independent of $T_\mathsf{P}(I)$ and
$I$. For the last case assume that $\mathsf{E}_1, \ldots,
\mathsf{E}_k$ are all the bodies of the rules defining $\mathsf{p}$
in $\mathsf{P}$. By definition of the $T_{\mathsf{P}}$ operator, it
holds that $\lsem \mathsf{p} \rsem_s(T_\mathsf{P}(I)) =
\bigsqcup_{(\mathsf{p}\leftarrow_{\pi}\mathsf{E}_i)\in \mathsf{P}}{
\lsem \mathsf{E}_i \rsem_s(I) }$. Moreover, $\lsem
\widehat{\mathsf{E}} \rsem_s(I) = \lsem \mathsf{E}_1 \bigvee_{\pi}
\cdots \bigvee_{\pi} \mathsf{E}_k \rsem_s(I) =
\bigsqcup_{(\mathsf{p}\leftarrow_{\pi}\mathsf{E}_i)\in \mathsf{P}}{
\lsem \mathsf{E}_i \rsem_s(I) }$. This completes the basis case. For
the induction step, all cases are immediate.\qed
\end{proof}
\begin{lemma}\label{lemma-derive-of-s-g}
Let $\mathsf{P}$ a program, $\mathsf{G}$,$\mathsf{G}'$ goals and
$\mathsf{G}' \in S_\mathsf{G}$. If $\mathsf{G}'
\stackrel{\theta}{\rightarrow} \mathsf{H}'$  then
$\mathsf{G}\stackrel{\theta}{\twoheadrightarrow} \mathsf{H}$, where
$\mathsf{H}' \in S_\mathsf{H}$.
\end{lemma}
\begin{proof}
The proof is by induction on the number $m$ of top-level
subexpressions of the goal $\mathsf{G}$ that are connected with the
logical constant $\wedge$. The basis case is for $m=1$, ie., it
applies to goal clauses $\mathsf{G}$ that do not contain a top-level
$\wedge$. Assume that $\mathsf{G} = \leftarrow \mathsf{A}$. The
cases we need to examine for $\mathsf{A}$ for the induction basis
are the following: $(\mathsf{p} \,\, \mathsf{A}_1 \cdots
\mathsf{A}_n)$, $(\mathsf{Q}\,\mathsf{A}_1,\ldots\mathsf{A}_n)$,
$((\lambda \mathsf{V}.\mathsf{A'})\,\, \mathsf{A}_1 \cdots
\mathsf{A}_n)$, $((\mathsf{A}' \bigvee_{\pi} \mathsf{A}'') \,\,
\mathsf{A}_1 \cdots \mathsf{A}_n)$, $((\mathsf{A}' \bigwedge_{\pi}
\mathsf{A}'') \,\, \mathsf{A}_1 \cdots \mathsf{A}_n)$,
$(\mathsf{A}_1\approx \mathsf{A}_2)$, and $(\exists \mathsf{V}
\,\mathsf{A})$. The only non-trivial case is $\mathsf{A}=(\mathsf{p}
\,\, \mathsf{A}_1 \cdots \mathsf{A}_n)$, which we demonstrate.
Assume that $\mathsf{p}$ is defined in $\mathsf{P}$ with a set of
$k$ rules with right-hand sides
$\mathsf{E}_{1},\ldots,\mathsf{E}_{k}$. Let $\mathsf{E}_{\mathsf{p}}
= \mathsf{E}_{1}\bigvee_{\pi}\cdots \bigvee_{\pi}\mathsf{E}_{k}$.
Since $\mathsf{G} =\leftarrow (\mathsf{p} \, \mathsf{A}_1\cdots
\mathsf{A}_n)$, we have that $\mathsf{G}'= \leftarrow
(\mathsf{A}'\,\mathsf{A}'_1\cdots \mathsf{A}'_n)$, with
$\mathsf{A}'\in S_{\mathsf{p}}, \mathsf{A}'_1\in
S_{\mathsf{A}_1},\ldots, \mathsf{A}'_n\in S_{\mathsf{A}_n}$. We
distinguish three cases for $\mathsf{A}'$:
\begin{itemize}
\item $\mathsf{A}' = \mathsf{p}$. Then
      $\mathsf{G}'\stackrel{\epsilon}{\rightarrow} \mathsf{H}'$, where
      $\mathsf{H}' = \leftarrow
      (\mathsf{E}_{j}\,\mathsf{A}'_1\cdots \mathsf{A}'_n)$ for
      some $j$. We also have that $\mathsf{G} \stackrel{\epsilon}{\rightarrow} \mathsf{H}$, where
      $\mathsf{H} = \leftarrow (\mathsf{E}_{j}\,\mathsf{A}_1\cdots \mathsf{A}_n)$.
      Obviously, it holds that $\mathsf{H}' \in S_\mathsf{H}$.

\item $\mathsf{A}' = \mathsf{E}_{\mathsf{p}}$ and $\mathsf{E}_{\mathsf{p}}$ contains more
      than one disjunct. Then $\mathsf{G}'\stackrel{\epsilon}{\rightarrow}
      \mathsf{H}'$, where $\mathsf{H}' = \leftarrow
      (\mathsf{E}_{j}\,\mathsf{A}'_1\cdots \mathsf{A}'_n)$.
      We also have that $\mathsf{G} \stackrel{\epsilon}{\rightarrow} \mathsf{H}$, where
      $\mathsf{H} = \leftarrow (\mathsf{E}_{j}\,\mathsf{A}_1\cdots \mathsf{A}_n)$.
      Again, it holds that $\mathsf{H}' \in S_\mathsf{H}$.

\item $\mathsf{A}' = \mathsf{E}_{\mathsf{p}}$ and $\mathsf{E}_{\mathsf{p}}$ contains exactly one
      disjunct. Then this disjunct must be a lambda abstraction of the form $(\lambda\mathsf{V}.\mathsf{A''})$.
      This implies that $\mathsf{G}'=\leftarrow(( \lambda\mathsf{V}.\mathsf{A''})\,\mathsf{A}'_1\cdots\mathsf{A}'_n)$
      and $\mathsf{G}'\stackrel{\epsilon}{\rightarrow}
      \mathsf{H}'$, where $\mathsf{H}' =
      \leftarrow(\mathsf{A''}\{\mathsf{V}/\mathsf{A}'_1\}\,\mathsf{A}'_2\cdots\mathsf{A}'_n)$.
      On the other hand, $\mathsf{G}\stackrel{\epsilon}{\rightarrow}\mathsf{H}_1$, where
      $\mathsf{H}_1 = \leftarrow((
      \lambda\mathsf{V}.\mathsf{A''})\,\mathsf{A}_1\cdots\mathsf{A}_n)$,
      and $\mathsf{H}_1\stackrel{\epsilon}{\rightarrow}\mathsf{H}$,
      where $\mathsf{H}=\leftarrow(\mathsf{A''}\{\mathsf{V}/\mathsf{A}_1\}\,\mathsf{A}_2\cdots\mathsf{A}_n)$.
      Therefore,
      $\mathsf{G}\stackrel{\epsilon}{\twoheadrightarrow}\mathsf{H}$
      where $\mathsf{H}' \in S_{\mathsf{H}}$.
 \end{itemize}
The above completes the proof for the basis case. For the induction
step, the goal must be of the form $\mathsf{G}
=\leftarrow(\mathsf{A}_1\wedge \mathsf{A}_2)$. Then, $\mathsf{G}'
=\leftarrow(\mathsf{A}'_1\wedge \mathsf{A}'_2)$ where $\mathsf{A}'_1
\in S_{\mathsf{A}_1}$ and $\mathsf{A}'_2 \in S_{\mathsf{A}_2}$.
Since $\mathsf{G}' \stackrel{\theta}{\rightarrow} \mathsf{H}'$, we
conclude without loss of generality that
$\mathsf{A}'_1\stackrel{\theta}{\rightarrow}\mathsf{H}'_1$ and
$\mathsf{H}' = \leftarrow (\mathsf{H}'_1 \wedge
\mathsf{A}'_2\theta)$. By the induction hypothesis, since
$\mathsf{A}'_1\stackrel{\theta}{\rightarrow}\mathsf{H}'_1$, we get
that
$\mathsf{A}_1\stackrel{\theta}{\twoheadrightarrow}\mathsf{H}_1$,
where $\mathsf{H}'_1 \in S_{\mathsf{H}_1}$. But then this easily
implies that $(\mathsf{A}_1 \wedge
\mathsf{A}_2)\stackrel{\theta}{\twoheadrightarrow}(\mathsf{H}_1
\wedge \mathsf{A}_2\theta)$, ie.,
$\mathsf{G}\stackrel{\theta}{\twoheadrightarrow} \mathsf{H}$, where
$\mathsf{H}' \in S_\mathsf{H}$. This completes the proof for the
induction step and the lemma.\qed
\end{proof}
\begin{lemma}\label{lemma-refute-of-s-g}
Let $\mathsf{P}$ be a program, $\mathsf{G}$,$\mathsf{G}'$ goals and
$\mathsf{G}' \in S_\mathsf{G}$. If there exists an SLD-refutation
for $\mathsf{P}\cup\{\mathsf{G}'\}$ using substitution $\theta$,
then there also exists an SLD-refutation for
$\mathsf{P}\cup\{\mathsf{G}\}$ using the same substitution $\theta$.
\end{lemma}
\begin{proof}
The proof is by induction on the length $n$ of the refutation of
$\mathsf{P}\cup\{\mathsf{G}'\}$. The induction basis is for $n=1$
and includes the following cases for $\mathsf{G}$: $(\Box \wedge
\Box)$, $(\Box \vee \mathsf{E}_1)$, $(\mathsf{E}_1  \vee \Box)$,
$(\mathsf{E}_1 \approx \mathsf{E}_2)$, $((\lambda \mathsf{V}.\Box)\,
\mathsf{E})$ and $(\exists \mathsf{V}\,\Box)$. It can be easily
verified that the lemma holds for all these cases.

Suppose now that the result holds for $n-1$. We demonstrate that it
also holds for $n$. Let $\mathsf{G}' = \mathsf{G}'_0,
\mathsf{G}'_1,\ldots, \mathsf{G}'_n$ be the derived goals of the
SLD-refutation of $\mathsf{G}'$ using the sequence of substitutions
$\theta_1,\ldots,\theta_n$. Since $\mathsf{G}'
\stackrel{\theta_1}{\rightarrow} \mathsf{G}'_1$, by
Lemma~\ref{lemma-derive-of-s-g} there exists a goal $\mathsf{G}_1$
such that $\mathsf{G} \stackrel{\theta_1}{\twoheadrightarrow}
\mathsf{G}_1$ and $\mathsf{G}'_1 \in S_{\mathsf{G}_1}$. By the
induction hypothesis, $\mathsf{P}\cup\{\mathsf{G}_1\}$ has an
SLD-refutation using $\theta_2\cdots\theta_n$. It follows that
$\mathsf{P}\cup\{\mathsf{G}\}$ also has an SLD-refutation using
$\theta = \theta_1\cdots\theta_n$.\qed
\end{proof}
\begin{corollary}\label{lemma-refute-of-hat}
Let $\mathsf{P}$ be a program and $\mathsf{G}$ a goal. If there
exists an SLD-refutation for
$\mathsf{P}\cup\{\widehat{\mathsf{G}}\}$ using substitution
$\theta$, then there also exists an SLD-refutation for
$\mathsf{P}\cup\{\mathsf{G}\}$ using the same substitution
$\theta$.\qed
\end{corollary}
\begin{lemma}\label{lemma-completeness-bottom-interp}
Let $\mathsf{P}$ be a program and $\mathsf{G}=\leftarrow \mathsf{A}$
be a goal such that $\lsem \mathsf{A}
\rsem_s(\bot_{\mathcal{I}_\mathsf{P}}) = 1$ for all Herbrand states
$s$. Then, there exists an SLD-refutation for
$\mathsf{P}\cup\{\mathsf{G}\}$ with computed answer equal to the
identity substitution.
\end{lemma}
The proof of the lemma can be found in
Appendix~\ref{appendix-bottom}.

As in the first-order case, we have various forms of completeness.
We can now prove the analogue of a theorem due to Apt and van Emden
(see~\cite{Apt90}[Lemma 3.17] or~\cite{lloyd}[Theorem 8.3]).
\begin{theorem}\label{theorem-completeness-for-closed}
Let $\mathsf{P}$ be a program, $\mathsf{G}=\leftarrow \mathsf{A}$ a
goal and assume that $\lsem \mathsf{A} \rsem_s(M_\mathsf{P}) = 1$
for all Herbrand states $s$. Then, there exists an SLD-refutation
for $\mathsf{P}\cup\{\mathsf{G}\}$ with computed answer equal to the
identity substitution.
\end{theorem}
\begin{proof}
We prove by induction on $n$ that if $\lsem \mathsf{A}
\rsem_s(T_\mathsf{P} \uparrow n) = 1$ for all Herbrand states $s$,
then $\mathsf{P}\cup\{\mathsf{G}\}$ has an SLD-refutation with
computed answer equal to the identity substitution. For $n=0$ the
proof is a direct consequence of
Lemma~\ref{lemma-completeness-bottom-interp}.

Now suppose that the result holds for $n-1$. For the induction step
assume that $\lsem \mathsf{A} \rsem_s(T_\mathsf{P}\!\uparrow\! n) =
1$ for all $s$. By Lemma \ref{lemma-equal-tp-goal-hat}, $\lsem
\widehat{\mathsf{A}} \rsem_s(T_\mathsf{P} \!\uparrow\! {(n-1)}) =
1$. By the induction hypothesis there exists an SLD-refutation for
$\mathsf{P}\cup\{\widehat{\mathsf{G}}\}$ with computed answer equal
to the identity substitution. Let $\theta$ be the composition of the
substitutions that are used during the SLD-refutation of
$\mathsf{P}\cup\{\widehat{\mathsf{G}}\}$. By
Corollary~\ref{lemma-refute-of-hat}, $\mathsf{P}\cup\{\mathsf{G}\}$
also has an SLD-refutation using the same substitution $\theta$. The
restriction of $\theta$ to the free variables of $\mathsf{G}$ is
equal to the restriction of $\theta$ to the free variables of
$\widehat{\mathsf{G}}$ which is equal to the empty substitution.
Therefore, $\mathsf{P}\cup\{\mathsf{G}\}$ has an SLD-refutation with
computed answer equal to the identity substitution. \qed
\end{proof}

The following theorem generalizes a result of Hill~\cite{H74} (see
also~\cite{Apt90}[Theorem 3.15]):
\begin{theorem}\label{theorem-completeness-unsatisfiable}
Let $\mathsf{P}$ be a program and $\mathsf{G}=\leftarrow \mathsf{A}$
a goal. Suppose that $\mathsf{P}\cup\{\mathsf{G}\}$ is
unsatisfiable. Then, there exists an SLD-refutation of
$\mathsf{P}\cup\{\mathsf{G}\}$.
\end{theorem}
\begin{proof}
Since $\mathsf{P}\cup\{\mathsf{G}\}$ is unsatisfiable and since
$M_{\mathsf{P}}$ is a model of $\mathsf{P}$, we conclude that $\lsem
\mathsf{G} \rsem_s(M_\mathsf{P}) = 0$, for some state $s$.
Therefore, $\lsem \mathsf{A}\rsem_s(M_\mathsf{P}) = 1$. By
Lemma~\ref{lemma-transform-state-to-subst} we can construct a
substitution $\theta$ such that $\lsem \mathsf{A}\theta
\rsem_{s'}(M_\mathsf{P}) = 1$ for all states $s'$. By
Theorem~\ref{theorem-completeness-for-closed}, there exists an
SLD-refutation for $\mathsf{P}\cup\{\mathsf{G}\theta\}$. By
Lemma~\ref{lifting-lemma} there exists an SLD-refutation for
$\mathsf{P}\cup\{\mathsf{G}\}$.\qed
\end{proof}

Finally, the following theorem is a generalization of Clark's
theorem~\cite{Cla79} (see also the more
accessible~\cite{Apt90}[Theorem 3.18]) for the higher-order case:
\begin{theorem}[Completeness]\label{theorem-completeness-substitutions}
Let $\mathsf{P}$ be a program and $\mathsf{G}=\leftarrow \mathsf{A}$
a goal. For every correct answer $\theta$ for
$\mathsf{P}\cup\{\mathsf{G}\}$, there exists an SLD-refutation for
$\mathsf{P}\cup\{\mathsf{G}\}$ with computed answer $\delta$ and a
substitution $\gamma$ such that $\mathsf{G}\theta =
\mathsf{G}\delta\gamma$.
\end{theorem}
\begin{proof}
Since $\theta$ is a correct answer for $\mathsf{P}\cup\{\leftarrow
\mathsf{A}\}$, it follows that $\lsem \mathsf{A}\theta
\rsem_s(M_\mathsf{P}) = 1$ for all Herbrand states $s$. By
Theorem~\ref{theorem-completeness-for-closed},
$\mathsf{P}\cup\{\mathsf{G}\theta\}$ has an SLD-refutation with
computed answer equal to the identity substitution. This means that
if $\sigma$ is the composition of the substitutions used in the
refutation of $\mathsf{P}\cup\{\mathsf{G}\theta\}$, then
$\mathsf{G}\theta\sigma = \mathsf{G}\theta$. By
Lemma~\ref{lifting-lemma} there exists an SLD-refutation for
$\mathsf{P}\cup\{\mathsf{G}\}$ using substitution $\delta'$ such
that for some substitution $\gamma$, $\mathsf{G}\delta'\gamma =
\mathsf{G}\theta \sigma$. Let $\delta$ be $\delta'$ restricted to
the variables in $\mathsf{G}$. Then, it also holds that
$\mathsf{G}\delta'\gamma = \mathsf{G}\delta\gamma$, and therefore
$\mathsf{G}\delta\gamma = \mathsf{G}\theta \sigma =
\mathsf{G}\theta$.\qed
\end{proof}

\section{Related Work}\label{rw-section}
As already discussed in Section~\ref{intro-section}, research in
higher-order logic programming can be categorized in two main
streams: the {\em extensional} approaches and the {\em intensional}
ones.

Work on extensional higher-order logic programming is rather
limited. Apart from the results of~\cite{Wa91a}\footnote{The work
in~\cite{Wa91a} has also been used in order to define a higher-order
extension of Datalog~\cite{KRW05}.}, the only other work that has
come to our attention is that of M. Bezem~\cite{B99-395,B01-203},
who considers higher-order logic programming languages with syntax
similar to that of~\cite{Wa91a}. In~\cite{B01-203} a notion of
extensionality is defined (called the {\em extensional collapse})
and it is demonstrated that many logic programs are extensional
under this notion; however, this notion appears to differ from
classical extensionality and has a more proof-theoretical flavor.

On the other hand, work on intensional higher-order logic
programming is much more extended.  The two main existing approaches
in this area are represented by the languages $\lambda$Prolog and
HiLog. Both systems have mature implementations and have been tested
in various application domains. It should be noted that both
$\lambda$Prolog and HiLog encourage a form of higher-order
programming that extends in various ways the higher-order
programming capabilities that are supported by functional
programming languages. For a more detailed discussion on this issue,
see~\cite{NM98}[section 7.4].

In the rest of this section, we give a very brief presentation of
certain characteristics of these two systems that are related to
their intensional behavior (ie., characteristics that will help the
reader further clarify the differences between the intensional and
extensional approaches). A detailed discussion on the syntax,
semantics, implementation and applications of the two languages, is
outside the scope of this paper (and the interested reader can
consult the relevant bibliography).

\noindent{\bf $\lambda$Prolog}: The language was initially designed
in the late 1980s~\cite{MN86,Nad87,NM90} in order to provide a proof
theoretic basis for logic programming. The syntax of $\lambda$Prolog
is based on the intuitionistic theory of higher-order hereditary
Harrop formulas. The resulting language is a powerful one, that
allows the programmer to quantify over function and predicate
variables, to use $\lambda$-abstractions in terms, and so on. The
semantics of $\lambda$Prolog is not extensional (see for example the
discussion in~\cite{NM98}). The following simple example illustrates
this idea.
\begin{example}
Consider the $\lambda$Prolog program (we omit type declarations):
\[
\begin{array}{l}
\mbox{\tt r p.}\\
\mbox{\tt p X :- q X.}\\
\mbox{\tt q X :- p X.}\\
\end{array}
\]
The goal {\tt ?-r q.} fails for the above program.\qed
\end{example}
In the following we briefly discuss the behavior of $\lambda$Prolog
with respect to queries that contain uninstantiated higher-order
variables (because, from a user point of view, this is a key concept
that can differentiate an intensional system from an extensional
one). Consider for example the query:
\[
\begin{array}{l}
\mbox{\tt ?-(R john mary).}
\end{array}
\]
The above goal is not a meaningful one for $\lambda$Prolog because
there exist too many suitable answer substitutions (ie., predicate
terms) for {\tt R} that one could think of (see the relevant
discussion in~\cite{NM98}[page 50]). One way to get answers to such
queries is to use an extra predicate, say {\tt rel}, which specifies
a collection of predicate terms that are relevant to consider as
substitutions. In this case, the above query will be written as:
\[
\begin{array}{l}
\mbox{\tt ?-(rel R),(R john mary).}
\end{array}
\]
and the implementation will return those terms for which the query
succeeds. In other words, an answer for the above query is a
predicate term (such as for example {\tt married}, {\tt father}, a
lambda expression, and so on). As remarked in~\cite{CKW93-187},
``equality in (a fragment of) $\lambda$Prolog corresponds to
$\lambda$-equivalence and is not extensional: there may exist
predicates that are not $\lambda$-equivalent but still extensionally
equal''. This is a key difference from the extensional approach
presented in this paper, in which an answer to a query is a set.

\noindent{\bf HiLog:} The language possesses a higher-order syntax
and a first-order semantics~\cite{CKW89-37,CKW93-187}. It extends
classical logic programming quite naturally, and allows the
programmer to write in a concise way programs that would be rather
awkward to code in Prolog. It has been used in various application
domains (eg. deductive and object-oriented databases, modular logic
programming, and so on).

As it is the case with $\lambda$Prolog, HiLog is also an intensional
higher-order language. The examples with uninstantiated higher-order
variables mentioned for $\lambda$Prolog have a somewhat similar
behavior in HiLog. However, consider the program:
\[
\begin{array}{l}
\mbox{\tt married(john,mary).}
\end{array}
\]
Then, the query:
\[
\begin{array}{l}
\mbox{\tt ?-R(john,mary).}
\end{array}
\]
is a meaningful one for HiLog, and the interpreter will respond with
{\tt R = married}. Intuitively, the interpreter searches the program
for possible candidate relations and tests them one by one. Of
course, if there is no binary relation defined in the program, the
above query will fail.

The above program behavior can be best explained by the following
comment from~\cite{CKW93-187}: ``in HiLog predicates and other
higher-order syntactic objects are not equal unless they (ie., their
names) are equated explicitly''.

\section{Implementation and Future Work}\label{future-work}
A prototype implementation of the proposed proof procedure has been
performed in Haskell\footnote{The code can be retrieved from {\tt
http://code.haskell.org/hopes}}. A detailed description of the
implementation is outside the scope of this paper. However, in the
following we outline certain points that we feel are important.

The main difference in comparison to a first-order implementation,
is that the proof procedure has to generate an infinite (yet
enumerable) number of basic templates. In order to make more
efficient the production of the basic templates, one main
optimization has been adopted. As we have already mentioned in
Definition~\ref{definition-basic-template}, a basic template is a
non-empty finite union of basic expressions of a particularly simple
form. In the implementation, the members of this union are generated
in a ``demand-driven way'', as the following examples illustrate.
\begin{example}
Consider the query $\leftarrow \mbox{\tt
(R\,\,a\,\,b),(R\,\,c\,\,d)}$. The proposed proof procedure would
try some basic templates until it finds one that satisfies the
query. However, if it first tries the basic template {\tt
($\lambda$X.$\lambda$Y.(X$\approx$Z)$\wedge$(Y$\approx$W))} then
this will obviously not lead to an answer (since a relation that
satisfies the above query must contain at least two pairs of
elements). In order to avoid such cases, our implementation
initially produces a basic expression that consists of the union of
a basic template with an uninstantiated variable (say {\tt L}) of
the same type as the template; intuitively, {\tt L} represents a
(yet undetermined) set of basic templates that may be needed later
during resolution and which need not yet be explicitly generated. In
our example, the implementation starts with the production of an
expression of the form {\tt
($\lambda$X.$\lambda$Y.(X$\approx$Z)$\wedge$(Y$\approx$W))\,\,$\bigvee$\,\,L}.
When the second application in the goal is reached, then a second
basic template will be generated together with a new uninstantiated
variable (say {\tt L1}). The final answer to the query will be an
expression of the form: {\tt
($\lambda$X.$\lambda$Y.(X$\approx$a)$\wedge$(Y$\approx$b))\,\,$\bigvee$\,\,
($\lambda$X.$\lambda$Y.(X$\approx$c)$\wedge$(Y$\approx$d))\,\,$\bigvee$\,\,
L1}. The intuitive meaning of the above answer is that the query is
satisfied by all relations that contain at least the pairs $({\tt
a},{\tt b})$ and $({\tt c},{\tt d})$.

Notice that an important practical advantage of the above
optimization is that a unique answer to the given query is
generated. Notice also that if the formal proof procedure of the
previous sections was followed faithfully in the implementation,
then an infinite number of answers would be generated: an answer
representing the two-element relation $\{({\tt a},{\tt b}),({\tt
c},{\tt d})\}$, an answer representing all three-element relations
$\{({\tt a},{\tt b}),({\tt c},{\tt d}),({\tt X1},{\tt X2})\}$, an
answer representing the four-element relations $\{({\tt a},{\tt
b}),({\tt c},{\tt d}),({\tt X1},{\tt X2}),({\tt X3},{\tt X4})\}$,
and so on.\qed
\end{example}
\begin{example}
Consider the {\tt ordered} predicate of
Example~\ref{ordered-example} and let $\leftarrow$ {\tt ordered R
[1,2,3]} be a query. Following the same ideas as in the previous
example, the implementation will produce the unique answer {\tt
($\lambda$X.$\lambda$Y.(X$\approx$1)$\wedge$(Y$\approx$2))}\,\,$\bigvee$\,\,\mbox{\tt
($\lambda$X.$\lambda$Y.(X$\approx$2)$\wedge$(Y$\approx$3))\,\,$\bigvee$\,\,L}.
Intuitively, this answer states that the list {\tt [1,2,3]} is
ordered under any relation of the form $\{({\tt 1},{\tt 2}),({\tt
2},{\tt 3})\} \cup L$.

Finally, consider Example~\ref{closure-example} defining the {\tt
closure} predicate. Consider also the query $\leftarrow$ {\tt
closure Q a b}. Then, the implementation will enumerate the
following (infinite set of) answers:
\[
\begin{array}{lll}
{\tt Q} & = &  \mbox{\tt ($\lambda$X.$\lambda$Y.(X$\approx$a)$\wedge$(Y$\approx$b))\,\,$\bigvee$\,\,L} \\
{\tt Q} & = & \mbox{\tt ($\lambda$X.$\lambda$Y.(X$\approx$a)$\wedge$(Y$\approx$Z))$\bigvee$($\lambda$X.$\lambda$Y.(X$\approx$Z)$\wedge$(Y$\approx$b))\,\,$\bigvee$\,\,L}\\
        & \ldots &
\end{array}
\]
which intuitively correspond to relations of the following forms:
\[
\begin{array}{lll}
{\tt Q} & = & \{({\tt a},{\tt b})\}\cup L \\
{\tt Q} & = & \{({\tt a},{\tt Z}),({\tt Z},{\tt b})\}\cup L\\
        &\ldots&
\end{array}
\]
Intuitively, the above answers state that the pair $({\tt a},{\tt
b})$ belongs to the transitive closure of all relations that contain
at least the pair $({\tt a},{\tt b})$; moreover, it also belongs to
the transitive closure of all relations that contain at least two
pairs of the form $({\tt a},{\tt Z})$ and $({\tt Z},{\tt b})$ for
any {\tt Z}, and so on.\qed
\end{example}

We are currently considering issues regarding an extended WAM-based
implementation of the ideas presented in the paper. We believe that
ideas originating from graph-reduction~\cite{FH88} will also prove
vital in the development of this extended implementation.

Another interesting direction for future research is the extension
of our higher-order fragment with negation-as-failure. The semantics
of negation in a higher-order setting could probably be captured
model-theoretically using the recent purely logical characterization
of the well-founded semantics through an appropriate infinite-valued
logic~\cite{RW05}.

\vspace{0.5cm} \noindent {\bf Acknowledgments}: We would like to
thank Costas Koutras for valuable discussions regarding algebraic
lattices and the reading group on programming languages at the
University of Athens for many insightful comments and suggestions.

\vspace{-0.1cm}

\bibliographystyle{alpha}


\appendix

\section{Proof of
Lemma~\ref{algebraicity-composition}}\label{appendix-algebraicity-composition}
In order to establish Lemma~\ref{algebraicity-composition}, we first
demonstrate the following auxiliary propositions:
\begin{proposition}\label{step-prop}
Let $A$ be a poset and $L$ be an algebraic lattice. Then, for each
step function $(a \searrow c)$ and for every $f:[A
\stackrel{m}{\rightarrow} L]$ it holds that $(a \searrow c)
\sqsubseteq f$ if and only if $c \sqsubseteq f(a)$.
\end{proposition}
\begin{proof}
If $(a \searrow c) \sqsubseteq f$, by applying both functions to $a$
we get $c \sqsubseteq f(a)$. Now suppose that $c \sqsubseteq f(a)$
and consider an arbitrary $x \in A$. In case $a \sqsubseteq x$, we
have $(a \searrow c)(x) = c$ thus, since $c \sqsubseteq f(a)$ and
$f$ is monotonic, $(a \searrow c)(x) \sqsubseteq f(x)$. Otherwise,
$(a \searrow c)(x) = \perp_L$ thus $(a \searrow c)(x) \sqsubseteq
f(x)$. It follows that $(a \searrow c) \sqsubseteq f$.\qed
\end{proof}
\begin{proposition}\label{algebraicity-criterion}
Let $L$ be a complete lattice and assume there exists $B\subseteq
{\cal K}(L)$ such that for every $x \in L$,  $x = \bigsqcup
B_{[x]}$. Then $L$ is an algebraic lattice ($\omega$-algebraic if
$B$ is countable) with basis ${\cal K}(L) = \left \{ \bigsqcup{M}
\mid M \text{ is a finite subset of } B \right \}$.
\end{proposition}
\begin{proof}
It is immediate that $L$ is algebraic, since by assumption every
element of $L$ can be written as the least upper bound of a set of
compact elements of $L$. The nontrivial part is establishing the
relation between ${\cal K}(L)$ and $B$.

Given $x \in L$, we let $\Delta(x)$ be the set $\{ \bigsqcup M \mid
M \text{ is a finite subset of } B_{[x]}\}$. Notice that $B_{[x]} =
\bigcup \{ M \mid M \text{ is a finite subset of } B_{[x]} \}$.
Using Proposition~\ref{set-of-sets}(2), we have that $\bigsqcup
B_{[x]} = \bigsqcup \Delta(x)$ and thus $\bigsqcup \Delta(x) = x$.
We show that for each $x \in L$ it holds that ${\cal K}(L)_{[x]} =
\Delta(x)$ by proving that each set is a subset of the other one.

First consider an arbitrary $c \in {\cal K}(L)_{[x]}$ and recall
that $c = \bigsqcup \Delta(c)$. By the compactness of $c$, there
exists a finite $A \subseteq \Delta(c)$ such that $c \sqsubseteq
\bigsqcup A$. But then $\bigsqcup A \sqsubseteq c$ because $c$ is an
upper bound of $\Delta(c)$, and therefore $c = \bigsqcup A$. By the
definition of $\Delta(c)$ and the fact that $A\subseteq \Delta(c)$,
we get that $c = \bigsqcup \{\bigsqcup M_1,\ldots,\bigsqcup M_r\}$,
where $M_1,\ldots,M_r$ are finite subsets of $B_{[c]}$. By
Proposition~\ref{set-of-sets}(2), $c=\bigsqcup{(M_1\cup \cdots \cup
M_r)}$. In other words there exists a finite set $M = M_1\cup \cdots
\cup M_r$ such that $M \subseteq B_{[c]} \subseteq B_{[x]}$ and $c =
\bigsqcup M$, which means that $c \in \Delta(x)$.

On the other hand, consider a finite set $M = \{c_1,\ldots,
c_n\}\subseteq B_{[x]}$ such that $\bigsqcup M \in \Delta(x)$. Let
$A$ be a subset of $L$ such that $\bigsqcup M \sqsubseteq \bigsqcup
A$. Due to the compactness of each $c_i$, by $c_i \sqsubseteq
\bigsqcup A$ we get $c_i \sqsubseteq \bigsqcup A_i$ for some finite
$A_i \subseteq A$. But then, for every $i$, $c_i \sqsubseteq
\bigsqcup A_i \sqsubseteq \bigsqcup \{\bigsqcup A_1,\ldots,\bigsqcup
A_n\} = \bigsqcup (A_1\cup \cdots \cup A_n)$. In other words,
$\bigsqcup M \sqsubseteq \bigsqcup (A_1 \cup \cdots \cup A_n)$,
which implies that $\bigsqcup M$ is compact. Moreover, since $x$ is
an upper bound of $M$, we have that $\bigsqcup M \in {\cal
K}(L)_{[x]}$. Hence, ${\cal K}(L)_{[x]} = \Delta(x)$.

To complete the proof, simply take $x = \bigsqcup L$ in the equality
${\cal K}(L)_{[x]} = \Delta(x)$. If, additionally, $B$ is countable,
the cardinality of ${\cal K}(L)$ is bounded by the number of finite
subsets of a countable set, which is countable. Hence, $L$ is an
$\omega$-algebraic lattice in this case.\qed
\end{proof}
We can now proceed to the proof of
Lemma~\ref{algebraicity-composition}:

\noindent{\bf Lemma~\ref{algebraicity-composition}} Let $A$ be a
poset and $L$ be an algebraic lattice. Then, $[A
\stackrel{m}{\rightarrow} L]$ is an algebraic lattice whose basis is
the set of all least upper bounds of finitely many step functions
from $A$ to $L$. If, additionally, $A$ is countable and $L$ is an
$\omega$-algebraic lattice then $[A \stackrel{m}{\rightarrow} L]$ is
an $\omega$-algebraic lattice.
\begin{proof}
Let $B$ denote the set of all step functions from $A$ to $L$. Recall
that $[A \stackrel{m}{\rightarrow} L]$ forms a complete lattice by
Proposition~\ref{monotonic-functions-make-a-complete-lattice}. Let
$(a \searrow c) \in B$ be an arbitrary step function. We show that
$(a \searrow c)$ is compact. Consider a set $F$ of monotonic
functions from $A$ to $L$ such that $(a \searrow c) \sqsubseteq
\bigsqcup{F}$. By Propositions~\ref{step-prop}
and~\ref{monotonic-functions-make-a-complete-lattice} we get that $c
\sqsubseteq \bigsqcup_{f \in F} f(a)$. By the compactness of $c$,
there exists a finite $F' \subseteq F$ such that $c \sqsubseteq
\bigsqcup_{f \in F'} f(a)$. Let $f'=\bigsqcup F'$. Then, $c
\sqsubseteq f'(a)$, or equivalently by Proposition~\ref{step-prop},
$(a \searrow c) \sqsubseteq f' = \bigsqcup F'$. Hence, $(a \searrow
c)$ is compact.

We now show that every monotonic function $f \in [A
\stackrel{m}{\rightarrow} L]$ is the least upper bound of $B_{[f]}$.
Since $f$ is an upper bound of this set, we let $g$ be an upper
bound of $B_{[f]}$ and prove that $f \sqsubseteq g$. In fact, we
consider an arbitrary $x \in A$ and prove that $f(x) \sqsubseteq
g(x)$. Suppose $S_x$ is the set of all step functions $h_c = (x
\searrow c)$ for every compact element $c \in {\cal K}(L)_{[f(x)]}$.
By Proposition \ref{step-prop}, we have that for all step functions
$h_c \in S_x$, are $h_c \sqsubseteq f$; thus $S_x$ is a subset of
$B_{[f]}$. Since $g$ is an upper bound of $B_{[f]}$, it must also be
an upper bound of $S_x$, therefore it holds that $h_c \sqsubseteq g$
for each $h_c \in S_x$. Applying this inequality for $x$ we get that
$c \sqsubseteq g(x)$ for each $c \in {\cal K}(L)_{[f(x)]}$,
therefore $\bigsqcup{\cal K}(L)_{[f(x)]} \sqsubseteq g(x)$. Since
$L$ is an algebraic lattice, $f(x)$ is the least upper bound of
${\cal K}(L)_{[f(x)]}$, thus $f(x) \sqsubseteq g(x)$. Hence, $f$ is
the least upper bound of $B_{[f]}$.

On the whole, we have shown that $B$ is a subset of ${\cal K}([A
\stackrel{m}{\rightarrow} L])$ such that each monotonic function $f$
from $A$ to $L$ is the least upper bound of $B_{[f]}$. Notice that
if, additionally, $A$ is countable and $L$ is an $\omega$-algebraic
lattice, then $B$ is countable because its cardinality is equal to
that of the cartesian product of two countable sets. Now apply
Proposition \ref{algebraicity-criterion}.\qed
\end{proof}

\section{Proof of
Lemma~\ref{monotonicity-of-semantics-wrt-state}}\label{appendix-monotonicity-state}
{\bf Lemma~\ref{monotonicity-of-semantics-wrt-state}} Let
$\mathsf{E}:\rho$ be an expression of ${\cal H}$ and let $D$ be a
nonempty set. Moreover, let $s,s_1,s_2$ be states over $D$ and let
$I$ be an interpretation over $D$. Then:
\begin{enumerate}
\item $\lsem \mathsf{E} \rsem_s(I) \in \lsem \rho \rsem_D$.

\item If $\mathsf{E}$ is positive and $s_1 \sqsubseteq_{{\cal S}_{{\cal H},D}} s_2$
      then $\lsem \mathsf{E} \rsem_{s_1}(I) \sqsubseteq_{\rho} \lsem \mathsf{E} \rsem_{s_2}(I)$.
\end{enumerate}
\begin{proof}
The two statements are established simultaneously by a structural
induction on $\mathsf{E}$.

\vspace{0.3cm} \noindent {\em Induction Basis:} The cases for
$\mathsf{E}$ being $\mathsf{0}$, $\mathsf{1}$, $\mathsf{c}$,
$\mathsf{p}$ or $\mathsf{V}$, are all straightforward.

\vspace{0.3cm} \noindent {\em Induction Step:} The interesting cases
are $\mathsf{E}=(\mathsf{E}_1 \mathsf{E}_2)$ and $\mathsf{E} =
(\lambda\mathsf{V}.\mathsf{E}_1)$. The other cases are easier and
omitted.

\vspace{0.2cm}

\noindent \underline{\em Case 1:} $\mathsf{E} = (\mathsf{E}_1
\mathsf{E}_2)$. We examine the two statements of the lemma:

\vspace{0.1cm}

\noindent {\em Statement 1}: Assume that $\mathsf{E}_1:\rho_1
\rightarrow \pi_2$ and $\mathsf{E}_2:\rho_1$. Then, it suffices to
demonstrate that $\lsem (\mathsf{E}_1 \mathsf{E}_2) \rsem_{s}(I) \in
\lsem \pi_2 \rsem_D$, or equivalently that $\bigsqcup_{b \in
B}(\lsem \mathsf{E}_1\rsem_{s}(I)(b)) \in \lsem \pi_2 \rsem_D$,
where $B = {\cal F}_D(type(\mathsf{E}_2))_{[\lsem
\mathsf{E}_2\rsem_{s}(I)]} =  \{b \in {\cal F}_D(type(\mathsf{E}_2))
\mid b \sqsubseteq \lsem \mathsf{E}_2 \rsem_{s} (I)\}$. By the
induction hypothesis, $\lsem \mathsf{E}_1 \rsem_{s}(I) \in \lsem
\rho_1 \rightarrow \pi_2 \rsem_D$ and $\lsem \mathsf{E}_2
\rsem_{s}(I) \in \lsem \rho_1 \rsem_D$. But then, for every $b \in
B$, $\lsem \mathsf{E}_1\rsem_{s}(I)(b) \in \lsem \pi_2 \rsem_D$ and
since $\lsem \pi_2 \rsem_D$ is a complete lattice, we get that
$\bigsqcup_{b \in B}(\lsem \mathsf{E}_1\rsem_{s}(I)(b)) \in \lsem
\pi_2 \rsem_D$.

\vspace{0.1cm}

\noindent {\em Statement 2}: It suffices to demonstrate that $\lsem
(\mathsf{E}_1 \mathsf{E}_2) \rsem_{s_1}(I) \sqsubseteq \lsem
(\mathsf{E}_1 \mathsf{E}_2) \rsem_{s_2}(I)$, or equivalently that
$\bigsqcup_{b_2 \in B_2}(\lsem \mathsf{E}_1\rsem_{s_1}(I)(b_2))
\sqsubseteq \bigsqcup_{b'_2 \in B'_2}(\lsem
\mathsf{E}_1\rsem_{s_2}(I)(b'_2))$, where $B_2 = {\cal
F}_D(type(\mathsf{E}_2))_{[\lsem \mathsf{E}_2\rsem_{s_1}(I)]}$ and
$B'_2 = {\cal F}_D(type(\mathsf{E}_2))_{[\lsem
\mathsf{E}_2\rsem_{s_2}(I)]}$. Notice that by definition, $B_2 = \{b
\in {\cal F}_D(type(\mathsf{E}_2)) \mid b \sqsubseteq \lsem
\mathsf{E}_2 \rsem_{s_1} (I)\}$ and $B'_2 = \{b \in {\cal
F}_D(type(\mathsf{E}_2)) \mid b \sqsubseteq \lsem \mathsf{E}_2
\rsem_{s_2}(I)\}$. By the induction hypothesis we have $\lsem
\mathsf{E}_2 \rsem_{s_1}(I) \sqsubseteq \lsem \mathsf{E}_2
\rsem_{s_2}(I)$, and therefore $B_2 \subseteq B'_2$. By the
induction hypothesis we also have that $\lsem \mathsf{E}_1
\rsem_{s_1}(I) \sqsubseteq \lsem \mathsf{E}_1 \rsem_{s_2}(I)$. By
the induction hypothesis for the first statement of the lemma, both
$\lsem\mathsf{E}_1\rsem_{s_1}(I)$ and
$\lsem\mathsf{E}_1\rsem_{s_2}(I)$ are monotonic functions since they
belong to $\lsem \rho_1\rightarrow \pi_2\rsem_D$. Therefore,
$\bigsqcup_{b_2 \in B_2}(\lsem \mathsf{E}_1\rsem_{s_1}(I)(b_2))
\sqsubseteq \bigsqcup_{b'_2 \in B'_2}(\lsem
\mathsf{E}_1\rsem_{s_2}(I)(b'_2))$, or equivalently $\lsem
(\mathsf{E}_1 \mathsf{E}_2) \rsem_{s_1}(I) \sqsubseteq \lsem
(\mathsf{E}_1 \mathsf{E}_2) \rsem_{s_2}(I)$.

\vspace{0.2cm}

\noindent \underline{\em Case 2:} $\mathsf{E} =
(\lambda\mathsf{V}.\mathsf{E}_1)$. We examine the two statements of
the lemma:

\vspace{0.1cm}

\noindent {\em Statement 1}: Assume that $\mathsf{V}:\rho_1$ and
$\mathsf{E}_1:\pi_1$. We show that $\lsem
(\lambda\mathsf{V}.\mathsf{E}_1) \rsem_{s}(I) \in \lsem \rho_1
\rightarrow \pi_1 \rsem_D$. We distinguish two cases, namely $\rho_1
= \iota$ and $\rho_1 = \pi$. If $\rho_1 = \iota$ then the result
follows easily using the induction hypothesis for the first
statement of the lemma. If $\rho_1 = \pi$, then we must demonstrate
that $\lsem (\lambda\mathsf{V}.\mathsf{E}_1) \rsem_{s}(I) \in \lsem
\pi \rightarrow \pi_1 \rsem_D = [{\cal K}(\lsem \pi \rsem_D)
\stackrel{m}{\rightarrow} \lsem \pi_1 \rsem_D]$. In other words, we
need to show that the function $\lambda d.\lsem \mathsf{E}_1
\rsem_{s[d/\mathsf{V}]}(I)$ is monotonic. But this follows directly
from the induction hypothesis for the second statement of the lemma.

\vspace{0.1cm}

\noindent {\em Statement 2}: It suffices to show that $\lsem
(\lambda\mathsf{V}.\mathsf{E}_1) \rsem_{s_1}(I) \sqsubseteq \lsem
(\lambda\mathsf{V}.\mathsf{E}_1) \rsem_{s_2}(I)$. By the semantics
of lambda abstraction, it suffices to show that $\lambda d. \lsem
\mathsf{E}_1 \rsem_{s_1[d/\mathsf{V}]}(I) \sqsubseteq \lambda d.
\lsem \mathsf{E}_1 \rsem_{s_2[d/\mathsf{V}]}(I)$, or that for every
$d$, $\lsem \mathsf{E}_1 \rsem_{s_1[d/\mathsf{V}]}(I) \sqsubseteq
\lsem \mathsf{E}_1 \rsem_{s_2[d/\mathsf{V}]}(I)$, which holds by the
induction hypothesis.\qed
\end{proof}

\section{Proof of Lemma~\ref{monotonicity-of-semantics}}\label{appendix-monotonicity}

{\bf Lemma~\ref{monotonicity-of-semantics}}
Let $\mathsf{P}$ be a program and let $\mathsf{E}:\rho$ be a
positive expression of $\mathsf{P}$. Let $I,J$ be Herbrand
interpretations and $s$ a Herbrand state of $\mathsf{P}$ . If $I
\sqsubseteq_{{\cal I}_\mathsf{P}} J$ then $\lsem \mathsf{E}
\rsem_s(I) \sqsubseteq_{\rho} \lsem \mathsf{E} \rsem_s(J)$.
\begin{proof}
The proof is by a structural induction on $\mathsf{E}$.

\vspace{0.4cm} \noindent {\em Induction Basis:} The cases for
$\mathsf{E}$ being $\mathsf{0}$, $\mathsf{1}$, $\mathsf{c}$,
$\mathsf{p}$ or $\mathsf{V}$, are all straightforward.

\vspace{0.4cm} \noindent {\em Induction Step:} The interesting cases
are $\mathsf{E}=(\mathsf{E}_1 \mathsf{E}_2)$ and $\mathsf{E} =
(\lambda\mathsf{V}.\mathsf{E}_1)$. The other cases are easier and
omitted.

\vspace{0.2cm} \noindent {\em Case 1:} $\mathsf{E} = (\mathsf{E}_1
\mathsf{E}_2)$. It suffices to demonstrate that $\lsem (\mathsf{E}_1
\mathsf{E}_2) \rsem_s(I) \sqsubseteq \lsem (\mathsf{E}_1
\mathsf{E}_2) \rsem_s(J)$, or equivalently that $\bigsqcup_{b_2 \in
B_2}(\lsem \mathsf{E}_1\rsem_s(I)(b_2)) \sqsubseteq \bigsqcup_{b'_2
\in B'_2}(\lsem \mathsf{E}_1\rsem_s(J)(b'_2))$, where $B_2 = {\cal
F}_D(type(\mathsf{E}_2))_{[\lsem \mathsf{E}_2\rsem_s(I)]}$ and $B'_2
= {\cal F}_D(type(\mathsf{E}_2))_{[\lsem \mathsf{E}_2\rsem_s(J)]}$.
Notice that by definition, $B_2 = \{b \in {\cal
F}_D(type(\mathsf{E}_2)) \mid b \sqsubseteq \lsem \mathsf{E}_2
\rsem_s (I)\}$ and $B'_2 = \{b \in {\cal F}_D(type(\mathsf{E}_2))
\mid b \sqsubseteq \lsem \mathsf{E}_2 \rsem_s (J)\}$. By the
induction hypothesis we have $\lsem \mathsf{E}_2 \rsem_s(I)
\sqsubseteq \lsem \mathsf{E}_2 \rsem_s(J)$, and therefore $B_2
\subseteq B'_2$. By the induction hypothesis we also have that
$\lsem \mathsf{E}_1 \rsem_s(I) \sqsubseteq \lsem \mathsf{E}_1
\rsem_s(J)$. Therefore, $\bigsqcup_{b_2 \in B_2}(\lsem
\mathsf{E}_1\rsem_s(I)(b_2)) \sqsubseteq \bigsqcup_{b'_2 \in
B'_2}(\lsem \mathsf{E}_1\rsem_s(J)(b'_2))$, or equivalently $\lsem
(\mathsf{E}_1 \mathsf{E}_2) \rsem_s(I) \sqsubseteq \lsem
(\mathsf{E}_1 \mathsf{E}_2) \rsem_s(J)$.

\vspace{0.2cm}

\noindent {\em Case 2:} $\mathsf{E} =
(\lambda\mathsf{V}.\mathsf{E}_1)$. It suffices to show that $\lsem
(\lambda\mathsf{V}.\mathsf{E}_1) \rsem_s(I) \sqsubseteq \lsem
(\lambda\mathsf{V}.\mathsf{E}_1) \rsem_s(J)$. By the semantics of
lambda abstraction, it suffices to show that $\lambda d. \lsem
\mathsf{E}_1 \rsem_{s[d/\mathsf{V}]}(I) \sqsubseteq \lambda d. \lsem
\mathsf{E}_1 \rsem_{s[d/\mathsf{V}]}(J)$, or that for every $d$,
$\lsem \mathsf{E}_1 \rsem_{s[d/\mathsf{V}]}(I) \sqsubseteq \lsem
\mathsf{E}_1 \rsem_{s[d/\mathsf{V}]}(J)$, which holds by the
induction hypothesis.\qed
\end{proof}

\section{Proof of Lemma~\ref{continuity-of-semantics}}\label{appendix-continuity}
{\bf Lemma~\ref{continuity-of-semantics}} Let $\mathsf{P}$ be a
program and let $\mathsf{E}$ be any positive expression of
$\mathsf{P}$. Let ${\cal I}$ be a directed set of Herbrand
interpretations and $s$ be a Herbrand state of $\mathsf{P}$. Then,
$\lsem \mathsf{E} \rsem_s (\bigsqcup {\cal I}) = \bigsqcup_{I \in
{\cal I}} \lsem \mathsf{E} \rsem_s (I)$.

\begin{proof}
The proof can be performed in two steps: we first show that $\lsem
\mathsf{E} \rsem_s (\bigsqcup {\cal I}) \sqsupseteq \bigsqcup_{I \in
{\cal I}} \lsem \mathsf{E} \rsem_s (I)$ and then that $\lsem
\mathsf{E} \rsem_s (\bigsqcup {\cal I}) \sqsubseteq \bigsqcup_{I \in
{\cal I}} \lsem \mathsf{E} \rsem_s (I)$.

For the first of these two statements observe that by
Lemma~\ref{monotonicity-of-semantics}, we have that $\lsem
\mathsf{E} \rsem_s (\bigsqcup {\cal I}) \sqsupseteq  \lsem
\mathsf{E} \rsem_s (I)$, for all $I \in {\cal I}$. But then $\lsem
\mathsf{E} \rsem_s (\bigsqcup {\cal I})$ is an upper bound of the
set $\{ \lsem \mathsf{E} \rsem_s (I) \mid I \in {\cal I}\}$, and
therefore $\lsem \mathsf{E} \rsem_s (\bigsqcup {\cal I}) \sqsupseteq
\bigsqcup_{I \in {\cal I}} \lsem \mathsf{E} \rsem_s (I)$. It remains
to show that $\lsem \mathsf{E} \rsem_s (\bigsqcup {\cal I})
\sqsubseteq \bigsqcup_{I \in {\cal I}} \lsem \mathsf{E} \rsem_s
(I)$. The proof is by a structural induction on $\mathsf{E}$.

\vspace{0.4cm} \noindent {\em Induction Basis:} The cases for
$\mathsf{E}$ being $\mathsf{0}$, $\mathsf{1}$, $\mathsf{c}$,
$\mathsf{p}$ or $\mathsf{V}$, are all straightforward.

\vspace{0.4cm} \noindent {\em Induction Hypothesis:} Assume that for
given expressions $\mathsf{E}_1,\mathsf{E}_2$ it holds that $\lsem
\mathsf{E}_i \rsem_s (\bigsqcup {\cal I}) = \bigsqcup_{I \in {\cal
I}} \lsem \mathsf{E}_i \rsem_s (I)$, $i \in \{1,2\}$. Notice that we
assume equality. This is due to the fact that the one direction has
already been established for all expressions while the other
direction is assumed.

\vspace{0.4cm} \noindent {\em Induction Step:} We distinguish the
following cases:

\vspace{0.2cm} \noindent \underline{\em Case 1:} $\mathsf{E} =
\mathsf{f}\,\,\mathsf{E}_1 \cdots \mathsf{E}_n$. This case is
straightforward since for every interpretation $I$ and for every
state $s$, the value of $\lsem \mathsf{f}\,\,\mathsf{E}_1
\cdots\mathsf{E}_n \rsem_s(I)$ only depends on $s$ (since the
expressions $\mathsf{E}_1,\ldots,\mathsf{E}_n$ are of type $\iota$
and do not contain predicate symbols).

\vspace{0.2cm} \noindent \underline{\em Case 2:} $\mathsf{E} =
(\mathsf{E}_1 \mathsf{E}_2)$. Assume that $\mathsf{E}_2 :\rho$.
Then:
\[
\begin{array}{lll}
&      &\lsem (\mathsf{E}_1 \mathsf{E}_2) \rsem_s (\bigsqcup {\cal I}) =\\
\\
&  =   & \bigsqcup_{b \in B}(\lsem \mathsf{E}_1 \rsem_s (\bigsqcup
         {\cal I})(b)),\,\, \mbox{where $B = \{b \in {\cal F}_D(\rho) \mid b
         \sqsubseteq \lsem \mathsf{E}_2\rsem_s(\bigsqcup {\cal I})\}$}\\
&      & \mbox{(Semantics of application)}\\
\\
&  =   &\bigsqcup_{b \in B}((\bigsqcup_{I \in {\cal I}} \lsem
         \mathsf{E}_1 \rsem_s (I))(b)),\,\, \mbox{where $B = \{b \in {\cal F}_D(\rho) \mid b
         \sqsubseteq \lsem \mathsf{E}_2\rsem_s(\bigsqcup {\cal I})\}$}\\
&      & \mbox{(Induction hypothesis)}\\
\\
&  =   &\bigsqcup_{b \in B}(\bigsqcup_{I \in {\cal I}} \lsem
         \mathsf{E}_1 \rsem_s (I)(b)),\,\, \mbox{where $B = \{b \in {\cal F}_D(\rho) \mid b
         \sqsubseteq \lsem \mathsf{E}_2\rsem_s(\bigsqcup {\cal I})\}$}\\
&      & \mbox{(Proposition~\ref{monotonic-functions-make-a-complete-lattice})}\\
\\
&  =   & \bigsqcup \{ \lsem \mathsf{E}_1 \rsem_s (I)(b) \mid I \in
         {\cal I},\,\, b \in {\cal F}_D(\rho), \,\, b
         \sqsubseteq \lsem \mathsf{E}_2\rsem_s(\bigsqcup {\cal
         I})\}\\
&      &\mbox{(Proposition~\ref{set-of-sets}(2))}\\
\\
&  =   & \bigsqcup \{ \lsem \mathsf{E}_1 \rsem_s (I)(b) \mid I \in
         {\cal I},\,\, b \in {\cal F}_D(\rho), \,\, b
         \sqsubseteq \bigsqcup_{I \in {\cal I}}\lsem \mathsf{E}_2\rsem_s(I)\}\\
&      &\mbox{(Induction hypothesis)}\\
\\
&  =   & \bigsqcup \{ \lsem \mathsf{E}_1 \rsem_s (I)(b) \mid I \in
         {\cal I},\,\, b \in {\cal F}_D(\rho), \,\, b
         \sqsubseteq \bigsqcup_{J \in F}\lsem \mathsf{E}_2\rsem_s(J),\,\,\mbox{$F$ finite subset of ${\cal I}$}\}\\
&      &\mbox{(Since $b$ is either a compact element or a member of $D$)}\\
\\
&  \sqsubseteq  & \bigsqcup \{ \lsem \mathsf{E}_1 \rsem_s (I)(b)
         \mid I \in {\cal I},\,\, b \in {\cal F}_D(\rho), \,\, b
         \sqsubseteq \lsem \mathsf{E}_2\rsem_s(J)\},\,\, \mbox{for some $J\in {\cal I}$}\\
&      &\mbox{(Because ${\cal I}$ is directed and $\lsem \mathsf{E}_2 \rsem_s$ is monotonic by Lemma~\ref{monotonicity-of-semantics})}\\
\\
&  \sqsubseteq  & \bigsqcup \{ \lsem \mathsf{E}_1 \rsem_s (I)(b)
         \mid I \in {\cal I},\,\,J \in {\cal I},\,\, b \in {\cal F}_D(\rho), \,\, b
         \sqsubseteq \lsem \mathsf{E}_2\rsem_s(J)\}\\
&      &\mbox{(Proposition~\ref{set-of-sets}(1))}\\
\\
&  \sqsubseteq  & \bigsqcup_{I \in {\cal I},J\in {\cal I}} \bigsqcup
         \{\lsem \mathsf{E}_1
         \rsem_s (I)(b)\mid  b \in {\cal F}_D(\rho), \,\, b
         \sqsubseteq \lsem \mathsf{E}_2\rsem_s(J)\} \\
&      &\mbox{(Proposition~\ref{set-of-sets}(2))}\\
\\
&  \sqsubseteq  & \bigsqcup_{I \in {\cal I}} \bigsqcup
         \{\lsem \mathsf{E}_1
         \rsem_s (I)(b)\mid  b \in {\cal F}_D(\rho), \,\, b
         \sqsubseteq \lsem \mathsf{E}_2\rsem_s(I)\} \\
&      &\mbox{(Proposition~\ref{diagonal})}\\
\\
&   =   &\bigsqcup_{I \in {\cal I}} \lsem (\mathsf{E}_1
        \mathsf{E}_2)\rsem_s (I)\\
&      &\mbox{(Semantics of application)}
\end{array}
\]

\vspace{0.2cm} \noindent \underline{\em Case 3:} $\mathsf{E} =
(\lambda \mathsf{V}. \mathsf{E}_1)$. We show that $\lsem (\lambda
\mathsf{V}. \mathsf{E}_1)\rsem_s (\bigsqcup {\cal I}) \sqsubseteq
\bigsqcup_{I \in {\cal I}} \lsem (\lambda \mathsf{V}. \mathsf{E}_1)
\rsem_s (I)$. Consider $b\in {\cal F}_D(type(\mathsf{V}))$. By the
semantics of lambda abstraction we get that $\lsem (\lambda
\mathsf{V}. \mathsf{E}_1)\rsem_s (\bigsqcup {\cal I})(b) = \lsem
\mathsf{E}_1\rsem_{s[b/\mathsf{V}]} (\bigsqcup {\cal I})$; by the
induction hypothesis this is equal to $\bigsqcup_{I \in {\cal I}}
\lsem \mathsf{E}_1\rsem_{s[b/\mathsf{V}]}(I)$, which by
Proposition~\ref{monotonic-functions-make-a-complete-lattice} is
equal to $(\bigsqcup_{I \in {\cal I}} \lsem (\lambda
\mathsf{V}.\mathsf{E}_1)\rsem_{s}(I))(b)$.

\vspace{0.2cm} \noindent \underline{\em Case 4:} $\mathsf{E} =
(\mathsf{E}_1 \bigvee_{\pi} \mathsf{E}_2)$. We show that $\lsem
(\mathsf{E}_1 \bigvee_{\pi} \mathsf{E}_2) \rsem_s (\bigsqcup {\cal
I}) \sqsubseteq \bigsqcup_{I \in {\cal I}} \lsem (\mathsf{E}_1
\bigvee_{\pi} \mathsf{E}_2) \rsem_s (I)$, ie., that for all
$b_1,\ldots,b_n$, if $\lsem (\mathsf{E}_1 \bigvee_{\pi}
\mathsf{E}_2) \rsem_s (\bigsqcup {\cal I}) \, b_1 \cdots b_n = 1$
then $(\bigsqcup_{I \in {\cal I}} \lsem (\mathsf{E}_1 \bigvee_{\pi}
\mathsf{E}_2) \rsem_s (I)) \, b_1 \cdots b_n = 1$. By the semantics
of $\bigvee_{\pi}$ we get that if $\lsem (\mathsf{E}_1 \bigvee_{\pi}
\mathsf{E}_2) \rsem_s (\bigsqcup {\cal I}) \, b_1 \cdots b_n = 1$
then $\lsem \mathsf{E}_1 \rsem_s (\bigsqcup {\cal I}) \, b_1 \cdots
b_n = 1$ or $\lsem \mathsf{E}_2 \rsem_s (\bigsqcup {\cal I}) \, b_1
\cdots b_n = 1$. By the induction hypothesis and
Proposition~\ref{monotonic-functions-make-a-complete-lattice} we get
that either $\bigsqcup_{I \in {\cal I}} (\lsem \mathsf{E}_1 \rsem_s
(I) \, b_1 \cdots b_n) = 1$ or $\bigsqcup_{I \in {\cal I}} (\lsem
\mathsf{E}_2 \rsem_s (I) \, b_1 \cdots b_n) = 1$. Then there must
exist $I\in {\cal I}$ such that either $\lsem \mathsf{E}_1 \rsem_s
(I) \, b_1 \cdots b_n = 1$ or $\lsem \mathsf{E}_2 \rsem_s (I) \, b_1
\cdots b_n = 1$. By the semantics of $\bigvee_{\pi}$ we get that
$\lsem (\mathsf{E}_1 \bigvee_{\pi} \mathsf{E}_2) \rsem_s (I) \, b_1
\cdots b_n = 1$ and therefore $(\bigsqcup_{I \in {\cal I}} \lsem
(\mathsf{E}_1 \bigvee_{\pi} \mathsf{E}_2) \rsem_s (I)) \, b_1 \cdots
b_n = 1$.

\vspace{0.2cm} \noindent \underline{\em Case 5:} $\mathsf{E} =
(\mathsf{E}_1 \bigwedge_{\pi} \mathsf{E}_2)$. We show that $\lsem
(\mathsf{E}_1 \bigwedge_{\pi} \mathsf{E}_2) \rsem_s (\bigsqcup {\cal
I}) \sqsubseteq \bigsqcup_{I \in {\cal I}} \lsem (\mathsf{E}_1
\bigwedge_{\pi} \mathsf{E}_2) \rsem_s (I)$. In other words, it
suffices to show that for all $b_1,\ldots,b_n$, if $\lsem
(\mathsf{E}_1 \bigwedge_{\pi} \mathsf{E}_2) \rsem_s (\bigsqcup {\cal
I}) \, b_1 \cdots b_n = 1$ then $(\bigsqcup_{I \in {\cal I}} \lsem
(\mathsf{E}_1 \bigwedge_{\pi} \mathsf{E}_2) \rsem_s (I)) \, b_1
\cdots b_n = 1$. But if $\lsem (\mathsf{E}_1 \bigwedge_{\pi}
\mathsf{E}_2) \rsem_s (\bigsqcup {\cal I}) \, b_1 \cdots b_n = 1$,
then by the semantics of $\bigwedge_{\pi}$ we get that $\lsem
\mathsf{E}_1 \rsem_s (\bigsqcup {\cal I}) \, b_1 \cdots b_n = 1$ and
$\lsem \mathsf{E}_2 \rsem_s (\bigsqcup {\cal I}) \, b_1 \cdots b_n =
1$. By the induction hypothesis and
Proposition~\ref{monotonic-functions-make-a-complete-lattice} this
implies that $\bigsqcup_{I \in {\cal I}} (\lsem \mathsf{E}_1 \rsem_s
(I) \, b_1 \cdots b_n) = 1$ and $\bigsqcup_{I \in {\cal I}} (\lsem
\mathsf{E}_2 \rsem_s (I) \, b_1 \cdots b_n) = 1$. This means that
there must exist $I_1,I_2\in {\cal I}$ such that $\lsem \mathsf{E}_1
\rsem_s (I_1) \, b_1 \cdots b_n = 1$ and $\lsem \mathsf{E}_2 \rsem_s
(I_2) \, b_1 \cdots b_n = 1$. Since ${\cal I}$ is directed, we get
that $I = \bigsqcup \{I_1,I_2\}$ exists in ${\cal I}$ and it holds
that $\lsem \mathsf{E}_1 \rsem_s (I) \, b_1 \cdots b_n = 1$ and
$\lsem \mathsf{E}_2 \rsem_s (I) \, b_1 \cdots b_n = 1$. By the
semantics of $\bigwedge_{\pi}$, $\lsem (\mathsf{E}_1 \bigwedge_{\pi}
\mathsf{E}_2) \rsem_s (I) \, b_1 \cdots b_n = 1$ and therefore
$(\bigsqcup_{I \in {\cal I}} \lsem (\mathsf{E}_1 \bigwedge_{\pi}
\mathsf{E}_2) \rsem_s (I)) \, b_1 \cdots b_n = 1$.

\vspace{0.2cm} \noindent \underline{\em Case 6:} $\mathsf{E} =
(\mathsf{E}_1 \approx \mathsf{E}_2)$. It suffices to show that
$\lsem (\mathsf{E}_1 \approx \mathsf{E}_2) \rsem_s (\bigsqcup {\cal
I}) \sqsubseteq \bigsqcup_{I \in {\cal I}} \lsem (\mathsf{E}_1
\approx \mathsf{E}_2) \rsem_s (I)$. This is straightforward since
the value of $\lsem (\mathsf{E}_1 \approx \mathsf{E}_2) \rsem$ only
depends on $s$ (since the expressions $\mathsf{E}_1,\mathsf{E}_2$ do
not contain predicate symbols).

\vspace{0.2cm} \noindent \underline{\em Case 7:} $\mathsf{E} =
(\exists \mathsf{V} \, \mathsf{E}_1)$. We show that $\lsem (\exists
\mathsf{} \, \mathsf{E}_1) \rsem_s (\bigsqcup {\cal I}) \sqsubseteq
\bigsqcup_{I \in {\cal I}} \lsem (\exists \mathsf{V} \,
\mathsf{E}_1) \rsem_s (I)$ or equivalently that if $\lsem (\exists
\mathsf{V} \, \mathsf{E}_1) \rsem_s (\bigsqcup {\cal I}) = 1$ then
$\bigsqcup_{I \in {\cal I}} \lsem (\exists \mathsf{V} \,
\mathsf{E}_1) \rsem_s (I)= 1$. Notice now that if $\lsem (\exists
\mathsf{V} \, \mathsf{E}_1) \rsem_s (\bigsqcup {\cal I}) = 1$ then
there exists $b$ such that $\lsem
\mathsf{E}_1\rsem_{s[b/\mathsf{V}]} (\bigsqcup {\cal I}) = 1$, which
by the induction hypothesis gives $\bigsqcup_{I \in {\cal I}} \lsem
\mathsf{E}_1\rsem_{s[b/\mathsf{V}]} (I) = 1$. But this last
statement implies that $\bigsqcup_{I \in {\cal I}} \lsem (\exists
\mathsf{V} \, \mathsf{E}_1) \rsem_s (I)= 1$.\qed

\end{proof}

\section{Proof of
Lemma~\ref{pre-lifting-lemma}}\label{appendix-lifting}

{\bf Lemma~\ref{pre-lifting-lemma}} Let $\mathsf{P}$ be a program,
$\mathsf{G}$ a goal and $\theta$ a substitution. Suppose that there
exists an SLD-refutation of $\mathsf{P}\cup\{\mathsf{G}\theta\}$
using substitution $\sigma$. Then, there exists an SLD-refutation of
$\mathsf{P}\cup\{\mathsf{G}\}$ using a substitution $\delta$, where
for some substitution $\gamma$ it holds that $\delta\gamma \supseteq
\theta\sigma$ and $dom(\delta\gamma-\theta\sigma)$ is a (possibly
empty) set of template variables that are introduced during the
refutation of $\mathsf{P}\cup\{\mathsf{G}\}$.
\begin{proof}
The proof is by induction on the length $n$ of the SLD-refutation of
$\mathsf{P}\cup \{\mathsf{G}\theta\}$.

\vspace{0.4cm} \noindent \underline{\em Induction Basis:} The basis
case is for $n=1$. We need to distinguish cases based on the
structure of $\mathsf{G}$.
%
The most interesting case is $\mathsf{G}=(\mathsf{E}_1 \approx
\mathsf{E}_2)$ (the rest are simpler and omitted). By assumption, it
holds that
$(\mathsf{E}_1\theta\approx\mathsf{E}_2\theta)\stackrel{\sigma}{\rightarrow}
\Box$, where $\sigma$ is an mgu of $\mathsf{E}_1\theta$ and
$\mathsf{E}_2\theta$. But then we also have that
$(\mathsf{E}_1\approx\mathsf{E}_2)\stackrel{\delta}{\rightarrow}
\Box$, where $\delta$ is an mgu of $\mathsf{E}_1$ and
$\mathsf{E}_2$. Since $\theta \sigma$ is a unifier of
$\mathsf{E}_1,\mathsf{E}_2$, there exists substitution $\gamma$ such
that $\theta\sigma = \delta\gamma$.

\vspace{0.4cm} \noindent \underline{\em Induction Step:} We
demonstrate the statement for SLD-refutations of length $n+1$. We
distinguish cases based on the structure of $\mathsf{G}$.

\vspace{0.2cm}\noindent \underline{\em Case 1:} $\mathsf{G}
=\leftarrow(\mathsf{p}\, \mathsf{E}_1\cdots \mathsf{E}_k)$. Then,
$\mathsf{G}\theta = \leftarrow(\mathsf{p}\,
(\mathsf{E}_1\theta)\cdots (\mathsf{E}_k\theta))$. By
Definition~\ref{derives-in-one-step} we get that $\mathsf{p}\,
(\mathsf{E}_1\theta)\cdots (\mathsf{E}_k\theta)
\stackrel{\epsilon}{\rightarrow} \mathsf{E}\,
(\mathsf{E}_1\theta)\cdots (\mathsf{E}_k\theta)$, where
$\mathsf{p}\leftarrow \mathsf{E}$ is a rule in $\mathsf{P}$. By
assumption, $\mathsf{E}\, (\mathsf{E}_1\theta)\cdots
(\mathsf{E}_k\theta)$ has an SLD-refutation of length $n$ using
$\sigma$. Consider now the goal $\mathsf{G}$. By
Definition~\ref{derives-in-one-step}, we get that $(\mathsf{p}\,
\mathsf{E}_1\cdots \mathsf{E}_k) \stackrel{\epsilon}{\rightarrow}
(\mathsf{E}\, \mathsf{E}_1\cdots \mathsf{E}_k)$. Notice now that
since $\mathsf{E}$ is a closed lambda expression, it holds that
$(\mathsf{E}\, \mathsf{E}_1\cdots \mathsf{E}_k)\theta =
(\mathsf{E}\, (\mathsf{E}_1\theta)\cdots (\mathsf{E}_k\theta))$.
Moreover, since $(\mathsf{E}\, (\mathsf{E}_1\theta)\cdots
(\mathsf{E}_k\theta))$ has an SLD-refutation of length $n$ using
$\sigma$, we get by the induction hypothesis that $(\mathsf{E}\,
\mathsf{E}_1\cdots \mathsf{E}_k)$ has an SLD-refutation using
substitution $\delta$, where for some substitution $\gamma$ it holds
that $\delta\gamma \supseteq \theta\sigma$ and
$dom(\delta\gamma-\theta\sigma)$ is a set of template variables that
are introduced during the refutation of $(\mathsf{E}\,
\mathsf{E}_1\cdots \mathsf{E}_k)$. But then, $(\mathsf{p}\,
\mathsf{E}_1\cdots \mathsf{E}_k)$ has an SLD-refutation which
satisfies the requirements of the lemma.

\vspace{0.2cm}\noindent \underline{\em Case 2:} $\mathsf{G}
=\leftarrow(\mathsf{Q}\, \mathsf{E}_1\cdots \mathsf{E}_k)$. Consider
first the case where $\theta(\mathsf{Q}) = \mathsf{B}$, for some
basic expression $\mathsf{B}$. Then, $\mathsf{G}\theta =
\leftarrow(\mathsf{B}\, (\mathsf{E}_1\theta)\cdots
(\mathsf{E}_k\theta))$. Notice now that $\mathsf{B}$ can be either a
higher-order predicate variable or a finite-union of lambda
abstractions. We examine the case where $\mathsf{B}$ is a single
lambda abstraction (the other two cases are similar). Since
$\mathsf{B}$ is a lambda abstraction, assume that
$\mathsf{B}=\lambda \mathsf{V}.\mathsf{C}$. By
Definition~\ref{derives-in-one-step} we get that $\mathsf{B}\,
(\mathsf{E}_1\theta)\cdots (\mathsf{E}_k\theta)
\stackrel{\epsilon}{\rightarrow}
\mathsf{C}\{\mathsf{V}/(\mathsf{E}_1\theta)\}\,(\mathsf{E}_2\theta)
\cdots (\mathsf{E}_k\theta)$. By assumption, $
\mathsf{C}\{\mathsf{V}/(\mathsf{E}_1\theta)\}\,(\mathsf{E}_2\theta)
\cdots (\mathsf{E}_k\theta)$ has an SLD-refutation of length $n$
using $\sigma$. Consider now the goal $\mathsf{G}$. By
Definition~\ref{derives-in-one-step}, we get that $(\mathsf{Q}\,
\mathsf{E}_1\cdots \mathsf{E}_k)
\stackrel{\{\mathsf{Q}/\mathsf{B}_t\}}{\rightarrow} \mathsf{B}_t\,
(\mathsf{E}_1\{\mathsf{Q}/\mathsf{B}_t\})\cdots
(\mathsf{E}_k\{\mathsf{Q}/\mathsf{B}_t\})$, where
$\mathsf{B}_t=\lambda \mathsf{V}.\mathsf{C}_t$, and
$\mathsf{B}=\mathsf{B}_t\gamma_1$, for some substitution $\gamma_1$
with $dom(\gamma_1) = FV(\mathsf{B}_t)$. We assume without loss of
generality that the set $dom(\gamma_1)$ is disjoint from
$FV(\mathsf{G})$ and from $dom(\theta)\cup FV(range(\theta))$. By
Definition~\ref{derives-in-one-step} we get that $\mathsf{B}_t\,
(\mathsf{E}_1\{\mathsf{Q}/\mathsf{B}_t\})\cdots
(\mathsf{E}_k\{\mathsf{Q}/\mathsf{B}_t\})
\stackrel{\epsilon}{\rightarrow}
\mathsf{C}_t\{\mathsf{V}/\mathsf{E}_1\{\mathsf{Q}/\mathsf{B}_t\}\}\,(\mathsf{E}_2\{\mathsf{Q}/\mathsf{B}_t\})
\cdots (\mathsf{E}_k\{\mathsf{Q}/\mathsf{B}_t\})$. Notice now that:
$$((\mathsf{C}_t\{\mathsf{V}/\mathsf{E}_1\{\mathsf{Q}/\mathsf{B}_t\}\})\,(\mathsf{E}_2\{\mathsf{Q}/\mathsf{B}_t\})
\cdots (\mathsf{E}_k\{\mathsf{Q}/\mathsf{B}_t\}))\theta\gamma_1 =
(\mathsf{C}\{\mathsf{V}/\mathsf{E}_1\theta\})\,(\mathsf{E}_2\theta)
\cdots (\mathsf{E}_k\theta)$$
Then, since
$(\mathsf{C}\{\mathsf{V}/\mathsf{E}_1\theta\})\,(\mathsf{E}_2\theta)
\cdots (\mathsf{E}_k\theta)$ has an SLD-refutation of length $n$
using $\sigma$, we get by the induction hypothesis that
$(\mathsf{C}_t\{\mathsf{V}/\mathsf{E}_1\{\mathsf{Q}/\mathsf{B}_t\}\})\,(\mathsf{E}_2\{\mathsf{Q}/\mathsf{B}_t\})
\cdots (\mathsf{E}_k\{\mathsf{Q}/\mathsf{B}_t\})$ has an
SLD-refutation using substitution $\delta'$, where for some
substitution $\gamma$ it holds $\delta'\gamma \supseteq
\theta\gamma_1\sigma$ and $dom(\delta'\gamma-\theta\gamma_1\sigma)$
is a set of template variables that are introduced during this
SLD-refutation. From the above discussion we conclude that
$(\mathsf{Q}\, \mathsf{E}_1\cdots \mathsf{E}_k)$ has an
SLD-refutation using substitution $\delta =
\{\mathsf{Q}/\mathsf{B}_t\}\delta'$. Moreover, it holds that
$\delta\gamma = \{\mathsf{Q}/\mathsf{B}_t\}\delta'\gamma \supseteq
\{\mathsf{Q}/\mathsf{B}_t\}\theta\gamma_1\sigma \supseteq
\theta\sigma$ and $dom(\delta\gamma-\theta\sigma)$ is a set of
template variables that are introduced during the refutation of
$(\mathsf{Q}\, \mathsf{E}_1\cdots \mathsf{E}_k)$.

Consider now the case where $\theta(\mathsf{Q})$ is undefined. Then,
$\mathsf{G}\theta = \leftarrow(\mathsf{Q}\,
(\mathsf{E}_1\theta)\cdots (\mathsf{E}_k\theta))$. By
Definition~\ref{derives-in-one-step} we get that $\mathsf{Q}\,
(\mathsf{E}_1\theta)\cdots (\mathsf{E}_k\theta)
\stackrel{\{\mathsf{Q}/\mathsf{B}_t\}}{\rightarrow}
\mathsf{B}_t\,(\mathsf{E}_1\theta\{\mathsf{Q}/\mathsf{B}_t\}) \cdots
(\mathsf{E}_k\theta\{\mathsf{Q}/\mathsf{B}_t\})$. We may assume
without loss of generality that the set $FV(\mathsf{B}_t)$ is
disjoint from $FV(\mathsf{G})$ and from $dom(\theta)\cup
FV(range(\theta))$. By assumption,
$\mathsf{B}_t\,(\mathsf{E}_1\theta\{\mathsf{Q}/\mathsf{B}_t\})
\cdots (\mathsf{E}_k\theta\{\mathsf{Q}/\mathsf{B}_t\})$ has an
SLD-refutation of length $n$ using $\sigma'$, where $\sigma =
\{\mathsf{Q}/\mathsf{B}_t\}\sigma'$. Consider now the goal
$\mathsf{G}$. By Definition~\ref{derives-in-one-step}, we get that
$(\mathsf{Q}\, \mathsf{E}_1\cdots \mathsf{E}_k)
\stackrel{\{\mathsf{Q}/\mathsf{B}_t\}}{\rightarrow} \mathsf{B}_t\,
(\mathsf{E}_1\{\mathsf{Q}/\mathsf{B}_t\})\cdots
(\mathsf{E}_k\{\mathsf{Q}/\mathsf{B}_t\})$. Notice now that:
$$(\mathsf{B}_t\,(\mathsf{E}_1\{\mathsf{Q}/\mathsf{B}_t\})
\cdots
(\mathsf{E}_k\{\mathsf{Q}/\mathsf{B}_t\}))\theta\{\mathsf{Q}/\mathsf{B}_t\}=
\mathsf{B}_t\,(\mathsf{E}_1\theta\{\mathsf{Q}/\mathsf{B}_t\}) \cdots
(\mathsf{E}_k\theta\{\mathsf{Q}/\mathsf{B}_t\})$$
Then, since
$\mathsf{B}_t\,(\mathsf{E}_1\theta\{\mathsf{Q}/\mathsf{B}_t\})
\cdots (\mathsf{E}_k\theta\{\mathsf{Q}/\mathsf{B}_t\})$ has an
SLD-refutation of length $n$ using $\sigma'$, we get by the
induction hypothesis that
$\mathsf{B}_t\,(\mathsf{E}_1\{\mathsf{Q}/\mathsf{B}_t\}) \cdots
(\mathsf{E}_k\{\mathsf{Q}/\mathsf{B}_t\})$ has an SLD-refutation
using substitution $\delta'$, where for some substitution $\gamma$
it holds $\delta'\gamma \supseteq \theta\{\mathsf{Q}/\mathsf{B}_t\}
\sigma'$ and $dom(\delta'\gamma-\theta\{\mathsf{Q}/\mathsf{B}_t\}
\sigma')$ is a set of template variables that are introduced during
this SLD-refutation; notice that these template variables can be
chosen to be different than the variables in $FV(\mathsf{B}_t)$.
From the above discussion we conclude that $(\mathsf{Q}\,
\mathsf{E}_1\cdots \mathsf{E}_k)$ has an SLD-refutation using
substitution $\delta = \{\mathsf{Q}/\mathsf{B}_t\}\delta'$.
Moreover, it holds that $\delta\gamma =
\{\mathsf{Q}/\mathsf{B}_t\}\delta'\gamma \supseteq
\{\mathsf{Q}/\mathsf{B}_t\}\theta\{\mathsf{Q}/\mathsf{B}_t\}\sigma'
= \theta\{\mathsf{Q}/\mathsf{B}_t\}\sigma' = \theta\sigma$ and
$dom(\delta\gamma-\theta\sigma)$ is a set of template variables that
are introduced during the refutation of $(\mathsf{Q}\,
\mathsf{E}_1\cdots \mathsf{E}_k)$.

\vspace{0.2cm}\noindent \underline{\em Case 3:} $\mathsf{G}
=\leftarrow((\lambda \mathsf{V}.\mathsf{E})\, \mathsf{E}_1\cdots
\mathsf{E}_k)$. Then, $\mathsf{G}\theta = \leftarrow((\lambda
\mathsf{V}.\mathsf{E}\theta)\, (\mathsf{E}_1\theta)\cdots
(\mathsf{E}_k\theta))$. By Definition~\ref{derives-in-one-step} we
get that $(\lambda \mathsf{V}.\mathsf{E}\theta)\,
(\mathsf{E}_1\theta)\cdots (\mathsf{E}_k\theta)
\stackrel{\epsilon}{\rightarrow}
(\mathsf{E}\theta\{\mathsf{V}/(\mathsf{E}_1\theta)\})\,
(\mathsf{E}_2\theta)\cdots (\mathsf{E}_k\theta)$. Moreover, by
assumption, $(\mathsf{E}\theta\{\mathsf{V}/(\mathsf{E}_1\theta)\})\,
(\mathsf{E}_2\theta)\cdots (\mathsf{E}_k\theta)$ has an
SLD-refutation of length $n$ using $\sigma$. Consider now the goal
$\mathsf{G}$. By Definition~\ref{derives-in-one-step}, we get that
$(\lambda \mathsf{V}.\mathsf{E})\, \mathsf{E}_1\cdots \mathsf{E}_k
\stackrel{\epsilon}{\rightarrow}
(\mathsf{E}\{\mathsf{V}/\mathsf{E}_1\})\, \mathsf{E}_2\cdots
\mathsf{E}_k$. Notice now that
$((\mathsf{E}\{\mathsf{V}/\mathsf{E}_1\})\, \mathsf{E}_2\cdots
\mathsf{E}_k)\theta =
(\mathsf{E}\theta\{\mathsf{V}/(\mathsf{E}_1\theta)\})\,
(\mathsf{E}_2\theta)\cdots (\mathsf{E}_k\theta)$, and since the
latter expression has an SLD-refutation of length $n$ using
$\sigma$, we get by the induction hypothesis that
$(\mathsf{E}\{\mathsf{V}/\mathsf{E}_1\})\, \mathsf{E}_2\cdots
\mathsf{E}_k$ has an SLD-refutation using a substitution $\delta$,
where for some substitution $\gamma$ it holds $\delta\gamma
\supseteq \theta\sigma$ and $dom(\delta\gamma-\theta\sigma)$ is a
set of template variables that are introduced during this
refutation. But then, $((\lambda \mathsf{V}.\mathsf{E})\,
\mathsf{E}_1\cdots \mathsf{E}_k)$ has an SLD-refutation using
substitution $\delta$ which satisfies the requirements of the lemma.

\vspace{0.2cm}\noindent \underline{\em Case 4:} $\mathsf{G}
=\leftarrow((\mathsf{E}'\bigvee_{\pi}\mathsf{E}'')\,
\mathsf{E}_1\cdots \mathsf{E}_k)$. Then, $\mathsf{G}\theta =
\leftarrow((\mathsf{E}'\theta\bigvee_{\pi}\mathsf{E}''\theta)\,
(\mathsf{E}_1\theta)\cdots (\mathsf{E}_k\theta))$. By
Definition~\ref{derives-in-one-step} we get that
$(\mathsf{E}'\theta\bigvee_{\pi}\mathsf{E}''\theta)\,
(\mathsf{E}_1\theta)\cdots (\mathsf{E}_k\theta)
\stackrel{\epsilon}{\rightarrow} (\mathsf{E}'\theta)\,
(\mathsf{E}_1\theta)\cdots (\mathsf{E}_k\theta)$ (and symmetrically
for $\mathsf{E}''$). By assumption, either $(\mathsf{E}'\theta)\,
(\mathsf{E}_1\theta)\cdots (\mathsf{E}_k\theta)$ or
$(\mathsf{E}''\theta)\, (\mathsf{E}_1\theta)\cdots
(\mathsf{E}_k\theta)$ has an SLD-refutation of length $n$ using
$\sigma$. Assume, without loss of generality, that
$(\mathsf{E}'\theta)\, (\mathsf{E}_1\theta)\cdots
(\mathsf{E}_k\theta)$ has an SLD-refutation of length $n$ using
$\sigma$. Consider now the goal $\mathsf{G}$. By
Definition~\ref{derives-in-one-step}, we get that
$(\mathsf{E}'\bigvee_{\pi}\mathsf{E}'')\, \mathsf{E}_1\cdots
\mathsf{E}_k \stackrel{\epsilon}{\rightarrow} \mathsf{E}'\,
\mathsf{E}_1\cdots \mathsf{E}_k$. Notice now that $(\mathsf{E}'\,
\mathsf{E}_1\cdots \mathsf{E}_k)\theta = (\mathsf{E}'\theta)\,
(\mathsf{E}_1\theta)\cdots (\mathsf{E}_k\theta)$, and since the
latter expression has an SLD-refutation of length $n$ using
$\sigma$, we get by the induction hypothesis that $\mathsf{E}'\,
\mathsf{E}_1\cdots \mathsf{E}_k$ has an SLD-refutation using a
substitution $\delta$, where for some substitution $\gamma$ it holds
$\delta\gamma \supseteq \theta\sigma$ and
$dom(\delta\gamma-\theta\sigma)$ is a set of template variables that
are introduced during this refutation. But then,
$(\mathsf{E}'\bigvee_{\pi}\mathsf{E}'')\, \mathsf{E}_1\cdots
\mathsf{E}_k$ has an SLD-refutation using substitution $\delta$
which satisfies the requirements of the lemma.

\vspace{0.2cm}\noindent \underline{\em Case 5:} $\mathsf{G}
=\leftarrow((\mathsf{E}'\bigwedge_{\pi}\mathsf{E}'')\,
\mathsf{E}_1\cdots \mathsf{E}_k)$. Then, $\mathsf{G}\theta =
\leftarrow((\mathsf{E}'\theta\bigwedge_{\pi}\mathsf{E}''\theta)\,
(\mathsf{E}_1\theta)\cdots (\mathsf{E}_k\theta))$. By
Definition~\ref{derives-in-one-step} we get
$(\mathsf{E}'\theta\bigvee_{\pi}\mathsf{E}''\theta)\,
(\mathsf{E}_1\theta)\cdots (\mathsf{E}_k\theta)
\stackrel{\epsilon}{\rightarrow} ((\mathsf{E}'\theta)\,
(\mathsf{E}_1\theta)\cdots (\mathsf{E}_k\theta)) \wedge
((\mathsf{E}''\theta)\, (\mathsf{E}_1\theta)\cdots
(\mathsf{E}_k\theta))$. By assumption, $((\mathsf{E}'\theta)\,
(\mathsf{E}_1\theta)\cdots (\mathsf{E}_k\theta)) \wedge
((\mathsf{E}''\theta)\, (\mathsf{E}_1\theta)\cdots
(\mathsf{E}_k\theta))$ has an SLD-refutation of length $n$ using
$\sigma$. Consider now the goal $\mathsf{G}$. By
Definition~\ref{derives-in-one-step}, we get that
$(\mathsf{E}'\bigwedge_{\pi}\mathsf{E}'')\, \mathsf{E}_1\cdots
\mathsf{E}_k \stackrel{\epsilon}{\rightarrow} (\mathsf{E}'\,
\mathsf{E}_1\cdots \mathsf{E}_k)\wedge (\mathsf{E}''\,
\mathsf{E}_1\cdots \mathsf{E}_k)$. Notice now that it holds that
$((\mathsf{E}'\, \mathsf{E}_1\cdots \mathsf{E}_k)\wedge
(\mathsf{E}''\, \mathsf{E}_1\cdots \mathsf{E}_k))\theta =
((\mathsf{E}'\theta)\, (\mathsf{E}_1\theta)\cdots
(\mathsf{E}_k\theta)) \wedge ((\mathsf{E}''\theta)\,
(\mathsf{E}_1\theta)\cdots (\mathsf{E}_k\theta))$; since the latter
expression has an SLD-refutation of length $n$ using $\sigma$, we
get by the induction hypothesis that $(\mathsf{E}'\,
\mathsf{E}_1\cdots \mathsf{E}_k)\wedge (\mathsf{E}''\,
\mathsf{E}_1\cdots \mathsf{E}_k)$ has an SLD-refutation using a
substitution $\delta$, where for some substitution $\gamma$ it holds
$\delta\gamma \supseteq \theta\sigma$ and
$dom(\delta\gamma-\theta\sigma)$ is a set of template variables that
are introduced during this refutation. But then,
$(\mathsf{E}'\bigwedge_{\pi}\mathsf{E}'')\, \mathsf{E}_1\cdots
\mathsf{E}_k$ has an SLD-refutation using substitution $\delta$
which satisfies the requirements of the lemma.

\vspace{0.2cm}\noindent \underline{\em Case 6:} $\mathsf{G}
=\leftarrow(\Box \wedge \mathsf{E})$. Then, $\mathsf{G}\theta =
\leftarrow(\Box \wedge (\mathsf{E}\theta))$. By
Definition~\ref{derives-in-one-step} we get $(\Box \wedge
(\mathsf{E}\theta))\stackrel{\epsilon}{\rightarrow}
\mathsf{E}\theta$. By assumption, $\mathsf{E}\theta$ has an
SLD-refutation of length $n$ using $\sigma$. Consider now the goal
$\mathsf{G}$. By Definition~\ref{derives-in-one-step}, we get that
$(\Box \wedge\mathsf{E}) \stackrel{\epsilon}{\rightarrow}
\mathsf{E}$. Since $\mathsf{E}\theta$ has an SLD-refutation of
length $n$ using $\sigma$, we get by the induction hypothesis that
$\mathsf{E}$ has an SLD-refutation using a substitution $\delta$,
where for some substitution $\gamma$ it holds $\delta\gamma
\supseteq \theta\sigma$ and $dom(\delta\gamma-\theta\sigma)$ is a
set of template variables that are introduced during this
refutation. But then, $(\Box \wedge \mathsf{E})$ has an
SLD-refutation using substitution $\delta$ which satisfies the
requirements of the lemma.

\vspace{0.2cm}\noindent \underline{\em Case 7:} $\mathsf{G}
=\leftarrow(\mathsf{E} \wedge \Box)$. Almost identical to the
previous case.

\vspace{0.2cm}\noindent \underline{\em Case 8:} $\mathsf{G}
=\leftarrow(\exists \mathsf{V}\,\mathsf{E})$. Then,
$\mathsf{G}\theta = \leftarrow( \exists \mathsf{V}\,
(\mathsf{E}\theta))$. By Definition~\ref{derives-in-one-step} we get
$( \exists \mathsf{V}\, (\mathsf{E}\theta))
\stackrel{\epsilon}{\rightarrow} \mathsf{E}\theta$. By assumption,
$\mathsf{E}\theta$ has an SLD-refutation of length $n$ using
$\sigma$. Consider now the goal $\mathsf{G}$. By
Definition~\ref{derives-in-one-step}, we get that $(\exists
\mathsf{V}\,\mathsf{E}) \stackrel{\epsilon}{\rightarrow}
\mathsf{E}$. Since $\mathsf{E}\theta$ has an SLD-refutation of
length $n$ using $\sigma$, we get by the induction hypothesis that
$\mathsf{E}$ has an SLD-refutation using a substitution $\delta$,
where for some substitution $\gamma$ it holds $\delta\gamma
\supseteq \theta\sigma$ and $dom(\delta\gamma-\theta\sigma)$ is a
set of template variables that are introduced during this
refutation. But then, $(\exists \mathsf{V}\,\mathsf{E})$ has an
SLD-refutation using substitution $\delta$ which satisfies the
requirements of the lemma.

\vspace{0.2cm}\noindent \underline{\em Case 9:} $\mathsf{G}
=\leftarrow(\mathsf{E}_1 \wedge \mathsf{E}_2)$. We may assume
without loss of generality that given the goal $\mathsf{G}\theta
=\leftarrow(\mathsf{E}_1\theta \wedge \mathsf{E}_2\theta)$, the
first step in the refutation will take place due to the
subexpression $\mathsf{E}_1\theta$. Moreover, again without loss of
generality, due to the associativity of $\wedge$, we assume that
$\mathsf{E}_1$ is not an expression that contains a top-level
$\wedge$ (ie., it is not of the form $\mathsf{A}_1 \wedge
\mathsf{A}_2$). The proof could be easily adapted to circumvent the
two assumptions just mentioned (but this would result in more
cumbersome notation). We perform a case analysis on $\mathsf{E}_1$:

\vspace{0.1cm} \noindent \underline{\em Subcase 9.1:} $\mathsf{E}_1
= (\mathsf{A}_1\,\approx \mathsf{A}_2)$. Then, we have
$((\mathsf{A}_1\approx \mathsf{A}_2)\theta \wedge
\mathsf{E}_2\theta)\stackrel{\sigma_1}{\rightarrow}(\Box\wedge
\mathsf{E}_2\theta\sigma_1)$, where $\sigma_1$ is an mgu of
$\mathsf{A}_1\theta$ and $\mathsf{A}_2\theta$. By assumption,
$(\Box\wedge \mathsf{E}_2\theta\sigma_1)$ has an SLD-refutation of
length $n$ using $\sigma'$, where $\sigma = \sigma_1 \sigma'$.
Consider now $(\mathsf{E}_1\wedge \mathsf{E}_2)$. By
Definition~\ref{derives-in-one-step}, it holds that
$(\mathsf{A}_1\approx \mathsf{A}_2)\stackrel{\delta_1}{\rightarrow}
\Box$, where $\delta_1$ is an mgu of $\mathsf{A}_1,\mathsf{A}_2$. By
Definition~\ref{derives-in-one-step} we get that
$((\mathsf{A}_1\approx \mathsf{A}_2) \wedge
\mathsf{E}_2)\stackrel{\delta_1}{\rightarrow}(\Box\wedge
\mathsf{E}_2\delta_1)$. Since $\theta\sigma_1$ is a unifier of
$\mathsf{A}_1,\mathsf{A}_2$, there exists $\theta'$ such that
$\theta\sigma_1 = \delta_1\theta'$, and since $(\Box\wedge
\mathsf{E}_2\theta\sigma_1)$ has an SLD-refutation of length $n$
using $\sigma'$, we get that $(\Box\wedge
\mathsf{E}_2\delta_1\theta') = (\Box\wedge
\mathsf{E}_2\delta_1)\theta'$ has an SLD-refutation of length $n$
using $\sigma'$. By the induction hypothesis we get that
$(\Box\wedge \mathsf{E}_2\delta_1)$ has an SLD-refutation using
$\delta'$, where $\delta'\gamma \supseteq \theta'\sigma'$ and
$dom(\delta'\gamma-\theta'\sigma')$ is a set of template variables
that are introduced during the refutation of this goal. But then,
$(\mathsf{E}_1\wedge\mathsf{E}_2)$ has an SLD-refutation using
substitution $\delta = \delta_1 \delta'$. Moreover, it holds that
$\delta\gamma = \delta_1 \delta'\gamma \supseteq
\delta_1\theta'\sigma' = \theta\sigma_1\sigma' = \theta\sigma$.

\vspace{0.1cm} \noindent \underline{\em Subcase 9.2:} $\mathsf{E}_1
= (\mathsf{Q}\,\mathsf{A}_1\cdots \mathsf{A}_r)$. Consider first the
case where $\theta(\mathsf{Q}) = \mathsf{B}$, for some basic
expression $\mathsf{B}$. Notice now that $\mathsf{B}$ can be either
a higher-order predicate variable or a finite-union of lambda
abstractions. We examine the case where $\mathsf{B}$ is a single
lambda abstraction (the other two cases are similar). Since
$\mathsf{B}$ is a lambda abstraction, we have that $\mathsf{E}_1
\theta \stackrel{\epsilon}{\rightarrow} \mathsf{E}'_1$, where
$\mathsf{E}'_1$ is the resulting expression after performing the
outer beta reduction in $\mathsf{E}_1 \theta$. By
Definition~\ref{derives-in-one-step} we have that
$(\mathsf{E}_1\wedge \mathsf{E}_2)\theta
\stackrel{\epsilon}{\rightarrow} \mathsf{E}'_1 \wedge
\mathsf{E}_2\theta$. By assumption, $\mathsf{E}'_1 \wedge
\mathsf{E}_2\theta$ has an SLD-refutation of length $n$ using
$\sigma$. Consider now $(\mathsf{E}_1\wedge \mathsf{E}_2)$. By
Definition~\ref{derives-in-one-step}, it holds that $\mathsf{E}_1
\stackrel{\{\mathsf{Q}/\mathsf{B}_t\}}{\rightarrow} \mathsf{E}''_1$,
where $\mathsf{E}''_1= \mathsf{E}_1\{\mathsf{Q}/\mathsf{B}_t\}$ and
$\mathsf{B} = \mathsf{B}_t\gamma_1$, for some substitution
$\gamma_1$, with $dom(\gamma_1) = FV(\mathsf{B}_t)$. We may assume
without loss of generality that the set $dom(\gamma_1)$ is disjoint
from $FV(\mathsf{G})$ and from $dom(\theta)\cup FV(range(\theta))$.
By Definition~\ref{derives-in-one-step}, we also get that
$\mathsf{E}''_1 \stackrel{\epsilon}{\rightarrow} \mathsf{E}'''_1$,
where $\mathsf{E}'''_1$ is the expression that results after
performing the outer beta reduction in $\mathsf{E}''_1$. Then, by
Definition~\ref{derives-in-one-step} we get that $\mathsf{E}_1\wedge
\mathsf{E}_2 \stackrel{\{\mathsf{Q}/\mathsf{B}_t\}}{\rightarrow}
\mathsf{E}''_1 \wedge \mathsf{E}_2\{\mathsf{Q}/\mathsf{B}_t\}$ and
$\mathsf{E}''_1 \wedge \mathsf{E}_2\{\mathsf{Q}/\mathsf{B}_t\}
\stackrel{\epsilon}{\rightarrow} \mathsf{E}'''_1 \wedge
\mathsf{E}_2\{\mathsf{Q}/\mathsf{B}_t\}$. Notice now that
$(\mathsf{E}'''_1\wedge
\mathsf{E}_2\{\mathsf{Q}/\mathsf{B}_t\})\theta\gamma_1 =
\mathsf{E}'_1\wedge \mathsf{E}_2\theta$, and since
$\mathsf{E}'_1\wedge \mathsf{E}_2\theta$ has an SLD-refutation of
length $n$ using $\sigma$, we get by the induction hypothesis that
$(\mathsf{E}'''_1\wedge \mathsf{E}_2\{\mathsf{Q}/\mathsf{B}_t\})$
has an SLD-refutation using $\delta'$, where for some substitution
$\gamma$ it holds $\delta'\gamma \supseteq \theta\gamma_1\sigma$ and
$dom(\delta'\gamma-\theta\gamma_1\sigma)$ is a set of template
variables that are introduced during this SLD-refutation. But then,
$\mathsf{E}_1\wedge \mathsf{E}_2$ has an SLD-refutation using
substitution $\delta = \{\mathsf{Q}/\mathsf{B}_t\}\delta'$.
Moreover, it holds that $\delta\gamma =
\{\mathsf{Q}/\mathsf{B}_t\}\delta'\gamma \supseteq
\{\mathsf{Q}/\mathsf{B}_t\}\theta\gamma_1\sigma \supseteq
\theta\sigma$ and $dom(\delta\gamma-\theta\sigma)$ is a set of
template variables that are introduced during the refutation of
$\mathsf{G}$.

Consider now the case where $\theta(\mathsf{Q})$ is undefined, ie.,
there does not exist a binding for $\mathsf{Q}$ in $\theta$. Then,
we have that $\mathsf{E}_1 \theta
\stackrel{\{\mathsf{Q}/\mathsf{B}_t\}}{\rightarrow} \mathsf{E}'_1$,
where
$\mathsf{E}'_1=\mathsf{B}_t\,(\mathsf{A}_1\theta\{\mathsf{Q}/\mathsf{B}_t\})\cdots
(\mathsf{A}_r\theta\{\mathsf{Q}/\mathsf{B}_t\})$. We may assume
without loss of generality that the set $FV(\mathsf{B}_t)$ is
disjoint from $FV(\mathsf{G})$ and from $dom(\theta)\cup
FV(range(\theta))$. By Definition~\ref{derives-in-one-step} we have
that $(\mathsf{E}_1\wedge \mathsf{E}_2)\theta
\stackrel{\{\mathsf{Q}/\mathsf{B}_t\}}{\rightarrow} \mathsf{E}'_1
\wedge (\mathsf{E}_2\theta\{\mathsf{Q}/\mathsf{B}_t\})$. By
assumption, $\mathsf{E}'_1 \wedge
(\mathsf{E}_2\theta\{\mathsf{Q}/\mathsf{B}_t\})$ has an
SLD-refutation of length $n$ using $\sigma'$, where
$\sigma=\{\mathsf{Q}/\mathsf{B}_t\}\sigma'$. Consider now
$(\mathsf{E}_1\wedge \mathsf{E}_2)$. By
Definition~\ref{derives-in-one-step}, it holds that $\mathsf{E}_1
\stackrel{\{\mathsf{Q}/\mathsf{B}_t\}}{\rightarrow} \mathsf{E}''_1$,
where $\mathsf{E}''_1= \mathsf{E}_1\{\mathsf{Q}/\mathsf{B}_t\}$. By
Definition~\ref{derives-in-one-step} we get that $\mathsf{E}_1\wedge
\mathsf{E}_2 \stackrel{\{\mathsf{Q}/\mathsf{B}_t\}}{\rightarrow}
\mathsf{E}''_1 \wedge \mathsf{E}_2\{\mathsf{Q}/\mathsf{B}_t\}$.
Notice now that $(\mathsf{E}''_1\wedge
\mathsf{E}_2\{\mathsf{Q}/\mathsf{B}_t\})\theta\{\mathsf{Q}/\mathsf{B}_t\}
= \mathsf{E}'_1\wedge
\mathsf{E}_2\theta\{\mathsf{Q}/\mathsf{B}_t\}$, and since
$\mathsf{E}'_1\wedge \mathsf{E}_2\theta\{\mathsf{Q}/\mathsf{B}_t\}$
has an SLD-refutation of length $n$ using $\sigma'$, we get by the
induction hypothesis that $(\mathsf{E}''_1\wedge
\mathsf{E}_2\{\mathsf{Q}/\mathsf{B}_t\})$ has an SLD-refutation
using $\delta'$, where for some substitution $\gamma$ it holds that
$\delta'\gamma \supseteq \theta\{\mathsf{Q}/\mathsf{B}_t\}\sigma'$,
and $dom(\delta'\gamma-\theta\{\mathsf{Q}/\mathsf{B}_t\}\sigma')$ is
a set of template variables that are introduced during this
SLD-refutation; notice that these template variables can be chosen
to be different than the variables in $FV(\mathsf{B}_t)$. Then,
$\mathsf{E}_1\wedge \mathsf{E}_2$ has an SLD-refutation using
substitution $\delta = \{\mathsf{Q}/\mathsf{B}_t\}\delta'$.
Moreover, it holds that $\delta\gamma =
\{\mathsf{Q}/\mathsf{B}_t\}\delta'\gamma \supseteq
\{\mathsf{Q}/\mathsf{B}_t\}\theta\{\mathsf{Q}/\mathsf{B}_t\}\sigma'=
\theta\{\mathsf{Q}/\mathsf{B}_t\}\sigma' = \theta\sigma$ and
$dom(\delta\gamma-\theta\sigma)$ is a set of template variables that
are introduced during the refutation of $\mathsf{G}$.

\vspace{0.1cm} \noindent \underline{\em Subcase 9.3:} $\mathsf{E}_1
=\leftarrow((\mathsf{A}'\bigvee_{\pi}\mathsf{A}'')\,
\mathsf{A}_1\cdots \mathsf{A}_r)$. Then, $\mathsf{E}_1\theta =
(\mathsf{A}'\theta\bigvee_{\pi}\mathsf{A}''\theta)\,
(\mathsf{A}_1\theta)\cdots (\mathsf{A}_r\theta)$. By
Definition~\ref{derives-in-one-step} we get that
$(\mathsf{A}'\theta\bigvee_{\pi}\mathsf{A}''\theta)\,
(\mathsf{A}_1\theta)\cdots (\mathsf{A}_r\theta)
\stackrel{\epsilon}{\rightarrow} (\mathsf{A}'\theta)\,
(\mathsf{A}_1\theta)\cdots (\mathsf{A}_r\theta)$ (and symmetrically
for $\mathsf{A}''$). By Definition~\ref{derives-in-one-step} we have
$(\mathsf{E}_1\theta \wedge \mathsf{E}_2\theta)
\stackrel{\epsilon}{\rightarrow}((\mathsf{A}'\theta)\,
(\mathsf{A}_1\theta)\cdots (\mathsf{A}_r\theta))\wedge
\mathsf{E}_2\theta$ and $(\mathsf{E}_1\theta \wedge
\mathsf{E}_2\theta) \stackrel{\epsilon}{\rightarrow}
((\mathsf{A}''\theta)\, (\mathsf{A}_1\theta)\cdots
(\mathsf{A}_r\theta))\wedge \mathsf{E}_2\theta$. By assumption,
either $((\mathsf{A}'\theta)\, (\mathsf{A}_1\theta)\cdots
(\mathsf{A}_r\theta)) \wedge \mathsf{E}_2\theta$ or
$((\mathsf{A}''\theta)\, (\mathsf{A}_1\theta)\cdots
(\mathsf{A}_r\theta)) \wedge \mathsf{E}_2\theta$ has an
SLD-refutation of length $n$ using $\sigma$. Assume, without loss of
generality, that $((\mathsf{A}'\theta)\, (\mathsf{A}_1\theta)\cdots
(\mathsf{A}_r\theta)) \wedge \mathsf{E}_2\theta$ has an
SLD-refutation of length $n$ using $\sigma$. Notice now that by
Definition~\ref{derives-in-one-step}, we have that
$(\mathsf{A}'\bigvee_{\pi}\mathsf{A}'')\, \mathsf{A}_1\cdots
\mathsf{A}_r \stackrel{\epsilon}{\rightarrow} \mathsf{A}'\,
\mathsf{A}_1\cdots \mathsf{A}_r$. Moreover, notice that
$((\mathsf{A}'\, \mathsf{A}_1\cdots \mathsf{A}_r) \wedge
\mathsf{E}_2)\theta =((\mathsf{A}'\theta)\,
(\mathsf{A}_1\theta)\cdots (\mathsf{A}_r\theta)) \wedge
\mathsf{E}_2\theta$, and since the latter expression has an
SLD-refutation of length $n$ using $\sigma$, we get by the induction
hypothesis that $((\mathsf{A}'\, \mathsf{A}_1\cdots \mathsf{A}_r)
\wedge \mathsf{E}_2)$ has an SLD-refutation using a substitution
$\delta$, where for some substitution $\gamma$ it holds
$\delta\gamma \supseteq \theta\sigma$ and
$dom(\delta\gamma-\theta\sigma)$ is a set of template variables that
are introduced during this refutation. But then,
$((\mathsf{A}'\bigvee_{\pi}\mathsf{A}'')\, \mathsf{A}_1\cdots
\mathsf{A}_r)\wedge \mathsf{E}_2$ has an SLD-refutation using
substitution $\delta$ which satisfies the requirements of the lemma.

\vspace{0.1cm} \noindent \underline{\em Subcase 9.4:} $\mathsf{E}_1$
has any other form except for the ones examined in the previous
three subcases. Then, it can be verified that in all these subcases
it holds that $\mathsf{E}_1 \theta \stackrel{\epsilon}{\rightarrow}
\mathsf{E}'_1$, for some $\mathsf{E}'_1$. By
Definition~\ref{derives-in-one-step} we have that
$(\mathsf{E}_1\wedge \mathsf{E}_2)\theta
\stackrel{\epsilon}{\rightarrow} \mathsf{E}'_1 \wedge
(\mathsf{E}_2\theta)$. By assumption, $\mathsf{E}'_1 \wedge
(\mathsf{E}_2\theta)$ has an SLD-refutation of length $n$ using
$\sigma$. Consider now $(\mathsf{E}_1\wedge \mathsf{E}_2)$. By
Definition~\ref{derives-in-one-step} and by examination of all the
possible cases for $\mathsf{E}_1$ it can be seen that $\mathsf{E}_1
\stackrel{\epsilon}{\rightarrow} \mathsf{E}''_1$, where
$\mathsf{E}'_1 = \mathsf{E}''_1\theta$. By
Definition~\ref{derives-in-one-step} we get that $\mathsf{E}_1\wedge
\mathsf{E}_2 \stackrel{\epsilon}{\rightarrow} \mathsf{E}''_1 \wedge
\mathsf{E}_2$. Notice now that $(\mathsf{E}''_1\wedge
\mathsf{E}_2)\theta = \mathsf{E}'_1\wedge \mathsf{E}_2\theta$, and
since $\mathsf{E}'_1\wedge \mathsf{E}_2\theta$ has an SLD-refutation
of length $n$ using $\sigma$, we get by the induction hypothesis
that $(\mathsf{E}''_1\wedge \mathsf{E}_2)$ has an SLD-refutation
using $\delta$, where for some substitution $\gamma$ it holds that
$\delta\gamma \supseteq \theta\sigma$, and
$dom(\delta\gamma-\theta\sigma)$ is a set of template variables that
are introduced during this SLD-refutation. Then, $\mathsf{E}_1\wedge
\mathsf{E}_2$ has an SLD-refutation using substitution $\delta$
which satisfies the requirements of the lemma.\qed
\end{proof}

\section{Proof of Lemma~\ref{lemma-completeness-bottom-interp}}\label{appendix-bottom}

{\bf Lemma~\ref{lemma-completeness-bottom-interp}} Let $\mathsf{P}$
be a program and $\mathsf{G}=\leftarrow \mathsf{A}$ be a goal such
that $\lsem \mathsf{A} \rsem_s(\bot_{\mathcal{I}_\mathsf{P}}) = 1$
for all Herbrand states $s$. Then, there exists an SLD-refutation
for $\mathsf{P}\cup\{\mathsf{G}\}$ with computed answer equal to the
identity substitution.
\begin{proof}
We establish a stronger statement which has the statement of the
lemma as a special case. Let us call a substitution $\theta$ {\em
closed}  if every expression in $range(\theta)$ is closed. We
demonstrate that for every closed basic substitution $\theta$, if
$\lsem \mathsf{A} \theta \rsem_s(\bot_{\mathcal{I}_\mathsf{P}}) = 1$
for all Herbrand states $s$, then there exists an SLD-refutation for
$\mathsf{P}\cup\{\mathsf{A}\theta\}$ with computed answer equal to
the identity substitution. The statement of the lemma is then a
direct consequence for $\theta = \epsilon$.

We start by noting that $\mathsf{A}$ is always of the form
$(\mathsf{A}_0\ \mathsf{A}_1\cdots\mathsf{A}_n)$, $n\geq 0$ (if
$n=0$ then $\mathsf{A}_0$ is of type $o$). We perform induction on
the type $\rho_1 \rightarrow \cdots \rho_n \rightarrow o$ of
$\mathsf{A}_0$.

\vspace{0.3cm}\noindent\underline{\em Outer Induction Basis:} The
outer induction basis is for $n=0$, ie., for $type(\mathsf{A}_0) =
o$, and in order to establish it we need to perform an inner
structural induction on $\mathsf{A}_0$.

\vspace{0.2cm}\noindent{\em Inner Induction Basis:} For the inner
induction basis we need to examine the cases where $\mathsf{A}_0$ is
$\mathsf{0}$, $\mathsf{1}$, $(\mathsf{E}_1 \approx \mathsf{E}_2)$
and $\mathsf{Q}$, where $type(\mathsf{Q})=o$. The first case is not
applicable since $\lsem \mathsf{0}\theta
\rsem_s(\bot_{\mathcal{I}_\mathsf{P}}) = 0$. The second case is
immediate. We examine the latter two cases:

\vspace{0.2cm}\noindent {\em Case 1:} $\mathsf{A}_0 = (\mathsf{E}_1
\approx \mathsf{E}_2)$. Since for all $s$ it holds that $\lsem
(\mathsf{E}_1 \approx \mathsf{E}_2)\theta
\rsem_s(\bot_{\mathcal{I}_\mathsf{P}}) = 1$, we get that for all
$s$, $\lsem \mathsf{E}_1\theta \rsem_s
(\bot_{\mathcal{I}_\mathsf{P}}) = \lsem \mathsf{E}_2\theta
\rsem_s(\bot_{\mathcal{I}_\mathsf{P}})$. By the fact that
$\bot_{\mathcal{I}_\mathsf{P}}$ is a Herbrand interpretation, we
conclude that $\mathsf{E}_1\theta$ and $\mathsf{E}_2\theta$ must be
identical expressions of type $\iota$, and therefore they are
unifiable using the identity substitution.

\vspace{0.2cm}\noindent {\em Case 2:} $\mathsf{A}_0 = \mathsf{Q}$,
with $type(\mathsf{Q})=o$. If $\theta(\mathsf{Q})=0$ then it can not
be the case that $\lsem \mathsf{A}_0
\rsem_s(\bot_{\mathcal{I}_\mathsf{P}}) = 1$, and therefore this case
is not applicable. If $\theta(\mathsf{Q})=1$, the result is trivial.
If on the other hand $\mathsf{Q}\not\in dom(\mathsf{\theta})$, then
this case is not applicable since it is not possible to have $\lsem
\mathsf{A}_0\theta \rsem _s(\bot_{\mathcal{I}_\mathsf{P}}) = 1$, for
all $s$ (eg. choose $s$ such that $s(\mathsf{Q}) = 0$).

\vspace{0.2cm}\noindent{\em Inner Induction Step:} We distinguish
the following cases:

\vspace{0.2cm}\noindent {\em Case 1:} $\mathsf{A}_0 = (\exists
\mathsf{Q}\,\mathsf{E})$. We can assume without loss of generality
that $\mathsf{Q} \not\in dom(\theta)$. Since for all $s$ it holds
that $\lsem (\exists \mathsf{Q}\, \mathsf{E})\theta
\rsem_s(\bot_{\mathcal{I}_\mathsf{P}}) = 1$, it follows that $\lsem
\mathsf{E}\theta
\rsem_{s[b/\mathsf{Q}]}(\bot_{\mathcal{I}_\mathsf{P}}) = 1$ for some
$b \in {\cal F}_{U_{\mathsf{P}}}(type(\mathsf{Q}))$. Let $\theta' =
\{\mathsf{Q}/\mathsf{B}\}$ where $\mathsf{B}$ is a closed basic
expression such that $\lsem \mathsf{B}
\rsem(\bot_{\mathcal{I}_\mathsf{P}}) = b$ (the existence of such an
expression $\mathsf{B}$ is ensured by
Lemma~\ref{for_every_b_exists_B}). Then it is easy to see that
$\lsem \mathsf{E}\theta\theta'
\rsem_{s}(\bot_{\mathcal{I}_\mathsf{P}}) = 1$ for all states $s$. By
the induction hypothesis there exists an SLD-refutation for
$\mathsf{P}\cup\{\mathsf{E}\theta\theta'\}$ using some substitution
$\sigma$ and with computed answer equal to the identity
substitution. Using Lemma~\ref{pre-lifting-lemma}, it follows that
there exists an SLD-refutation of
$\mathsf{P}\cup\{\mathsf{E}\theta\}$ using substitution $\delta$
where for some substitution $\gamma$ it holds that $\delta\gamma
\supseteq \{\mathsf{Q}/\mathsf{B}\}\sigma$; moreover,
$dom(\delta\gamma-\{\mathsf{Q}/\mathsf{B}\}\sigma)$ is a set of
template variables that are introduced during the refutation of
$\mathsf{P}\cup\{\mathsf{E}\theta\}$. Since the restriction of
$\sigma$ to the free variables of $\mathsf{E}\theta\theta'$ is the
identity substitution, it follows that the restriction of $\delta$
to the free variables of $\mathsf{E}\theta$ will either be empty or
it will only contain the binding $\mathsf{Q}/\mathsf{B}$. We
conclude that there exists a refutation of $\mathsf{P}\cup\{(\exists
\mathsf{Q}\,\mathsf{E})\theta\}$ using substitution $\epsilon\delta
= \delta$. The computed answer is the identity substitution since
$\mathsf{Q}$ is not a free variable of $(\exists
\mathsf{Q}\,\mathsf{E})\theta$.

\vspace{0.2cm}\noindent {\em Case 2:} $\mathsf{A}_0 = (\mathsf{E}_1
\wedge \mathsf{E}_2)$. By assumption, $\lsem (\mathsf{E}_1 \wedge
\mathsf{E}_2)\theta \rsem_s(\bot_{\mathcal{I}_\mathsf{P}}) = 1$, for
all $s$. Then, it holds $\lsem \mathsf{E}_1\theta
\rsem_s(\bot_{\mathcal{I}_\mathsf{P}}) = 1$ and $\lsem
\mathsf{E}_2\theta \rsem_s(\bot_{\mathcal{I}_\mathsf{P}}) = 1$. By
the induction hypothesis there exist SLD-refutations for
$\mathsf{P}\cup\{\leftarrow \mathsf{E}_1\theta\}$ and
$\mathsf{P}\cup\{\leftarrow \mathsf{E}_2\theta\}$ with computed
answers equal to the identity substitution. Let $\theta_1$ and
$\theta_2$ be the compositions of the substitutions used for the
refutations of $\mathsf{P}\cup\{\leftarrow \mathsf{E}_1\theta\}$ and
$\mathsf{P}\cup\{\leftarrow \mathsf{E}_2\theta \}$ respectively.
Now, since the computed answer of the refutation for
$\mathsf{P}\cup\{\mathsf{E}_1\theta\}$ is the identity, this implies
that the free variables of $\mathsf{E}_2\theta$ that also appear
free in $\mathsf{E}_1\theta$ do not belong to $dom(\theta_1)$.
Moreover, the rest of the free variables of $\mathsf{E}_2\theta$ do
not belong to $dom(\theta_1)$, because the variables of $\theta_1$
have been obtained by using resolution steps that only involve
``fresh'' variables. In conclusion, the restriction of $\theta_1$ to
the free variables of $(\mathsf{E}_1 \wedge \mathsf{E}_2)\theta$ is
the identity substitution (and similarly for $\theta_2$). These
observations imply that $\mathsf{E}_2\theta\theta_1 =
\mathsf{E}_2\theta$. But then, we can construct a refutation for
$\mathsf{P}\cup\{\leftarrow (\mathsf{E}_1\theta \wedge
\mathsf{E}_2\theta) \}$ by first deriving $\Box$ from
$\mathsf{E}_1\theta$ and then deriving $\Box$ from
$\mathsf{E}_2\theta\theta_1=\mathsf{E}_2\theta$. The substitution
used for the refutation of $\mathsf{P}\cup\{\leftarrow (\mathsf{E}_1
\wedge \mathsf{E}_2)\theta \}$ is $\theta_1\theta_2$ and the
computed answer is equal to the restriction of $\theta_1\theta_2$ to
the free variables of $(\mathsf{E}_1 \wedge \mathsf{E}_2)\theta$,
which gives the identity substitution.

\vspace{0.2cm}\noindent {\em Case 3:} $\mathsf{A}_0 = (\mathsf{E}_1
\vee \mathsf{E}_2)$. By assumption, $\lsem (\mathsf{E}_1 \vee
\mathsf{E}_2)\theta \rsem_s(\bot_{\mathcal{I}_\mathsf{P}}) = 1$, for
all $s$. Then, it either holds that $\lsem \mathsf{E}_1\theta
\rsem_s(\bot_{\mathcal{I}_\mathsf{P}}) = 1$ or $\lsem
\mathsf{E}_2\theta \rsem_s(\bot_{\mathcal{I}_\mathsf{P}}) = 1$.
Without loss of generality, assume that $\lsem \mathsf{E}_1\theta
\rsem_s(\bot_{\mathcal{I}_\mathsf{P}}) = 1$. By the induction
hypothesis there exists an SLD-refutation for
$\mathsf{P}\cup\{\leftarrow \mathsf{E}_1\theta\}$ with computed
answer equal to the identity substitution. But then
$\mathsf{P}\cup\{\leftarrow (\mathsf{E}_1\theta \vee
\mathsf{E}_2\theta)\}$ has an SLD-refutation whose first step leads
to $\leftarrow (\mathsf{E}_1\theta)$ using $\epsilon$ and then
proceeds according to the SLD-refutation of $\leftarrow
(\mathsf{E}_1\theta)$. The computed answer of this refutation is
obviously the identity substitution.

\vspace{0.3cm}\noindent\underline{\em Outer Induction Step:} Assume
the lemma holds when $\mathsf{A}_0$ has type $\rho_1 \rightarrow
\cdots \rho_{n-1} \rightarrow o$, $n\geq 1$. We establish the lemma
for the case where $\mathsf{A}_0$ has type $\pi = \rho_1 \rightarrow
\cdots \rho_{n} \rightarrow o$. We distinguish the following cases:

\vspace{0.2cm}\noindent {\em Case 1:} $\mathsf{A}_0 = \mathsf{p}$
(ie., $\mathsf{A} = \mathsf{p}\,\mathsf{A}_1\cdots \mathsf{A}_n$).
This case is not applicable since $\bot_{\mathcal{I}_\mathsf{P}}
(\mathsf{p}) = \perp_{\pi}$ and therefore $\lsem \mathsf{A} \rsem
_s(\bot_{\mathcal{I}_\mathsf{P}}) = 0$, for all $s$.

\vspace{0.2cm}\noindent {\em Case 2:} $\mathsf{A}_0 =\mathsf{Q}$
(ie., $\mathsf{A}= \mathsf{Q}\,\mathsf{A}_1\cdots \mathsf{A}_n$). If
$\mathsf{Q}\not\in dom(\theta)$ then this case is not applicable
since it is not possible to have $\lsem \mathsf{A} \rsem
_s(\bot_{\mathcal{I}_\mathsf{P}}) = 1$, for all $s$ (eg. take
$s(\mathsf{Q}) = \perp_{\pi}$). If on the other hand $\mathsf{Q} \in
dom(\theta)$, then $\theta(\mathsf{Q})$ is a basic expression of
type $\pi$, ie., it is a non-empty finite union of lambda
abstractions. We demonstrate the case where $\theta(\mathsf{Q})$ is
a single lambda abstraction; the more general case is similar and
omitted. Assume therefore that $\theta(\mathsf{Q}) = \lambda
\mathsf{V}.\mathsf{E}$. Then, since $\lsem (\lambda
\mathsf{V}.\mathsf{E})\, (\mathsf{A}_1\theta) \cdots
(\mathsf{A}_n\theta) \rsem_s(\bot_{\mathcal{I}_\mathsf{P}}) = 1$, by
Lemma~\ref{beta-reduction-lemma} we get that $\lsem
(\mathsf{E}\{\mathsf{V}/\mathsf{A}_1\theta\})\, (\mathsf{A}_2\theta)
\cdots (\mathsf{A}_n\theta) \rsem_s(\bot_{\mathcal{I}_\mathsf{P}}) =
1$. By assumption, $\theta$ is a closed substitution and therefore
the only free variable that appears in $\mathsf{E}$ is $\mathsf{V}$.
Therefore, $\lsem ((\mathsf{E}\{\mathsf{V}/\mathsf{A}_1\})\,
\mathsf{A}_2 \cdots \mathsf{A}_n)\theta
\rsem_s(\bot_{\mathcal{I}_\mathsf{P}}) = 1$. By the outer induction
hypothesis we get that $\mathsf{P}\cup\{\leftarrow
((\mathsf{E}\{\mathsf{V}/\mathsf{A}_1\})\, \mathsf{A}_2 \cdots
\mathsf{A}_n)\theta \}$ has an SLD-refutation  using substitution
$\delta$, with computed answer equal to the identity substitution.
By the definition of SLD-resolution we get that
$\mathsf{P}\cup\{\leftarrow (\lambda \mathsf{V}.\mathsf{E})\,
(\mathsf{A}_1\theta) \cdots (\mathsf{A}_n\theta))\}$ has an
SLD-refutation using the substitution $\epsilon\delta = \delta$; the
computed answer of this refutation is the restriction of $\delta$ to
the free variables of $((\lambda \mathsf{V}.\mathsf{E})\,
\mathsf{A}_1 \cdots \mathsf{A}_n)\theta$ which (by our previous
discussion) gives the identity substitution.

\vspace{0.2cm}\noindent {\em Case 3:} $\mathsf{A}_0=\lambda
\mathsf{V}.\mathsf{E}$ (ie., $\mathsf{A}=( \lambda
\mathsf{V}.\mathsf{E})\, \mathsf{A}_1 \cdots \mathsf{A}_n$. We can
assume without loss of generality that $\mathsf{V} \not\in
dom(\theta)\cup FV(range(\theta))$. Then, since $\lsem (\lambda
\mathsf{V}.\mathsf{E}\theta)\, (\mathsf{A}_1\theta) \cdots
(\mathsf{A}_n\theta) \rsem_s(\bot_{\mathcal{I}_\mathsf{P}}) = 1$, by
Lemma~\ref{beta-reduction-lemma} we get that $\lsem
(\mathsf{E}\theta\{\mathsf{V}/\mathsf{A}_1\theta\})\,
(\mathsf{A}_2\theta) \cdots (\mathsf{A}_n\theta)
\rsem_s(\bot_{\mathcal{I}_\mathsf{P}}) = 1$. By the outer induction
hypothesis $\mathsf{P}\cup\{\leftarrow
(\mathsf{E}\theta\{\mathsf{V}/\mathsf{A}_1\theta\})\,
(\mathsf{A}_2\theta) \cdots (\mathsf{A}_n\theta) \}$ has an
SLD-refutation  using substitution $\delta$, with computed answer
equal to the identity substitution. By the definition of
SLD-resolution we get that $\mathsf{P}\cup\{\leftarrow ((\lambda
\mathsf{V}.\mathsf{E}\theta)\, (\mathsf{A}_1\theta) \cdots
(\mathsf{A}_n\theta))\}$ has an SLD-refutation using the
substitution $\epsilon\delta = \delta$; the computed answer of this
refutation is the restriction of $\delta$ to the free variables of
$((\lambda \mathsf{V}.\mathsf{E})\, \mathsf{A}_1 \cdots
\mathsf{A}_n)\theta$ which (by our previous discussion) gives the
identity substitution.

\vspace{0.2cm}\noindent {\em Case 4:}  $\mathsf{A}_0 = (\mathsf{E}'
\bigvee_{\pi} \mathsf{E}'')$ (ie., $\mathsf{A} = (\mathsf{E}'
\bigvee_{\pi} \mathsf{E}'')\, \mathsf{A}_1\cdots\mathsf{A}_n$),
where $\pi \neq o$. The proof for this case follows easily using the
outer induction hypothesis.

\vspace{0.2cm}\noindent {\em Case 5:} $\mathsf{A}_0 = (\mathsf{E}'
\bigwedge_{\pi} \mathsf{E}'')$ (ie., $\mathsf{A} = (\mathsf{E}'
\bigwedge_{\pi} \mathsf{E}'')\, \mathsf{A}_1\cdots\mathsf{A}_n$),
where $\pi \neq o$. The proof for this case follows easily using the
outer induction hypothesis.\qed
\end{proof}

\end{document}